\documentclass[11pt]{article}
\usepackage{graphicx} 
\usepackage{setspace} 
\usepackage[margin=1in]{geometry} 

\usepackage[utf8]{inputenc}
\usepackage[parfill]{parskip}    
\usepackage{graphicx}
\usepackage{amsmath}
\usepackage{amssymb}
\usepackage{amsthm}
\usepackage{soul}

\usepackage{algorithm}
\usepackage[noend]{algpseudocode}
\usepackage{todonotes}
\usepackage{tabularray}
\usepackage{tabularx}
\MakeRobust{\Call}
\usepackage{makecell}

\usepackage{hyperref}
\usepackage{xcolor}
\hypersetup{
    colorlinks,
    linkcolor={red!50!black},
    citecolor={blue!50!black},
    urlcolor={blue!80!black}
}

\usepackage{soul}

\usepackage{wrapfig}

\usepackage[makeroom]{cancel}

\usepackage{amsthm}
\newtheorem{theorem}{Theorem}[section]
\theoremstyle{definition}

\newtheorem{lemma}[theorem]{Lemma}
\newtheorem{defn}[theorem]{Definition}

\newtheorem{claim}[theorem]{Claim}
\newtheorem{observation}[theorem]{Observation}
\newtheorem{corollary}[theorem]{Corollary}

\newtheorem{innercustomthm}{Theorem}
\newenvironment{customthm}[1]
  {\renewcommand\theinnercustomthm{#1}\innercustomthm}
  {\endinnercustomthm}

\usepackage{mathtools}
\DeclarePairedDelimiter\ceil{\lceil}{\rceil}
\DeclarePairedDelimiter\floor{\lfloor}{\rfloor}

\algdef{SE}[SUBALG]{Indent}{EndIndent}{}{\algorithmicend\ }%
\algtext*{Indent}
\algtext*{EndIndent}

\def\poly{\mbox{\text{poly}}}

\def\Exp{\hbox{\rm I\kern-2pt E}}
\def\Prob{\hbox{\rm I\kern-2pt P}}
\def\Naturals{\hbox{\rm I\kern-2pt I\kern-3.1pt N}}

\def\inline#1:{\par\vskip 7pt\noindent{\bf #1:}\hskip 10pt}
\def\dnsinline#1:{\par\vskip -7pt\noindent{\bf #1:}\hskip 10pt}
\def\dnsitem{\vspace{-6pt}\item}

\long\def\commabs #1\commabsend{}
\long\def\CommentedStart #1\CommentedEnd{}

\def\BI{\mbox{\sf BI}}
\def\tn{{\tilde n}}
\def\Work{\mbox{\sf Work}}

\def\BYZ{\mathcal{B}}
\def\cC{\mathcal{C}}

\def\cE{\mathcal{E}}
\def\Random{\mathcal{R}}
\def\execution{\mathcal{E}}
\def\event{\mathcal{EV}}
\def\cH{\mathcal{H}}

\def\Wmax{\mbox{\sf W}_{max}}
\def\Known{\mathcal{K}}
\def\Unknown{\mathcal{U}}

\def\KTA{\mbox{KTA}}
\def\NKTA{\mbox{NKTA}}
\def\TwoCK{2-common-knowledge}
\def\INPUT{\mathcal{X}}
\def\INDEX{\mathcal{I}}

\def\hINDEX{{\hat{\mathcal{I}}}}
\def\cF{\mathcal{F}}
\def\cHt{\mathcal{H}_{\mathcal{F}}}

\def\hn{{\hat n}}

\def\hM{{\hat M}}
\def\bitval{b}
\def\Det{\mathbf{Det}}
\def\Harsh{\mathbf{Harsh}}

\def\Benign{\mathbf{Benign}}
\def\alln{n}
\def\IDp{{\textsf{Download}}}
\def\ORp{{\textsf{Disjunction}}}
\def\XORp{{\textsf{Parity}}}
\def\newexp{{\text{Large Set Expander}}}
\def\num{\nu}
\def\size{\sigma}
\def\Time{\mathcal{T}}
\def\Query{\mathcal{Q}}
\def\TotQuery{\mathcal{TQ}}
\def\Message{\mathcal{M}}
\def\Machine{M}
\def\CORE{\mbox{\sf CORE}}

\def\CloudQuery{\mbox{\tt Cloud\_Query}}
\def\CloudRG{\mbox{\tt Cloud\_RG}}

\def\Aor{\mathsf{A}_\mathsf{or}}
\def\Advers{\mathtt{Adv}}
\def\Algor{\mathtt{ALG}}
\def\res{res}
\def\resM{res_M}
\def\rand{rand}
\def\nul{null}
\def\semiparity{semiparity}
\def\parity{parity}

\def\Spread{\mbox{\tt Spread}}
\def\RandomizedDisjunctionAlgo{\mbox{\tt Randomized\_Disjunction}}
\def\CollectRequests{\mbox{\tt Collect\_Requests}}
\def\CommitteeWork{\mbox{\tt Committee\_Work}}
\def\Gossip{\mbox{\tt Gossip}}
\def\KTAList{\mbox{\tt KTA\_List}}
\def\WeakResolve{\mbox{\tt Weak\_Resolve}}
\def\FastWeakResolve{\mbox{\tt Fast\_Weak\_Resolve}}
\def\WeakXORResolve{\mbox{\tt Weak\_Parity\_Resolve}}
\def\WeakParityResolve{\mbox{\tt Weak\_Parity\_Resolve}}

\def\InputConvergecast{\mbox{\tt Input\_Convergecast}}
\def\InputConvergecast{\mbox{\tt Convergecast}}

\def\ConvergecastForest{\mbox{\tt Convergecast\_Forest}}
\def\ElectPublic{\mbox{\tt Elect\_Public}}
\def\ElectPrivate{\mbox{\tt Elect\_Private}}
\def\DownloadParallel{\mbox{\tt Parallel\_Download}}

\def\LSEDisjunctOne{\mbox{\tt LSE\_Disjunct\_1}}
\def\LSEDisjunctTwo{\mbox{\tt LSE\_Disjunct\_2}}
\def\GLSEExplicitDisjunct{\mbox{\tt GLSE\_Explicit\_Disjunct}}
\def\DownloadBlacklist{\mbox{\tt Blacklist\_Download}}
\def\DownloadGossip{\mbox{\tt Gossip\_Download}}
\def\LinearDownload{\mbox{\tt Linear\_Download}}
\def\FastLinearDownload{\mbox{\tt Fast\_Linear\_Download}}
\def\MajorizingDownload{\mbox{\tt Majorizing\_Download}}
\def\ConvergeParity{\mbox{\tt Converge\_Parity}}
\def\MajorizingParity{\mbox{\tt Majorizing\_Parity}}

\def\PARAM{Z} 

\def\CCA{\mbox{\sf CCA}}

\def\Density{\delta}
\def\gamin{\min(\Density,1-\Density)}
\def\ModifiedDensity{\mbox{\boldmath$\Density$}}
\def\InverseDensity{\frac{1}{\ModifiedDensity}}

\def\byzfrac{\beta}
\def\goodfrac{\gamma}

\ifdefined\ShowComment
\def\jeffin#1{{\color{orange}{\bf Jeffin:} #1}}

\newcommand{\davidTODO}[1]{\todo[author=David, inline, color=red!40]{#1}}
\def\david#1{{\color{red}{\bf David:} #1}}
\def\john#1{{\color{blue}{\bf John:} #1}}

\def\sri#1{{\color{purple}{\bf Sri:} #1}}
\def\aish#1{{\color{magenta}{\bf Aishwarya:} #1}}

\else
\newcommand{\jeffin}[1]{{}}
\newcommand{\davidTODO}[1]{{}}
\def\david#1{{}}
\def\john#1{{}}
\newcommand{\aish}[1]{}
\def\sri#1{{}}
\def\aish#1{{}}
\fi

\newcommand{\remove}[1]{}

\usepackage{paralist}
\setdefaultleftmargin{1em}{1em}{}{}{}{}

\title{Byzantine Resilient Computing with the Cloud}
\date{ }
\author{
John Augustine$^{1}$, 
Jeffin Biju$^{1}$,
Shachar Meir$^{2}$,
David Peleg$^{2}$, \\
Srikkanth Ramachandran$^{1}$,  
Aishwarya Thiruvengadam$^{1}$ \\
        \small $^{1}$Indian Institute of Technology Madras \\
        \small $^{2}$Weizmann Institute of Science
}

\begin{document}

\maketitle

\john{Candidate title: 
\\
Byzantine Resilient Distributed Computing Via Cloud Assistance, 
\\
BRICCS: Byzantine Resilient Congested Clique With Cloud Support, 
\\
BRICCS: Byzantine Resilient Computing with Cloud Support. 
\\
The Patt-Shamir paper \cite{AGP21-sss} uses ``Computing with the Cloud (CWC)" as the model name. If we go with that, then perhaps 
\\
Byzantine Resilient Computing with the Cloud.
\\
BRCWC seems to be a mouthful. BRICC may work.}
\david{A reason for avoiding the word "support" is that it might imply a connection to the "supported model".}

\begin{abstract}
We study a framework for modeling distributed network systems assisted by a reliable and powerful cloud service. Our framework aims at capturing hybrid systems based on a point to point message passing network of machines, with the additional capability of being able to access the services of a trusted high-performance external entity (the cloud). 
We focus on one concrete aspect that was not studied before, namely, ways of utilizing the cloud assistance in order to attain increased resilience against Byzantine behavior of machines in the network.
Our network is modeled as a congested clique comprising $k$ machines that are completely connected to form a clique and can communicate with each other by passing small messages. In every execution, up to $\beta k$ machines (for suitable values of $\beta \in [0, 1)$) are allowed to be \emph{Byzantine}, i.e., behave maliciously including colluding with each other, with the remaining $\gamma k$ or more machines  being \emph{honest} (for $\gamma=1-\beta$). Additionally, the machines in our congested clique can access data through a trusted cloud via queries. This externality of the data captures many real-world distributed computing scenarios and provides a natural context for exploring Byzantine resilience for essentially all conceivable problems. Moreover, we are no longer bound by the usual limits of $\beta < 1/3$ or even $\beta < 1/2$ that are typically seen in Byzantine Agreement.

The current paper focuses on utilizing one specific technique for ensuring reliability, based on restricting the allowed cloud services. Specifically, the cloud will only provide read-only access to shared memory and potentially a random bit generator service (i.e., the machines will not be allowed to write on the cloud, nor will the cloud offer computational services). This is used in order to limit the potential influence of the Byzantine machines.

\john{The Patt-Shamir paper \cite{AGP21-sss} studies the download problem under the name of cloudcast or cCast as they call it. Should we rename download to cCast?}
\david{Good idea, but we should remember to do it everywhere, i.e., including inside macros etc. Let's do it before the submission to a conference}

We focus on a few fundamental problems.
We start with the ${\textsf{Download}}$ problem, wherein the cloud stores $n$ bits and these $n$ bits must be downloaded to \textit{all} of the $k$ machines. 
The maximum number of bits that each honest machine must query (called the query complexity and denoted by $\mathcal{Q}$) is  
$\mathcal{Q}= \Omega\left(\frac{\beta n}{\gamma} + \frac{n}{\gamma k}\right)$
for all $\beta$. 
Conversely, for all $\beta \in [0,1)$, suitable randomization (despite an adaptive adversary) and deliberate blacklisting of verified malicious nodes reduces $\mathcal{Q}$ to $\tilde{O}$\footnote{We use $\tilde{O}(f(n))$ as shorthand for $O(f(n) \log^c n)$ for some constant $c$}$\left(\frac{n}{\gamma k} + \sqrt{n}\right)$ w.h.p\footnote{We say that an event holds with high probability (w.h.p) if its probability is 
$1-O(1/n)$.}, and further reduces it to $\tilde{O}\left(\frac{n}{\gamma k}\right)$ when $\beta$ is sufficiently small. When the adversary is non-adaptive, we show that $\mathcal{Q} = \tilde{O}\left(\frac{n}{k}\right)$ (w.h.p.) for all fixed $\beta \in [0,1)$.
In addition to ${\textsf{Download}}$, we study the problem of computing the ${\textsf{Disjunction}}$ and ${\textsf{Parity}}$ of the bits in the cloud. For ${\textsf{Disjunction}}$, we parameterize $\mathcal{Q}$ by $\delta$, the fraction of bits in the cloud that are $1$. 
With randomization, we show that
$\mathcal{Q} = \tilde{O}\left(\frac{1}{\gamma k} \cdot \frac{1}
{\mbox{\boldmath$\delta$}}
\right)$ (w.h.p.),
where $\mbox{\boldmath$\delta$} = \max(1/n, \delta)$.
For the deterministic case, we show that $\mathcal{Q} = \tilde{O}\left(\frac{n}{\gamma k} + \frac{1}{\mbox{\boldmath$\delta$}} + k\right)$ and  $\mathcal{Q} = \Omega\left(\frac{n}{\gamma k} + \frac{1}{\mbox{\boldmath$\delta$}}\right)$, 
which implies that it is impossible to match the randomized query complexity. Our deterministic solution uses a natural variant of bipartite expanders which may be of independent interest. For the ${\textsf{Parity}}$ problem, we show improved bounds for the number of communication rounds and messages that are better than for ${\textsf{Download}}$.

\end{abstract}

\clearpage

\tableofcontents

\section{Introduction}

\subsection{Background and Motivation}

This work studies distributed computing systems assisted by a trusted external entity that provides reliable services to all the participants. Inspired by the plethora of cloud technologies available in the market~\cite{aws,google}, we refer to the external entity  as the {\em cloud}. The scope of our work is more general and potentially applicable in other contexts where the external entity could take other forms -- some of which we discuss shortly. 

There are many real-world contexts in which a network cooperates with an entity that is external to it. Afergan, Leighton, and Parikh of Akamai Technologies, Inc. 
hold a patent for content delivery networks (CDNs) wherein a content server is assisted by a P2P network~\cite{Afergan_Leighton_Parikh_2012}. They have shown the effectiveness of this hybrid approach by which content can be delivered to requesting peers either directly from the content server (i.e., the cloud) or from other peers who may already have the content. 
Going beyond clouds, one can envision such interplays even in pure P2P systems like blockchains and distributed ledgers. Blockchains are large  data sets encrypted in such a way that tampering is impossible without detection~\cite{narayanan2016bitcoin} under standard cryptographic assumptions. The Bitcoin blockchain is currently nearly 500 gigabytes while the ethereum blockchain has exceeded 1000 gigabytes.  Consider the problem of processing the data in such a large blockchain using a P2P network. While these blockchains  fit into the secondary storage of most modern systems, processing them for any data analytics can be prohibitively expensive as the data needs to be decrypted and moved into main memory for processing. In this context, again, a collection of volunteer P2P machines might serve us well. Each machine only needs to decrypt a  small portion of the data set and the machines can collectively process the data.  

One immediate concern in bringing peers together in a coordinated fashion is that not all peers can be trusted. Some peers might behave unexpectedly for reasons ranging from innocuous crashes or network outages to more malicious reasons that may even include collusion. 
Thus, towards demonstrating the robustness of cloud assisted P2P technologies, we focus on one concrete aspect, namely, developing an understanding of how cloud assistance can be utilized in order to attain increased resilience against Byzantine behavior of machines in the network. 
Our work is, to the best of our knowledge, the first 
to consider 
the potential roles of combining cloud services with distributed networks in ensuring resilience against Byzantine failures. 
It focuses on one specific approach, namely, ensuring trust and curbing the influence of Byzantine machines via restricting the cloud services, so that it provides only read-only access to data and a random number generator (for some of our algorithms).


The theory of Byzantine fault tolerance has been a fundamental part of distributed computing  ever since its introduction by Pease, Shostak, and Lamport~\cite{LSM82,PSL80} in the early 80's,
and had a profound
influence on cryptocurrencies, blockchains, distributed ledgers, and other decentralized peer-to-peer systems. 
It largely focused on a canonical set of  problems like Broadcast~\cite{DS83},  Agreement~\cite{B87,LSM82,PSL80,R83}, $k$-set Agreement~\cite{CELT00}, Common Coin~\cite{MR90}, and State Machine Replication~\cite{CL99}.
Some studies have injected Byzantine fault tolerance into other related areas~(cf. \cite{AMPV22,AK08,BGKKS09,DRA22,DPP14}), but these have been ad hoc in nature.
%
In most of these studies, the main parameter of interest is the maximum fraction $\byzfrac$ of the machines that can be corrupted by the adversary in an execution. 

Consider the Byzantine Agreement problem that requires $n$ machines, each with an input bit, to agree on a common output bit that is \emph{valid}, in the sense that at least one honest (non-Byzantine) machine held it as input. In the synchronous setting, even  without cryptographic assumptions, there are agreement algorithms that can tolerate any fraction $\byzfrac < 1/3$ of Byzantine machines~\cite{LSM82} (and this can be extended to asynchronous settings as well~\cite{B87}). When $\byzfrac \ge 1/3$, agreement becomes impossible in these settings~\cite{LSM82}. However, the bound improves to $\byzfrac < 1/2$ with message authentication by cryptographic digital signatures~\cite{RSA78}.  Due to a well-known network partitioning argument (discussed shortly in more detail), $\byzfrac < 1/2$ is required for any form of Byzantine agreement. For most of the Byzantine fault tolerance literature, $\byzfrac$ hovers around either 1/3 or 1/2, with some notable exceptions like authenticated broadcast~\cite{DS83} that can tolerate any $\byzfrac < 1$. 

The main reason for this limitation stems from inherent coupling of data and computing. Consider for instance any Byzantine Agreement variation with $\byzfrac \ge 1/2$.  When all honest machines have the same input bit (say, 1), the Byzantine machines hold at least half the input bits and can unanimously claim 0 as their input bits. This ability of Byzantine machines to spoof input bits makes it fundamentally impossible for honest machines to reach a correct agreement with the validity requirement intact. At the heart of this impossibility is the power of the adversary to control information that is crucial to solving the problem. In fact, this is a common issue leading to many impossibilities and inability to solve  problems exactly (see e.g, \cite{AMP21}). 

Interestingly, having a reliable cloud that provides the data in read-only fashion yields a distributed computing context where access to data cannot be controlled by Byzantine machines. Taken to the extreme, honest machines can simply solve all problems by directly querying the cloud for all required data. However, queries are charged for, and can be quite expensive. So the challenge is to design effective and secure collaborative techniques to solve the problem at hand while minimizing the number of queries made by each honest machine.

\subsection{Byzantine Resilient Congested Clique with Cloud Assistance  Model (\CCA)}

The \CCA\ model consists of (i) a read-only \emph{cloud} that stores the input array comprising $n$ bits and (ii) $k$ machines that form a \emph{congested clique}.

\paragraph{Congested Clique.}
Each of the $k$ machines is uniquely identified by an ID assumed (without loss of generality) to be from the range $[1,k]$. The machines are connected via a complete network. In each round, every machine can send at most one $O(\log n)$ bit message to each of the other machines.  This communication mechanism is referred to as \emph{machine-machine} communication. 

\paragraph{The Cloud.} The $n$-bit input array $\INPUT=\{x_1,\ldots,x_n\}$ (with $n \gg k$) is stored in the cloud. It allows machines to retrieve that data through  queries  of the form $\CloudQuery(i)$, $1\le i\le n$. 
The answer returned by the cloud would then be $x_i$, the $i^{th}$ element in the array.
This type of communication is referred to as \emph{cloud-machine} communication.  

\textbf{Global Random Bits.}
Besides data storage, in one of the three settings described shortly, the cloud provides a \emph{random bit generator} service $\CloudRG(s)$ to the machines. It takes a parameter $s \in \Naturals$ and returns the same random bit to all invocations of $\CloudRG(s)$ with that $s$.
Thus, whenever a subset of machines wish to use a common random bit (possibly at different points in time), they can do so by invoking $\CloudRG(s)$ with the same $s$. For convenience, we use $\CloudRG(S)$ for $S\subset\Naturals$ 
as a shorthand for simultaneously invoking 
$\CloudRG(s)$ for all $s\in S$. 

\textbf{Types of Communication.}
Communication (and computation) is performed in synchronous rounds. Each round consists of two sub-rounds:
(1)~The \emph{query sub-round} of cloud-machine communication, comprising queries of the form $\CloudQuery(\cdot)$ or $\CloudRG(\cdot)$ from a machine to the cloud and responses received from the cloud, followed by (2)~The \emph{message-passing sub-round} of machine-machine communication, consisting of messages exchanged between machines.


At the beginning of the query sub-round,
every machine $M$ can send up to $O(n)$ queries to the cloud. As we shall see, the number of queries is an important complexity measure, so our algorithms typically send significantly fewer queries; $O(n)$ is merely an upper limit. 
$M$ will receive the responses from the cloud by the end of that sub-round. 
At the beginning of the message-passing sub-round, every machine $M$ can send messages of size $O(\log n)$ bits to every other machine. These messages will be received by the intended recipients by the end of the round.
We assume that a machine $M$ can choose to ignore (not process), messages received from any other machine during any round of the algorithm. Such messages incur no communication cost for $M$.\footnote{An honest machine $M$ can ignore the messages of a known Byzantine machine $M'$ and thus thwart any ``denial of service'' attack that $M'$ attempts on $M$. Such messages sent by the Byzantine machine $M'$ to $M$ will not be counted towards the message complexity.}
Besides communicating, in each round, every machine can locally perform some computation on the data that resides locally. We remark there may be rounds without a query sub-round, but every round must include a message passing sub-round. Some machines may need to query more than others during a query phase. We assume that all machines -- knowing the pseudocode -- will wait for a sufficiently long period of time to ensure that all query responses are received.
Note that while machine-machine communication is required to conform to the CONGEST model, limiting the maximal allowable message size, cloud-machine communication is unrestricted. This assumption is suitable for modeling settings where cloud communication is allocated very high bandwidths, so it makes sense to ignore the time required for it (but not the cost charged by the cloud for its services, which might be high). In settings where cloud communication speeds are comparable to machine-machine communication speeds, it may be necessary to modify the model and impose an upper bound on message size for queries and responses to queries. We remark that our results carry over to this model, with the appropriate increase in the time complexity, incurred by the slower cloud-machine communication.
\david{This last remark requires verification!} 

\john{There may be a way to state the above without the necessity for modifying the model. We can view time as a linear (or some other) combination of query complexity and round complexity. If we take this approach, round complexity is just one measure of time, so we should be careful not to use the terms round complexity and time complexity interchangeably. Our  model allows the query and round complexities to be independently optimized. This gives future researchers the flexibility to either focus on query complexity when queries dominate or on round complexity when the round complexity dominates. Of course, optimizing both makes sense when they both contribute equally.}

\textbf{The Adversarial and Model Settings.}
The behavior of the environment in which our algorithms operate is modeled via an adversary $\Advers$ that is in charge of selecting the input data and fixing the machines' failure pattern. In an execution of an algorithm, a machine is considered \emph{honest} if it obeys the algorithm throughout the execution. A Byzantine machine
is one that deviates from the algorithm in an arbitrary manner
(controlled by $\Advers$).
Denote the set of Byzantine (respectively, honest) machines in the execution by $\BYZ$. (resp., $\cH$). 
Note that even if a machine $M'$ deviates from the algorithm very late in the execution, we consider it Byzantine from the beginning of the execution. (Of course, a Byzantine machine may follow the algorithm in some of its actions.)

We consider three different types of adversaries, leading to different model settings. In all models, $\Advers$ can corrupt at most $\byzfrac k$ machines for some given
$\byzfrac \in (0,1)$. Note that by assumption, $\Advers$ is always allowed to corrupt at least one machine but cannot corrupt all of them; our results are stated under this assumption.  The complementary failure free ($\byzfrac =0$) case is discussed 
in Sec.~\ref{discussion:nofailures}. 
Let $\goodfrac = 1-\byzfrac$,
i.e., there is (at least) a $\goodfrac$ fraction of honest machines. Note that we do not assume $\byzfrac$ to be a fixed constant (unless mentioned otherwise), but we do assume that the number of Byzantine machines, $\byzfrac k$, is integral. The honest machines are unaware of which machines are Byzantine. 

We now define our three adversarial models. For each of them, we specify an upper bound $\eta$ for $\byzfrac$.
In particular, we look at the significant points of $\eta \in \{1,~~ 1/2,~~ \eta_{small}\}$,
where $\eta_{small}$ is a small constant depending on the problem at hand. The notation {\bf Model}$(\byzfrac<\eta)$ captures a \emph{family} of models, one corresponding to any possible value $\byzfrac$ strictly smaller than $\eta$. Algorithms designed for this model assume knowledge of $\byzfrac$, which may be used in the code. (Clearly, the actual number of failures in each execution is a-priori unknown, and cannot be used by the algorithm.)

\begin{description}
\item[$\Det(\byzfrac < \eta$):] 
In this pessimistic setting, the algorithm cannot use randomization at all, and the adversary $\Advers$ is all-knowing. Thus, $\Advers$ knows exactly how the complete execution of an algorithm will proceed and can select Byzantine nodes right at the beginning based on this knowledge. 
\item[$\Harsh(\byzfrac  < \eta)$:] 
Here, the machines may
generate random bits locally, but the cloud does \emph{not} 
provide a random bit generation service.
At the beginning of each round~$i$, $\Advers$ has knowledge of $\INPUT$, all the local random bits generated up to round $i-1$, and all machine-machine and cloud-machine communications up to round $i-1$. At the start of round $i$, it can  corrupt as many machines as it desires, provided the total number of machines corrupted since the beginning of the execution does not exceed $\byzfrac k$. Such an adversary is said to be \emph{adaptive}.
\item[$\Benign(\byzfrac  < \eta)$:] 
Here, in addition to locally generated random bits, the machines can also use the \emph{random bit generation} service of the cloud. Moreover, the adversary is \emph{static}, in the sense that it must choose the corrupted machines before the first round, i.e., relying only on the input.
\end{description}

\textbf{Complexity Measures.}
The following complexity measures are used to analyze our algorithms.
\begin{itemize}
    \dnsitem 
    Query Complexity ($\Query$): the maximum number of queries made by an honest machine during the execution of the algorithm,
    \dnsitem Message Complexity ($\Message$): the total number of messages sent by honest machines during the execution of the algorithm, and
    \dnsitem Round Complexity ($\Time$): the number of rounds (or \emph{time}) it takes for the algorithm to terminate.
\end{itemize}

An algorithm that minimizes the 
query, message or round complexity is said to be 
query, message or round optimal, respectively. 
As queries to the cloud are expected to be the more expensive component in the foreseeable future, we
primarily focus on optimizing the query complexity $\Query$, only trying to optimize $\Time$ and $\Message$ when $\Query$ is optimal (within $\log(n)$ factors). While both $\Query$ and $\Time$ model execution time, we remark that $\Query$ may model other resources like charges levied by the cloud for queries. Moreover, since queries and message passing could employ different technologies, their relative times may exhibit significant variation across implementations. Finally, our definition of~$\Query$ (measuring the maximum cost per machine rather than the total cost) favors a fair and balanced load of queries across honest machines.
Clearly, there are other interesting model combinations to explore besides those studied in this paper. See Section \ref{ssec: future work} for discussion on future work. 

\subsection{Our Main Results and Methods}
\label{ssec:results}

We start by discussing two basic algorithmic techniques used in this paper.  

\inline A. Committees:
Several of our algorithms organize the machines in \emph{committees}, assigned to perform a common task. We distinguish between several types of committees.

\begin{itemize}
\dnsitem
\textbf{Majorizing / representative committees:}
A \emph{majorizing} committee is one guaranteed to have a strict majority of honest machines.
%
A \emph{representative} committee of quality $\rho$, or \emph{$\rho$-representative committee} for short, is one that contains at least $\rho$ honest machines.
\dnsitem 
\textbf{Public / private committees:}
A \emph{public} committee is one whose members are known to all machines. 
(To build small public committees, we use global random bits, which is why we only construct them in the $\Benign$ model.)  
A public committee can be representative or majorizing.
\\
In contrast, a \emph{private} committee is one whose member identities are not publicly known to all. Machines join such a committee based on local random coin tosses, so Byzantine machines can masquerade as committee members. In fact, they can claim membership in a committee in one context and disclaim it in another context. It follows that we do not have an upper bound on the committee size, although we can guarantee a lower bound on the number of \textit{honest} machines per committee. This means that a private committee can only be representative, but not majorizing (except when $\byzfrac$ is sufficiently small). 
\end{itemize}
Depending on the model parameters ($\byzfrac$, availability of global random bits and adversarial power), one can construct committees of different sizes at different query, round and message complexities.
We use the term \emph{weak committee} to refer to a $1$-representative public committee.


\inline B. Blacklisting Byzantine machines:
During an execution, honest machines can \emph{blacklist} Byzantine ones, after identifying a deviation from the behavior expected of an honest machine, and subsequently ignore their messages.
A Byzantine machine $M'$ can be blacklisted for several reasons.
The most common way used in our algorithm is by directly ``catching'' $M'$ in a lie about the value of some bit. Later, we discuss some additional blacklisting methods.

\subsubsection{Problems of Focus and Their Complexity in the Failure Free Model}\label{discussion:nofailures}


We next introduce the three main problems we focus on in the Byzantine resilient \CCA\ model. To establish a baseline for our various results, we first outline the best possible complexity measures  when there are no Byzantine failures. For $\Query$, the best performance is the total required work divided by $k$, since this work can be distributed in a simple manner.

\textbf{$\IDp$.} We begin with the fundamental $\IDp$ problem, where each of the $k$ machines needs to obtain a copy of all $n$ input bits from the cloud.
$\IDp$ is of interest due to the fact that once it is performed, each honest machine holds the entire input, and thus can perform  any desired computation over the input locally, at no additional costs. Hence, its query cost  serves as a baseline against which to compare the costs of our other, specialized algorithms for specific problems. Moreover, once we have performed $\IDp$, we no longer depend on the availability of the cloud. Observe that any lower bound on $\Query$ for computing any Boolean function on the input will serve as a lower bound for $\IDp$ as well.
%
%
%
To solve this problem in the absence of failures, all $n$ bits need to be queried and this workload can be shared  evenly among $k$ machines, giving $\Query = \Theta(n/k)$. The message complexity is $\Message = \tilde{O}(n k)$ and round complexity is $\Time = \tilde{O}(n/k)$ since  $\Omega(n/k)$ bits need to be sent along each communication link when the workload is shared. 

\textbf{$\ORp$}.
In the $\ORp$ problem, the honest machines must learn whether at least one of the input bits in $\INPUT$ is a 1. We also consider an \emph{Explicit} $\ORp$ version where each machine must learn an index $i$ such that $\INPUT[i]=1$ (or output 0 if there are no 1's).
The complexity of the problem is closely tied to the \emph{density} $\Density$ (i.e., the fraction of ones) in the input. In fact, the relevant parameter is often $1/\ModifiedDensity$ where $\ModifiedDensity=\max(1/n,\Density)$ to handle the exceptional case when $\Density=0$.

Let us 
consider the Explicit $\ORp$ problem.
In the deterministic setting, 
at least $n - \ModifiedDensity n+1$ queries are required in total. Consequently the best deterministic query complexity is $\Query= O(n(1-\ModifiedDensity)/k)$. The round complexity is $\Time = O(1)$ and message complexity is $\Message = O(k)$. Machines that found a 1 bit can send the index to a ``leader'' machine that
will then broadcast the answer.

Randomization helps when $\Density$ is large. Querying $\left (\frac{1}{\ModifiedDensity} \cdot \ln n \right )$ bits uniformly at random in search for a 1 bit has failure probability of $(1 - \ModifiedDensity)^{\ln n / \ModifiedDensity} \leq 1/n$. Thus $O\left(\InverseDensity \cdot \ln n\right)$ queries are sufficient to find a $1$ w.h.p. 
Even without knowledge of $\Density$, one can simply try density values in decreasing powers of $2$, starting with $1/2$ and eventually land at a $1$ having made at most $O\left(\InverseDensity \cdot \ln n\right)$ queries. We can distribute the work equally amongst $k$ machines, and thus 
$\Query = O(1 + \InverseDensity \cdot \frac{1}{k} \cdot \ln n)$.
 The time and message analysis is similar to the deterministic case, i.e, $\Time = O(1)$, $\Message = O(k)$. Note that $\Query = \Omega(\frac{1}{k} \cdot \InverseDensity)$, for any algorithm that solves the $\ORp$ problem with constant probability. 
 (See Thm.~\ref{thm:LB-ORp-2} for formal proof.)

\textbf{$\XORp$}. In the $\XORp$ problem, the honest machines must output the parity  of the input data $\INPUT$ (i.e., output 1 if and only if the number of 1's in $\INPUT$ is odd). We can adopt a similar approach to solving this as the $\IDp$ problem. However each machine need not send all of its bits, rather just the $\XORp$ of its own bits to a leader who computes and broadcasts the $\XORp$ of all $n$ bits. This reduces message complexity to $\Message = O(k)$  and round complexity to $\Time = O(1)$.

\subsubsection{Results for the Deterministic Model}

We now present our main results and briefly discuss the key ideas behind them. 
%
%
%
We begin with the $\IDp$ problem. Notice that the trivial brute-force deterministic algorithm wherein each machine independently queries all $n$ bits directly from the cloud yields the following baseline.

\begin{theorem}
\label{thm:naiv-ID}
In the $\Det(\byzfrac < 1)$ model, there is an algorithm for the $\IDp$ problem with $\Query = n$ and $\Time, \Message = 0$.
\end{theorem}
This is the best one can hope for when $\byzfrac$ is large.
However, this is far from the lower bound of $\Query= \Omega\left (\frac{n}{\goodfrac k}\right)$ that stems from the need for at least one honest machine to have read each input bit, or the following somewhat stronger lower bound.

\begin{customthm}{\ref{thm:LB-det-InputDistr}}
\label{thm:LB-det-InputDistr-intro}
In the $\Det(\byzfrac < 1)$ model, any algorithm for the $\IDp$ problem has 
$\Query= \Omega(\byzfrac n)$.
\end{customthm}

Indeed, the following complementary theorem shows that the $\IDp$ problem can be solved more efficiently than $\Query = n$ (in fact, with near-optimal query complexity) when $\byzfrac$ is small.   

\begin{customthm}{\ref{thm: alg ID det}}
\label{thm: alg ID det-intro}
In the $\Det(\byzfrac<1/2)$ model, there is an algorithm for the $\IDp$ problem
with 
$\Query= O(\byzfrac n )$, 
$\Time= \tilde{O}(\byzfrac n)$
and 
$\Message= \tilde{O}(\byzfrac nk^2)$.
\end{customthm}

For clarity, in the rest of this section we focus on fixed values of $\byzfrac$. More comprehensive results are provided in later sections.

Further gains in complexity
may be expected when considering specific problems, such as $\ORp$ and $\XORp$.
%
%
For the $\XORp$ problem, we show that a deterministic algorithm cannot do any better (in terms of $\Query$) than solving the $\IDp$ problem first and continuing to solve parity locally.


In contrast, the situation is better for the $\ORp$ problem. The case of $\InverseDensity = O(1)$ is the easiest, wherein we can split the work of querying the bits nearly equally among the honest machines\footnote{only when $k = O(\sqrt{n})$}. Unsurprisingly, $\ORp$ becomes harder when $\InverseDensity$ is larger. The case of $\InverseDensity = \Omega(n)$ is the hardest, in which case $\ORp$ is as hard as the $\IDp$ problem and there is negligible advantage in collaboration. We state our results formally below.

\begin{customthm}
{\ref{thm:lb-det-OR}}[Simplified]
\label{thm:lb-det-OR-intro}
In the $\Det(\byzfrac < 1)$ model for fixed $\byzfrac$, 
any algorithm for Explicit $\ORp$ 
requires 
$\Query= \Omega\left(\InverseDensity + \frac{(1-\Density)n}{k}\right)$ 
even when the density $\Density$ is known to all machines.
\end{customthm}

This result is complemented by two algorithms.
The first algorithm operates in $\Det(\byzfrac < 1)$, and achieves the same query complexity up to polylog factors, plus a term of $k$. Observe that when $k = O(\sqrt{n})$, the upper and lower bounds nearly match upto $\log$ factors. When $k = \omega(\sqrt{n})$, however, the additive $k$ term is dominant.
Specifically, we get the following result.
\begin{customthm}
{\ref{thm: det 1 OR}}[Simplified]
\label{thm: det 1 OR intro}
In the $\Det(\byzfrac < 1)$ model for fixed $\byzfrac$,  there is an algorithm that solves the $\ORp$ problem with 
$\Query= \tilde{O}\left( \frac{n}{ k} +  \InverseDensity+ k\right)$, 
 $\Time= \tilde{O}(1)$ and $\Message= \tilde{O}( k^2 )$.
\end{customthm}

The key idea behind the algorithm is as follows. One can represent the access pattern of the machines to the input array $\INPUT$ as a bipartite graph $G(L,R,E)$, where $L$ represents the $n$ input bits, $R$ represents the $k$ machines, and an edge $(i,j)\in E$ indicates that $M_j$ queries $\INPUT[i]$. We would like to ensure that if the number of bits set to ones in $\INPUT$ exceeds some value $s$, then no matter which set $S$ of indices corresponds to these $s$ ones, the set $\Gamma(S)$ of neighbors of $S$ in $G$ will contain at least $\byzfrac k +1$ machines, guaranteeing that \emph{at least one honest machine will query at least one of the set bits of} $S$. This can be ensured by taking $G$ to be a \emph{Large Set Expander (LSE)}, an expander variant defined formally later on. Not knowing the density $\Density$ in advance, we can search for it, starting with the hypothesis that $\Density$ is close to 1 (and hence using a sparse LSE and spending a small number of queries), and gradually trying denser LSE's (and spending more queries), until we reach the correct density level allowing some honest machine to discover and expose a set bit.

The second algorithm operates in $\Det(\byzfrac<1/2)$, and takes advantage of the fact that there is a majority of honest machines in order to get rid of the extra term of $k$, achieving the optimal query complexity up to polylog factors. We get the following:
\begin{customthm}
{\ref{thm:OR Det 1/2}}[Simplified]
\label{thm:OR Det 1/2 intro}
In the $\Det(\byzfrac < 1/2)$ model  for fixed $\byzfrac$, there is an algorithm that solves the $\ORp$ problem with $\Query = \tilde{O}(\frac{n}{k} + \InverseDensity )$, $\Time = \tilde{O}(1)$ and $\Message = \tilde{O}(k^2 )$.
\end{customthm}


We also look at the explicit version of $\ORp$ and give it an algorithm in the $\Det(\byzfrac<1)$ model. It is based on using a modified expander variant suited to this case. We get the following:
\begin{customthm}
{\ref{thm: det explicit OR + alg name}}[Simplified]
\label{thm: det explicit OR intro}
In the $\Det(\byzfrac < 1)$ model for fixed $\byzfrac$, there is an algorithm that solves the Explicit $\ORp$ problem with $\Query = \tilde{O}(\frac{n}{k} + \InverseDensity + \sqrt{n})$,
$\Time = O(n)$ and $\Message = O(n k^2)$.
\end{customthm}

\subsubsection{Results for the Harsh Model}
\label{sss:harsh-overview}
Before detailing our results for the $\Harsh$ model, we discuss some difficulties posed by the model and the techniques we employ to overcome them.

\inline Targeted failing:
The adversary can choose the failed machines online, based on the progress of the algorithm. This implies that if the algorithm appoints some random machine $M$ to query a bit $x_i$ on some round $t$ of the execution, but communicate the bit to other machines at a \emph{later} round $t'$, then we cannot rely on the hope that the randomly selected $M$ will be honest, say, with probability $1-\byzfrac$, since the adversary gets an opportunity to learn the identity of the chosen $M$ on round $t$ and subsequently corrupt it before round $t'$. 
We overcome this difficulty by adopting the policy that if machine $M$ is randomly chosen for some task on round $t$, then $M$ completes that task \emph{on the same round}.

%
A natural way to manage the work of the machines on various problems is by organizing them into small committees, and assigning tasks to the different committees. Membership in the committees is determined by each machine individually, using private random bits, as the model does not allow the use of global random bits. Hence, in this model we may only use private committees.
This exposes a second difficulty that arises in the $\Harsh$ model.

\inline Byzantine committee infiltration:
As committees are private, Byzantine machines can ``masquerade'' as belonging to any number of committees of their choice, and even gain the majority in some committees. This has two problematic implications.
\begin{itemize}
\dnsitem
{\bf Infeasibility of majorizing committees:} 
One cannot ensure that a \emph{majority} of the committee members are honest, or in other words, we cannot use majorizing committees. 
\dnsitem
{\bf Decision disruption in $\rho$-representative committees:}
Even when using representative committees, Byzantine infiltration might disrupt the ability of the honest machines to decide.
\end{itemize}

To understand this latter difficulty, let us illustrate the use of such committees for the $\IDp$ problem in the $\Harsh(\byzfrac<1)$ model.
Sequentially in rounds $i = 1, 2, \dots n$, 
each (honest) machine in a private $\rho$-representative committee $\cC_i$ queries the $i$th input bit~$x_i$ from the cloud.
Then (still on the same round, to avoid targeted failing), each machine in $\cC_i$ sends the value of $x_i$ to every other machine. A machine $M$ not in $\cC_i$ might receive incorrect values from the Byzantine machines in $\cC_i$. However, as long as it receives $\rho$ or more $b$ values for $b \in \{0, 1\}$ and \emph{fewer} than $\rho$ values being $1-b$, $M$ can be confident (w.h.p) that the correct value is $b$, since the committee contains at least $\rho$ honest machines. 
Unfortunately, once many Byzantine machines infiltrate $\cC_i$, $M$ might receive $\rho$ or more 
zeros \emph{and} $\rho$ or more ones, 
rendering it unable to decide.

To bound the potential effects of the Byzantine committee infiltration problem, and particularly the decision disruption problem, we employ a number of \emph{blacklisting} techniques.
The simplest and most frequently used is the following.
\begin{itemize}
\dnsitem
\textbf{Direct blacklisting for false reporting a query outcome}:
The honest machine $M$ cloud-queries a bit $x_i$ for which the suspected machine $M'$ claimed a particular value, $x_i=b$. If the result of the cloud-query conflicts with this claim, then $M'$ must be Byzantine.
\end{itemize}

Our first algorithm for the $\IDp$ problem in the $\Harsh(\byzfrac<1)$ model employs direct blacklisting when an honest machine $M$ cannot decide on the value of some $x_i$, 
having received 
$\rho$ or more zeros and $\rho$ or more ones.
When this happens, $M$ resorts to querying the cloud for the answer, exposing at least $\rho$ Byzantine machines. Choosing $\rho$ optimally yields the following.

\begin{customthm}
{\ref{thm:ID harsh 1 + alg name}}[Simplified]
\label{thm:ID harsh 1 + alg name intro}
In the $\Harsh(\byzfrac<1)$ model for fixed $\byzfrac$, there is an algorithm that w.h.p. 
solves the $\IDp$ problem with 
$\Query = \tilde{O}\left(\frac{n}{ k} + \sqrt{n}\right)$,
$\Time= O(n)$ and 
$\Message= \tilde{O}(kn + k^2\sqrt{n})$. 
\end{customthm}

The above algorithm falls short of yielding optimal query complexity, due to the  $\sqrt{n}$ additive term. 
We show that a query-optimal algorithm for $\IDp$ exists for smaller $\byzfrac$. (Specifically, we show it for $\byzfrac < 1/3$;
this constant can plausibly be improved.) 
Generally, the algorithm proceeds in phases, each aiming at reducing the number of unknown bits by a constant factor.
Its structure and analysis are rather involved, so we next describe three key ideas on which they are based.
%

\inline A. Cloud, committee and gossip verification:
A machine $M$ has three ways of acquiring and verifying
the value of a bit $x_i$.
\begin{itemize}
\dnsitem
$M$ \emph{cloud-verifies} $x_i=b$ if it directly queries the cloud and receives a reply that $x_i=b$.
\dnsitem
$M$ \emph{comm-verifies} $x_i=b$ by using a $\rho$-representative committee $\cC_i$ (for suitable $\rho$). Specifically, it learns that $x_i=b$ if it receives a message saying that $x_i=b$ from at least $\rho$ members of $\cC_i$, and a message saying that $x_i=1-b$ from fewer than $\rho$ members of $\cC_i$.
\dnsitem
$M$ \emph{gossip-verifies} $x_i=b$ if it receives messages from $\byzfrac k+1$ or more machines, each testifying that it verified $x_i=b$. This suffices since necessarily at least one of these senders must have been an honest machine.
\end{itemize}

\inline B. Additional blacklisting methods:
In addition to direct blacklisting, the algorithm employs the following two methods by which an honest machine~$M$ can blacklist a Byzantine $M'$.
\begin{itemize}
\dnsitem
\textbf{Blacklisting for requesting unnecessary work:}
If $M'$ claims that a certain bit $x_i$ is unknown to it and requests to learn that bit, $M$ can check if this bit is listed at $M$ as \emph{known to all} (meaning that $M$ not only knows the bit value, but also knows that its value is known to all honest machines). If so, then $M$ knows that $M'$ must be Byzantine.
\dnsitem
\textbf{Blacklisting for over-activity:}
$M'$ can be blacklisted as Byzantine for being \emph{over-active}, namely, claiming to have been randomly selected to many more committees than expected.
\end{itemize}

\textbf{C. Boosting knowledge via gossip:}
Each phase employs two rounds of gossip, serving to boost knowledge towards common knowledge\footnote{abusing the formal definition of the (considerably stronger) notion of ``common knowledge."}, in the following sense.
A machine $M$ marks a bit $x_i$ as ``known-to-all" (and keeps $i$ in a set $\KTA_M$) if it not only knows the value~$x_i$, but also knows that every other machine knows this value.
$M$ marks $x_i$ as \emph{\TwoCK}\footnote{meaning that it satisfies the first two levels of common knowledge: everyone knows it, and everyone knows that everyone knows it.} if it knows that every other machine knows that $x_i$ is known-to-all.
Consider the set $\CORE$ of bits $x_i$ that were comm-verified by $\byzfrac k+1$ or more honest machines. Each of these bits will be gossip-learned by all honest machines in the first gossip round. Subsequently, in the second gossip round, all honest machines will report knowing it, so by the end of that round, it will become \TwoCK\ among the honest machines. The set $\CORE$ serves a key role in the analysis; specifically, we show that its complement shrinks by a constant factor in each phase. 

These techniques along with some non-trivial analysis allow us to prove the following theorem.

\begin{customthm}
{\ref{thm:ID harsh small beta + alg name}}[Simplified]
\label{thm: ID harsh small beta intro}
In the $\Harsh(\byzfrac<\frac{1}{3})$ model for fixed $\byzfrac$, there is an algorithm that w.h.p. solves the $\IDp$ problem with 
$\Query = \tilde{O}\left(\frac{n}{k}\right)$, 
$\Time = \tilde{O}(n)$, and
$\Message = \tilde{O}(nk^2)$. 
\end{customthm}

Turning to the $\ORp$ problem, we utilize randomness in order to improve upon the above deterministic algorithm, roughly by a factor of $k$. We use a \emph{verification spreading} procedure, whose task is to disseminate the index of a set bit, known to some honest machine,
to all machines.
The resulting query complexity (given in the following theorem) is near optimal by Thm.~\ref{thm:LB-ORp-2-intro}. 

\begin{customthm}
{\ref{thm: OR Harsh + alg name}}[Simplified]
\label{thm: OR Harsh + alg name intro}
In the~$\Harsh(\byzfrac<1)$ model for fixed~$\byzfrac$, 
there is an algorithm
that w.h.p. solves the $\ORp$ problem with
$\Query = \tilde{O}\left( \frac{1}{k}\cdot\InverseDensity \right)$, 
$\Time = \tilde{O}\left( 1\right)$, 
$\Message = \tilde{O}(k^2)$. 
\end{customthm}

\subsubsection{Results for the Benign Model}

In the $\Benign$ model, the adversary is static and the cloud provides a (\emph{global}) \emph{random bit generator} service  . This allows us to construct public committees using Hashing techniques, eliminating the Byzantine committee infiltration problem of the $\Harsh$ model. In particular, we can construct small weak (\emph{1-representative}) committees leading to more efficient algorithms, asymptotically matching the lower bounds on query complexity, with improved message and time complexities. 


We next elaborate on our results.
We start by efficiently constructing weak committees and public majorizing committees using hash functions, spending $\Query=\tilde{O}(1)$ queries.
We then focus on a natural subproblem of $\IDp$ referred to as the \emph{Weak Resolution} problem. Here,
given that a weak committee $\cC$ knows bits at indices $i\in\Known\subseteq[1,n]$, a machine $\Machine$ has to learn these bits from~$\cC$ while reducing the number of direct queries to the cloud. We provide Procedure $\WeakResolve$ that solves the above problem using $O(|\cC|)$ queries and $O(|\Known|)$ rounds. The key idea is that $\Machine$ can blacklist at least one machine in $\cC$ every time it queries the cloud. 
We then describe a faster procedure 
that takes only $\tilde{O}(\gamin n)$ rounds, so it is very fast on inputs that are either very \emph{dense} or very \emph{sparse}. 
We use this procedure to provide an algorithm for the $\IDp$ problem that achieves near optimal query complexity of $\tilde{O}(n/k)$ and takes $\tilde{O}(\gamin nk)$ rounds.  

\begin{customthm}
{\ref{thm: fast linear download}}[Simplified]
\label{thm: fast linear download-intro}
In the $\Benign(\byzfrac<1)$ model for fixed $\byzfrac$, there is an algorithm
for solving the $\IDp$ problem w.h.p. with
$\Query=\tilde{O}\left(\frac{n}{k}\right)$, 
$\Time=\tilde{O}(\gamin nk+n)$, 
$\Message=\tilde{O}\left(\gamin kn\right)$.
\end{customthm}

To speed up the algorithm for $\IDp$ independently of the density, we develop a \emph{parallel} version of the algorithm. This is based on
organizing a collection of weak committees in a specific manner using a structure referred to as a convergecast forest, allowing us to leverage parallelism effectively while achieving near optimality for $\Query$. We describe a convergecasting procedure (that uses Procedure $\WeakResolve$) such that at the end of the procedure, every bit is known to some root of a convergecast forest.
Using these structures and procedures, we get the following result.

\begin{customthm} 
{\ref{thm: download-parallel}}[Simplified]
\label{thm: download-parallel-intro}
In the $\Benign(\byzfrac<1)$ model for fixed $\byzfrac$, there is an algorithm
for solving the $\IDp$ problem w.h.p. with 
$\Query = \tilde{O}(\frac{n}{k})$, 
$\Time = \tilde{O}(\frac{n}{k}+k)$, 
$\Message= \tilde{O}(nk)$.
\end{customthm}

When $\byzfrac<1/2$, we use \emph{majorizing} committees to obtain a faster solution for $\IDp$.

\begin{customthm}
{\ref{thm: Disjunction majorizing Benign(1/2)}}[Simplified]
\label{thm: Disjunction majorizing Benign(1/2)-intro}
In the $\Benign(\byzfrac<1/2)$ model, there is an algorithm
for solving the $\IDp$ problem with $\Query= \tilde{O}\left(\frac{n}{k}\right), \Time = \tilde{O}\left(\frac{n}{k}\right)$, $\Message = \tilde{O}(nk)$. 
\end{customthm}

We next investigate if it is possible to solve the $\XORp$ problem faster than $\IDp$. We describe a subproblem similar to Weak Resolution called \emph{Weak $\XORp$ Resolution} in which a weak committee knows bits at indices $i\in\Known$ and a machine $\Machine$ has to learn the $\XORp$ of these bits. We provide an algorithm
that uses an efficient blacklisting procedure akin to binary search and takes only $\tilde{O}(1)$ rounds. This along with our convergecasting procedure leads to a faster algorithm for $\XORp$:

\begin{customthm}{\ref{thm: benign parity converge}}[Simplified]
\label{thm: benign parity converge intro}
In the $\Benign(\byzfrac<1)$ model for fixed $\byzfrac$, there is an algorithm that w.h.p. solves the $\XORp$ problem with 
$\Query = \tilde{O}(\frac{n}{k})$, 
$\Time = \tilde{O}(1)$ when $k^2\le n$ and 
$\Time = \tilde{O}(\frac{n}{k}+\frac{k^2}{n})$ when $k^2>n$,
$\Message = \tilde{O}(n)$.
\end{customthm}

When $\byzfrac<1/2$, we use public majorizing committees to solve the $\XORp$ problem in $\tilde{O}(1)$ rounds for all $k$. We also provide the following lower bound for $\ORp$.

\begin{customthm}{\ref{thm:LB-ORp-2}}
\label{thm:LB-ORp-2-intro}
In the $\Benign(\byzfrac < 1)$ model, any randomized algorithm that solves $\ORp$  with constant probability has $\Query=\Omega(\frac{1}{ k}\cdot\InverseDensity)$ in expectation. 
\end{customthm}

\subsection{Directions for Future Work}
\label{ssec: future work}

Our framework adds Byzantine resilience to standard distributed computing  with the help of a cloud, an entity that is external to the network. We initiated this study through the $\Det(\cdot)$, $\Harsh(\cdot)$, and $\Benign(\cdot)$ models, focusing on the $\IDp$, $\ORp$, and $\XORp$ problems, and developed several algorithms, tools, and techniques. Our emphasis was on optimizing the query complexity, but also considered time and message complexities.
Extending our work to other model variations and/or broader classes of problems like graph and geometric problems, data analytics and machine learning problems are natural next steps. 
Ideas from oracle based computation such as  property testing~\cite{property} can be easily adapted to our context. For instance, the cloud can provide access to graph data in the form of adjacency predicates (i.e., queries of the form: is there an edge between vertices $u$ and $v$?), incidence functions (i.e., queries of the form: what is the $i$th neighbor of $u$?), and combinations thereof. We leave these considerations for future work.

Our work has shown that this framework is well-suited for Byzantine resilience owing to decoupling of data and computation that lends well to ``trust, but verify" techniques in an algorithmically rigorous manner. It will be interesting to see the limits to which Byzantine resilience can be pushed in this framework. Going beyond Byzantine resilience, our framework could potentially lead to significantly faster algorithms and exhibit other benefits unavailable in traditional distributed computing.

This framework can be interpreted in multiple ways and applied to a wide variety of contexts.  One can also envision variants in which the cloud offers a richer set of services that may include computation or data re-organization at its end that the machines may need to pay for. Such dynamics can potentially uncover many algorithmic and game theoretic issues like pricing mechanisms and coalition formation.

%
%
%
%
%

\subsection{Related Work}
\label{ssec: related work}

\david{Look at this paper to see if there were earlier papers on DC+cloud} \john{I could only spot one other work \cite{FriedmanKK13} apart from \cite{AGP21-sss}. \cite{FriedmanKK13} deals with consensus under benign crash failures. They claim to address Byzantine failures in the full version, but I am unable to find it online. In fact, I don't have access to the published conference version of \cite{FriedmanKK13} as it is behind a paywall. Perhaps, David and Srikkanth can see if their institutions provide access to it and share it with the rest of us.}
\david{Sent you the paper.} 

As outlined earlier, to the best of our knowledge, we are the first to 
study the impact of combining distributed computing with an external cloud on fault resilience aspects of computing. In essence, we have proposed a hybrid combination of two technologies -- the cloud and a distributed network -- that is in some ways better than the sum of its parts. Such hybrid combinations leading to overall improvements is not new~\cite{sigcomm11-J}. In the context of distributed computing theory, we have seen hybrid networks~\cite{AugustineHKSS20,DBLP:conf/opodis/CoyC0HKSSS21,DBLP:conf/wdag/Fraigniaud0PRWT22,DBLP:conf/podc/KuhnS20} that benefit from combining low-bandwidth global  networks with high-bandwidth local networks.

There have been a couple of prior works that study distributed computing aided by a cloud. To the best of our knowledge, Friedman {\it et al.}~\cite{FriedmanKK13} were the first to study distributed computing aided by the cloud. Their studied asynchronous consensus with the cloud providing a common compare-and-swap (CAS) register access. They focused on optimizing the number of CAS operations and showed a lower bound of $f+1$ CAS operations for deterministic protocols when the number of benign crash failures are at most $f$. To complement it, they provided a matching deterministic consensus mechanism. Additionally, they showed improved results with randomization.  

More recently, Afek {\it et al.}~\cite{AGP21-sss} introduced the computing with cloud (CWC) model wherein traditional distributed computing models were augmented with one or more cloud nodes that are typically connected to several regular nodes. Under this setting, they study several fundamental primitives leading to  operations like cloud combine (i.e., computing an aggregation function over data distributed in the network) and cloud cast (where data stored in the cloud is to be disseminated through the network). Subsequently, they develop a range of applications based on these operations. Their work is limited to fault free settings and does not use any cloud features for fault tolerance of any form.

It also known how trusted hardware primitives can be used in various ways to enhance Byzantine resilience in asynchronous systems. See~\cite{PODC:DN21} and references therein. Our work can fall under those that employ \emph{shared memory} to enhance Byzantine resilience. However, there is an important philosophical difference here in that the motivation of our work is not necessarily to provide shared memory but to decouple data from computing. Further, we consider Byzantine resilience in problems beyond Byzantine agreement and exhibit settings where an arbitrary number of faults can be tolerated.
\david{It looks like their approach aims at dealing with overcoming asynchrony. They seem to assume that in synchronous systems there's no problem. I'm not sure why... Check if they assume the use of cryptographic tools.}

Most other such hybrid combinations have lent to improved efficiency.  Our work has shown that combining cloud and distributed network yields greater Byzantine resilience, which has been a subject of significant research investigation since its introduction in the early 80's~\cite{LSM82,PSL80}. More recently, there has been a significant uptick in Byzantine resilience research stemming from bitcoin~\cite{nakamoto2009bitcoin} and the subsequent proliferation of peer-to-peer distributed ledgers~\cite{baird2016swirlds,algorand,wood2014ethereum}. 

As discussed earlier, Byzantine resilience research was largely limited to a few problems like Byzantine Agreement, Byzantine Broadcast, State Machine Replication, etc. More recently, we have seen many investigations of Byzantine resilience in other problems and models. Quite naturally, it has been explored in peer-to-peer (P2P) settings to ensure robust membership sampling~\cite{BGKKS09} and resilient  P2P overlay design~\cite{DBLP:conf/esa/FiatSY05,DBLP:conf/spaa/AugustineCP22}. Apart from that, we have seen Byzantine resilience explored in the context of mobile agents~\cite{DPP14,DBLP:journals/dc/BouchardDL22,DBLP:conf/atal/DatarNA23} and graph algorithms~\cite{AMPV22}. In the last decade, we have seen quite a bit of interest in Byzantine resilient learning starting with multi-armed bandit problems~\cite{AK08}. More recently, we have seen a flurry of works inspired by the popularity of Byzantine resilient optimization algorithms in federated and distributed learning~\cite{DBLP:conf/amcc/SuV16,DBLP:conf/nips/BlanchardMGS17,DBLP:conf/podc/GuptaV20,DBLP:conf/nips/El-MhamdiFGGHR21,DBLP:journals/tac/SuV21,DBLP:conf/icml/FarhadkhaniGHV22}.

Most works on Byzantine resilience have focused on models and problems where the data is integrated into the network, making it difficult to get Byzantine resilience past $\byzfrac < 1/3$ or $\byzfrac < 1/2$. However, there have been some  exceptions that were observed quite early in the Byzantine resilience literature, like authenticated broadcast~\cite{DS83} that can be achieved for any $\byzfrac < 1$.  More recently, we have seen the power of decoupling data and computing in the context of mobile agents. Consider the problem of gathering in which mobile agents on a graph need to gather at one location. When Byzantine agents cannot change their identity labels, honest agents can gather for all $\byzfrac < 1$, provided the network size is known~\cite{DPP14}. Crucially, the honest agents can explore every part of the graph. The Byzantine agents do not control any portion of the graph. On the flip side, consider~\cite{AMPV22} that studies graph connectivity testing in the congested clique network in the presence of Byzantine failures. They provide efficient algorithms to distinguish between the following two cases: (i) the graph induced by the honest nodes is connected and (ii) the graph as a whole is far from connected. Unfortunately, there is large gap between these two extremes. As shown in~\cite{AMPV22}, this gap cannot be bridged because significant portions of the graph are only accessible to Byzantine nodes.

Finally, our work can be viewed as a step towards understanding the power of oracles in distributed computing. Oracles have been effectively used in computational complexity to hide some computational steps to focus on others. See~\cite{hemaspaandra2001complexity} for an excellent treatment of the various structural complexity theory results that have been obtained through oracles.  Interestingly, the power of oracles in distributed computing has been explored before in distributed computing  in the context of overcoming the challenges posed by failures in asynchronous settings~\cite{mostefaoui2002introduction}.  On the broader algorithmic front, the property testing model~\cite{property} can be viewed as using oracles to access data that is only available through expensive queries. The essence of property testing (and more generally, sub-linear algorithms) is to compute approximate solutions with very few queries. In fact, our framework can be viewed as a distributed computing extension to the property testing framework. Developed further, we envision that our framework can be suitably adapted to include broader notions of oracles.

\section{\CCA\ Algorithms in the Deterministic Model}

In this section, we study the deterministic version $\Det$ of the \CCA\ model wherein the machines do not have access to any random bits. Thus, the adversary will know exactly how the algorithm execution will progress. In particular, the adversary can select which nodes will be Byzantine based on complete knowledge of the executions under each choice of Byzantine machines and their behaviors. Recall that $\Det(\byzfrac < \eta)$ refers to the deterministic version of the \CCA\ model with up to $\byzfrac k$ Byzantine machines.

We begin by studying the most fundamental problem, namely the $\IDp$ problem wherein every honest machine must learn all $n$ bits of the input array $\INPUT$, thereby serving as a good baseline solution for all problems.  We will subsequently  present algorithms  for the $\ORp$ problem with complexities that improve upon the baseline solution. 

\subsection{The $\IDp$ Problem in the Deterministic Model}
\label{subsec: ID-det}

The naive algorithm for $\IDp$, where each machine queries the entire input array, works in $\Det(\byzfrac < 1)$ and requires a query complexity $\Query=n$ per honest machine (Thm. \ref{thm:naiv-ID}). 


However, we can improve the query complexity $\Query$ when $\byzfrac$ is smaller.
\begin{theorem}
\label{thm: alg ID det}
In the $\Det(\byzfrac<1/2)$ model, there is an algorithm for $\IDp$ 
with $\Query= O(\byzfrac n)$,
$\Time= \tilde{O}(\byzfrac n)$
and 
$\Message= \tilde{O}(\byzfrac nk^2)$.
\end{theorem}


\begin{proof}
The key observation is that it suffices if each bit $i$ is queried by a public majorizing committee $\cC_i$, since when such a committee sends the bit to every other machine, each other (honest) machine can trust the majority of the received messages.

A public majorizing committee can be defined deterministically by pre-defining $2\byzfrac k +1$ machines as its members. 
For each bit $x_i \in \INPUT$, $1 \le i \le n$ taken in order, we assign $2\byzfrac k+1$ machines as $\cC_i$ in a round robin with wrap-around fashion. Specifically, we assign machines $(i-1)(2\byzfrac k +1)+j \pmod k + 1$, $0 \le j < 2 \byzfrac k + 1$, to $\cC_i$. This will ensure that
\begin{enumerate}
    \item each committee gets $2 \byzfrac k + 1$ members, thereby establishing correctness, and
    \item at most each machine appears in at most $O(\byzfrac n + n/k)$$=O(\byzfrac n)$ (since $\byzfrac\ge1/k$)  committees, thereby establishing  bounds on all complexity measures. \qedhere
\end{enumerate}
\end{proof}

We now present a matching lower bound on $\Query$ for the $\IDp$ problem.

\begin{theorem}
\label{thm:LB-det-InputDistr}
In the $\Det(\byzfrac<1)$ model,  
any algorithm for the $\IDp$  problem has
$\Query=\Omega(\byzfrac n)$.
\end{theorem} 



\begin{proof}

To establish this, we prove a slightly stronger claim.
Consider a deterministic algorithm $\Algor$ for the $\IDp$ problem. For an $n$-bit input $\INPUT$, let $\execution(\INPUT)$ denote the (unique) execution of $\Algor$ on $\INPUT$ in which none of the machines has failed. Then the following holds.

\begin{lemma} \label{lem:lb:4}
For every $\INPUT$, every bit $\INPUT[i]$ ($1\le i\le n$) is queried by at least $\byzfrac k+1$ machines during the execution  $\execution(\INPUT)$.
\end{lemma}
\begin{proof}
Towards contradiction, suppose there exists an input $\INPUT=\{x_1,\ldots,x_n\}$ and an index $1\le i\le n$ such that in the execution $\execution=\execution(\INPUT)$, the set $\hM$ of machines 
that queried the bit $x_i$ is of size $|\hM| \le \byzfrac k$. Without loss of generality let $x_i=0$.

The adversary can now apply the following strategy. It first simulates the algorithm $\Algor$ on $\INPUT$ and identifies the set $\hM$. It now generates an execution $\execution’$ similar to $\execution$ except for the following changes: (a) The input $\INPUT'=\{x'_1,\ldots,x'_n\}$ in $\execution’$ is the same as $\INPUT$ except that $x’_i=1$. (b) The machines of $\hM$ are Byzantine; all other machines are honest. (c) Each Byzantine machine $M\in \hM$ behaves according to $\Algor$ except that it pretends that $x'_i=0$, or in other words, it behaves as if the input is $\INPUT$ (and the execution is $\execution$). 

One can verify (e.g., by induction on the rounds) that the honest machines cannot distinguish between the executions $\execution$ and $\execution’$. Therefore, they end up with the same output in both executions. This contradicts the fact that their output in $\execution$ must be $\INPUT$ and their output in $\execution’$ must be $\INPUT’$.
\end{proof}

The lemma implies that for every input $\INPUT$, the total query complexity of the algorithm is greater than $\byzfrac kn$.
Theorem \ref{thm:LB-det-InputDistr} follows.
\end{proof}

\subsection{The $\ORp$ Problem in the Deterministic Model}

In this section we consider the problem of computing the $\ORp$ of the input bits. We first establish some basic lower bounds, and then present some efficient deterministic $\ORp$ algorithms.

\subsubsection{Lower bounds for the $\ORp$ problem}

It is clear that even allowing randomization, in the worst case one needs to spend $\Omega(n/k)$ queries per machine. 
If every machine queries $o(n/k)$ bits (in expectation), then overall there are only $o(n)$ indices (in expectation) read by all machines. There must be some bit whose probability of being read by \textit{any} of the $k$ machines is $o(1)$. 

A more interesting picture emerges when one considers
the complexity of the $\ORp$ problem as a function of 
the input \emph{density} $\Density$ that quantifies the proportion of ones in the input; see the definition of $\ORp$ below. For simplicity, $\Density n$ is always assumed to be integral.
In particular, as the density of the input increases, a witness index $i$ (namely, such that $x_i=1$) can be found with fewer queries. 
To facilitate a more nuanced understanding, we present the following definition.

The \emph{$\ORp(\Density)$ problem} for $0<\Density\le1$ requires machines to output the $\ORp$ of all bits in the input $\INPUT$
given that the exact number of 1's in $\INPUT$ is either 0 or $\Density n$. 

The \emph{Explicit $\ORp(\Density)$ problem} is the same as the $\ORp(\Density)$ problem except that if the output is 1, each honest machine must also output an index $i$ such that $x_i$ is 1.

All our lower bounds are for the  $\ORp(\Density)$ problem. 

\begin{lemma}
\label{lem:LB-det-OR-dense}
In $\Det(\byzfrac<1)$, 
any algorithm for $\ORp(\Density)$ requires 
$\Query= \Omega\left (\frac{(1-\Density)n}{\goodfrac k}\right)$ queries 
in the worst case. The bound holds also for the \emph{Explicit} $\ORp(\Density)$ problem. 
\end{lemma}

\begin{proof}
Fix some algorithm $\Algor$ for $\ORp(\Density)$.
Consider an adversary that simulates the entire execution of the algorithm $\Algor$. During the simulation, for each round $t = 1,\ldots, r$, let $\INDEX_t$ be the set of indices $i$ such that $\INPUT[i]$ was queried by an honest machine in round $t$ and it was never queried by any honest machine in any of the prior rounds $1,\ldots, t-1$. Let $q_t$ be the number of queries made by the honest machines in round~$t$. Clearly, $|\INDEX_t| \leq q_t$.

Let round $t^*$ be the first round such that $\sum_{j=1}^{t^*} |\INDEX_j| \le (1 - \Density) n$ and $\sum_{j=1}^{t^*+1} |\INDEX_j| > (1 - \Density) n$. If such a round $t^*$ exists, then the total query complexity is at least
$\sum_{j=1}^{t^*} q_j \ge \sum_{j=1}^{t^*} |\INDEX_j| \ge (1 - \Density)n$. This completes the proof. So, it suffices to show that such a round $t^* < r$ exists.

Towards a contradiction, let us assume that such a round $t^*$ does not exist.
This implies that $\sum_{j=1}^r |\INDEX_j| \le (1 - \Density)n$ where $r$ is the total number of rounds taken by the algorithm $\Algor$ to terminate.

The adversary now chooses the following two distinct inputs $\INPUT$ and $\INPUT'$, where $\INPUT'$ is the all-0 vector.
Let $\INDEX^* = \bigcup_{j=1}^{r} \INDEX_j$ where $r$ is the total number of rounds taken by the algorithm to terminate.
It defines $\INPUT = \{x_1,\ldots, x_n\}$ to be such that the value of $x_i$ for indices $i \in \INDEX^*$ is 0.
For $i \notin \INDEX^*$, choose the lexicographically first $(1 - \Density)n - \sum_{j=1}^r |\INDEX_j|$ indices to be such that $x_i = 0$ and let the remaining indices be 1.
Note that this implies that $\INPUT$ has exactly $\Density n$ 1's.

Let us consider the executions $\execution(\INPUT)$ and $\execution'(\INPUT')$ of the algorithm $\Algor$ on inputs $\INPUT$ and $\INPUT'$. In the execution of $\execution'(\INPUT')$, the Byzantine machines proceed according to the algorithm. However, in the execution of $\execution(\INPUT)$, the Byzantine machines behave exactly as they did in $\execution'(\INPUT')$. Since the view of the honest machines in both executions are the same (as the the set of indices queried by the honest machines in both executions is $\INDEX^*$), the honest machines end up with the same output in both executions which leads to a contradiction.
\end{proof}

\begin{lemma}
\label{lem:lb-det-or-sparse}
In $\Det(\byzfrac < 1)$ model, any algorithm for $\ORp(\Density)$ has 
$\Query \ge \byzfrac\cdot\InverseDensity$.
The bound holds also for the \emph{Explicit} $\ORp(\Density)$ problem. 
\end{lemma}

\begin{proof}
Consider a deterministic algorithm $\Algor$ for $\ORp(\Density)$ that makes fewer than $\byzfrac k\cdot\InverseDensity$ queries in total. This also implies that there must exist at least one machine making fewer than $\byzfrac\cdot\InverseDensity$ queries.
The adversary first simulates the execution $\execution$ of algorithm $\Aor$ assuming that the input consists of only zeros. 
The adversary then builds a bipartite graph $G(L, R, E)$ where $L$ denotes the set of machines and $R$ denotes the set of input bits. We draw an edge between machine $i$ and bit $x_j$ if machine $i$ queried $x_j$ during the simulation.

\begin{claim}
There exists $B \subseteq R$ such that
$|B| \leq \Density n$ and $|\Gamma(B)| \leq \byzfrac k$.
\end{claim}

\begin{proof}
Let $B$ be the set of the $n\Density$ smallest degree vertices in $R$. Then 
$$|\Gamma(B)| ~\le~ \sum_{v\in B} deg(v) ~\le~ 
\frac{n\Density}{n} \cdot \sum_{v\in R} deg(v) ~<~
\frac{n\Density}{n} \cdot \byzfrac k\cdot\InverseDensity ~=~ \byzfrac k,$$ 
where the penultimate inequality follows by the fact that $\sum_{v\in R} deg(v)$ equals the number of queries in the execution.
Hence the chosen $B$ satisfies the requirements of the claim.
\end{proof}

Consider an execution $\execution'$ of $\Algor$ when, (i) input in the cloud is such that $B$ contains $1$ and other bits are $0$, (ii) adversary corrupts $\Gamma(B)$ at the start of $\Algor$ and instructs the machines to locally swap the results to queries in $B$, and otherwise follow $\Algor$.

Observe that the executions 
$\execution$ and $\execution'$ are indistinguishable to honest machines.
Therefore,  $\Algor$ is incorrect, thereby establishing our contradiction.
\end{proof}
Putting Lemma~\ref{lem:LB-det-OR-dense} and Lemma~\ref{lem:lb-det-or-sparse} together, we get the following theorem.
\begin{theorem}
\label{thm:lb-det-OR}
In the $\Det(\byzfrac<1)$ model, any algorithm for the $\ORp(\Density)$ and the Explicit $\ORp(\Density)$ problems has 
$\Query= \Omega\left(\byzfrac\cdot\InverseDensity+\frac{(1-\Density)n}{\goodfrac k}\right)$. 
\end{theorem}

This concludes the part on the lower bounds for the $\ORp$ problem.
The remainder of this section deals with efficient deterministic algorithms for $\ORp$ and Explicit $\ORp$ under different settings.

\subsubsection{The $\ORp$ Problem in $\Det(\byzfrac<1)$}
\label{sssec:ORp det}


We describe here a deterministic algorithm to compute the $\ORp$ function whose query complexity is $\tilde{O}(\frac{n}{k} + \InverseDensity + k)$. (For constant $\Density$, this matches our randomized approach for $\XORp$.)
The algorithm makes use of bipartite graphs with certain desirable expansion properties for large vertex sets.
We refer to these graphs as \emph{large set expanders} and define them below.

\begin{defn}[\newexp\ (LSE)]\label{def:lse}
A bipartite graph $G(L, R)$ is an \emph{$(n, k, \byzfrac, \Density)$-Large Set Expander} (or $(n, k, \byzfrac, \Density)$-LSE) if $n=|L|, k=|R|$ and $|\Gamma(S)| > \byzfrac k$ for all $S \subseteq L$ with $|S| \geq n\Density$.
\end{defn}

Informally, a large set expander is such that for every large enough subset $S$, i.e., $S \subseteq L$ and $|S| \geq \Density n$, its neighborhood cannot be covered fully by any subset of $\byzfrac k$ vertices, i.e., $|\Gamma(S)| > \byzfrac k$.
The definition of a large set expander is similar to that of expander graphs and we use a similar probabilistic analysis to prove their existence. We formalize this in the lemma below (whose proof is deferred to the appendix).

\begin{lemma}
\label{lem:large-set-expander}
There exists a bipartite graph $G(L, R)$ that is a $(n, k, \byzfrac, \Density)$ $\newexp$ such that, 
(1)~Every vertex in~$L$ has degree at most~$d$, and
(2)~Every vertex in~$R$ has degree at most~$\frac{2nd}{k}$,
for all $d$ satisfying
\begin{equation*}
        d > \max \left\{ \frac{1 + \log (e\cdot\InverseDensity)}{\log \frac{1}{\byzfrac}} + \frac{\byzfrac k}{\Density n} \cdot \frac{\log\frac{e}{\byzfrac}}{\log \frac{1}{\byzfrac}} , \frac{3k\ln{2k}}{n}\right\}.
\end{equation*}
\end{lemma}

\def\PROOFA{
    Consider the following random experiment that constructs a simple bipartite graph.
    \begin{enumerate}
        \item Construct an empty ($0$ edge) bipartite graph $G(L, R)$ with $n$ vertices in $L$ and $k$ vertices in $R$.
        \item For every $u \in L$, choose $d$ vertices in $R$ independently and uniformly at random. Add an edge between $u$ and every vertex sampled.
    \end{enumerate}

    Let the graph obtained by the above procedure be $G(L, R)$. By construction, every vertex in $L$ has degree at most $d$.

    We calculate the probability that $G$ satisfies the required properties in the lemma. If we can lower bound this probability by a non-zero value, then existence is guaranteed.

    We need to show that $G$ satisfies two properties, first is that it is a $(\byzfrac, \Density)$ $\newexp$ and second is that the degrees of the vertices in $R$ are small.

    Let $E_1$ be the event that $G$ is not a $(\byzfrac, \Density)$ $\newexp$ and let $E_2$ be the event that some vertex in $R$ has degree larger than $\frac{2nd}{k}$. We show independently that $\Prob(E_1) < \frac{1}{2}$ and $\Prob(E_2) < \frac{1}{2}$. It follows from the union bound that $\Prob(E_1 \cup E_2) < 1$ and that $\Prob(\overline{E_1} \cap \overline{E_2}) > 0$.

    To bound $\Prob(E_1)$, first consider the case when sets $S \subseteq L$ with $|S| = \Density n$ and $T \subseteq R$ with $|T| \leq \byzfrac k$ are fixed. Let $p_{fail, S, T}$ be the probability that $G$ fails the $\newexp$ property due to $\Gamma(S) \subseteq T$. If this happens, then $\Gamma(i) \subseteq T$ for every $i \in S$. Let $p_{fail, S, T, i}$ be the probability of this event. Let the r.v. for the $d$ sampled vertices for vertex $i$ be $X_{i,1}, X_{i,2} \dots X_{i, d}$. We have,
    \begin{equation*}
        p_{fail, S, T, i} = \Prob\left(\mathop{\bigcap}\limits_{j=1}^d X_{ij} \in T\right) = \byzfrac^d.
    \end{equation*}

    We also have that $p_{fail, S, T} = \mathop{\bigcap}\limits_{i\in S} p_{fail, S, T, i}$ and therefore $p_{fail, S, T} \leq p_{fail, S, T, i}^{\Density n}$. Finally using the union bound we get,
    \begin{equation*}
            \Prob(E_1) \leq \binom{n}{n\Density} \cdot \binom{k}{\byzfrac k} \cdot \byzfrac^{dn\Density} \leq 2\wedge \left( \Density n \log(e\cdot\InverseDensity) + \byzfrac k \log(\frac{e}{\byzfrac}) - d \cdot (\Density n \log{\frac{1}{\byzfrac}}) \right).        
    \end{equation*}
    In the last step we used the inequality $\binom{n}{r} \leq \left(\frac{ne}{r} \right)^r$. To bound $\Prob(E_1)$ to less than $1/2$, it is sufficient to choose $d$ such that,

    \begin{equation*}
        d > \frac{1 + \log(e\cdot \InverseDensity)}{\log \frac{1}{\byzfrac}} + \frac{\byzfrac k}{\Density n} \cdot \frac{\log\frac{e}{\byzfrac}}{\log \frac{1}{\byzfrac}}
    \end{equation*}

Lastly we look to bound $\Prob(E_2)$. Let $Y_v$ be the random variable corresponding to the degree of a vertex $v \in R$. By symmetry, $Y_v$ is identical for all vertices $v \in R$ and is given by $Y_v \sim Bin(nd,\frac{1}{k})$. Let $\mu=E[Y_v]=\frac{nd}{k}$. Using Chernoff bounds, we get that 
\begin{equation*}
    \begin{split}
        \Prob(Y_v > (1+\epsilon)\mu) \le e^{-\frac{\epsilon^2\mu}{3}}\\
        \Prob(Y_v > \frac{2nd}{k}) \le e^{-\frac{nd}{3k}}\\
        \Prob(E_2)=\Prob\left(\mathop{\bigcup}\limits_{v \in R} Y_v > \frac{2nd}{k}\right) &\leq k \exp(-\frac{nd}{3k}).
    \end{split}
\end{equation*}
If $d>\frac{3k\ln{2k}}{n}$, then $\Prob(E_2)<\frac{1}{2}$. Consequently if we choose $d$ to be greater than the two values obtained for each $E_1, E_2$, we have that $\Prob(E_1 \cup E_2) < 1$ as desired.
} 

We use the existence of large set expanders to devise Algorithm~\ref{alg:det:or} below which allows us to compute the $\ORp$ function such that the query complexity is $\tilde{O}(\frac{n}{k} + \InverseDensity + k)$. We state the result in the following theorem.

\begin{theorem}
\label{thm: det 1 OR}
In the $\Det(\byzfrac<1)$ model, 
Algorithm $\LSEDisjunctOne$ solves $\ORp$ with $\Query = O\left( \frac{n}{k} \cdot \left(\log_{\frac{1}{\byzfrac}} (e^2\InverseDensity) + \log k \right) \cdot \log \InverseDensity + \InverseDensity \cdot  (\byzfrac \log_{\frac{1}{\byzfrac}}{\frac{e}{\byzfrac}})+ k\right)$, $\Time = O(\log n)$ and 
$\Message = O(\byzfrac k^2 \log n)$.
\end{theorem}


Let us unpack the essence of Theorem \ref{thm: det 1 OR} ignoring $\log$ factors and constants dependent on $\byzfrac$. We see that Algorithm \ref{alg:det:or} has query complexity $\Query = \tilde{O}(n/k + \InverseDensity + k)$, essentially matching the lower bound (except for the additive $k$ term). The constant factors increase as $\byzfrac$ gets closer to $1$ and reduce to the naive algorithm when $\byzfrac = 1 - 1/k$.

\begin{proof}
Given the existence of large set expanders, each machine first generates (using the same deterministic algorithm) $\lceil \log n \rceil$ $\newexp\text{s}$ $G_r$, with the $r^{th}$ graph having $\Density_r = \frac{1}{2^r}$, $L=\{x_1, x_2, \dots x_n\}$, i.e., the input bits, and $R = \{M_1, M_2, \dots M_k\}$, i.e., the machines. We then execute the Algorithm \ref{alg:det:or} (written from the perspective of a single machine).

\algnewcommand{\LeftComment}[1]{\Statex \(\triangleright\) #1}
\begin{algorithm}[htb]
\begin{algorithmic}[1]
\State Array $res$ with $\res[i] \gets \nul$ for $i = 1,\ldots, n$ \Comment{Array of bits learnt from the cloud}
\State $\BYZ \gets \emptyset$ \Comment{Machines blacklisted as Byzantine}
%
\For{$r = 1, 2, \dots {\lceil \log n \rceil}$ (sequentially)} 
\State Construct $G_r$ to be an $(n, k,\byzfrac,1/2^r)$-LSE \Comment{See Lemma \ref{lem:large-set-expander}}. 

\State $\res[j] \gets \CloudQuery(j)$ for every
    $j \in \Gamma(M_{\ell})$ in the graph $G_r$ 

\If{$\res[j]=1$ for some $j$} ~send $j$ to all machines and return 1
\EndIf
\State Receive all indices sent by other machines
\For{ every machine $M' \not \in \BYZ$ that sent some index $i$ (in this iteration or the previous)}
\State $\res[i] \gets \CloudQuery(i)$
\If{$\res[i] = 0$} ~Add $M'$ to $\BYZ$
\EndIf
\EndFor
\If{$\res[i]=1$ for some $i$} 
~send $i$ to all other machines and \Return 1 
\EndIf
\State Receive all indices sent by other machines \Comment{handle them in next iteration}
\EndFor
\State $\res[i]\gets\CloudQuery(i)$ for every ${i\in[1,n]}$
\If{$\res$ is an all-0 array} \Return $0$ \Else{} \Return 1 \EndIf. 
\end{algorithmic}
\caption{Deterministic Algorithm $\LSEDisjunctOne$,
 $\Det(\byzfrac<1)$, Code for machine $M_{\ell}$}
\label{alg:det:or}
\end{algorithm}

%
%

It is easy to see that every honest machine returns the correct answer. Either a $1$ is found and verified explicitly by the cloud, or all the bits are queried and no $1$ is found.

The claim about the query complexity is correct whenever $\Density = O(\frac{1}{n})$ since the worst case number of queries is $O(n + \frac{n \log^2 n}{k})$. 

It remains to show the complexity when $\Density = \omega(\frac{1}{n})$. Let $p$ be the smallest positive integer such that $2^p \geq \InverseDensity$. Clearly $p < \lceil \log n \rceil$.
If some honest machine $M_t$ terminated before phase $p$, say phase $q$ then it must have sent $1$ and by phase $q + 1$, every machine receives a correct index that $M_t$ broadcasted (in addition to all possibly many others sent by other machines).
When the machines in phase $q+1$ verify for every other machine, one of the indices that it sent, either they identify some correct answer, or after at most $\byzfrac k + 1 $ queries, verifies the index sent by $M_t$ and terminate successfully.

If none of the honest machines terminated before phase $p$, by the $\newexp$ property, at least one of the honest agents in phase $p$ query the correct answer and by phase $p+1$ all honest agents know the answer.

Let $d_i$ be the degree of the graph at phase $i$.
The maximum number of queries per machine 
is $d_1 + d_2 + \dots d_p + k$ and is equivalent to the expression given in the theorem.
\end{proof}

\subsubsection{The $\ORp$ Problem in $\Det(\byzfrac < 1/2)$}

Whenever there is a strict honest majority in the machines, i.e., $\byzfrac < 1/2$, we show that $\ORp$ can be solved with query complexity $\tilde{O}(n/k + \InverseDensity)$. The idea is to use expanders as before, but this time, whenever more than $1/2$ of the machines found a $1$, the remaining honest machines can conclude that the answer is $1$ without performing any validating queries. We outline the procedure in Algorithm \ref{alg:or-small-beta-2}. 

\begin{algorithm}
\begin{algorithmic}[1]
        \State Set $y \gets 0$
        \For{$r = 1, 2, \dots u= \lceil \log n \rceil$}
        \State Construct $G_r$ to be an $(n, k, 1/2+\byzfrac, \delta_r=1/2^r)$-LSE
        \Comment{See Lemma \ref{lem:large-set-expander}}
        \If{$y = 0$}
                \State $\res[j] \gets \CloudQuery(j)$ for every $x_j \in \Gamma(M_l)$ in the graph $G_r$
                \State $y \gets \bigvee_{x_j \in \Gamma(M_l)} \res[j]$ 
            \EndIf
            \State  Send $y$ to all machines
            \If{at least $k/2$ machines sent $1$} 
            \State $y \gets 1$
            \EndIf
        \EndFor
        \State \Return $y$
\end{algorithmic}
\caption{Deterministic Algorithm $\LSEDisjunctTwo$,
$\Det(\byzfrac<1/2)$, Code for Machine $M_l$}
\label{alg:or-small-beta-2}
\end{algorithm}

\begin{theorem}
\label{thm:OR Det 1/2}
In the $\Det(\byzfrac<1/2)$ model, Algorithm $\LSEDisjunctTwo$ solves $\ORp$ with $\Query = O\left(\frac{n \log n}{k \log(2/(2 \byzfrac + 1))} \right.$
$\left. + \frac{1}{ \log(2/(2 \byzfrac + 1))}\cdot\InverseDensity + \log^2 n\right)$, $\Time = O(\log n)$ and $\Message = O(\byzfrac k^2 \log n)$.
\end{theorem}


\begin{proof}
    We first prove correctness, i.e., the output of all honest machines is equal to the $\ORp$ of the bits in the cloud. We will then argue about the various complexity measures.

    We prove correctness by contradiction. It is easy to see that the value $y$ initially is set to $0$ and either remains $0$ or changes to $1$ at some point in the algorithm. Let $f(i)$ be the phase $r$ in which Machine $M_i$ changed $y$ from $0$ to $1$. If $y$ is $0$ at the end of the algorithm, i.e., after $u$ phases, then define $f(i) = u + 1$.

Let $M_i$ be any honest machine such that (i) the final output $y$ of $M_i$ does not equal the bitwise OR in the cloud and (ii) among all honest machines that satisfy (i), $M_i$ has the least $f(i)$.

\textbf{Case 1}. $f(i) \leq u$. In this case, the actual output is $0$ and some honest machine $i$ outputs $1$ by either Line 6 or Line 8 of the algorithm. It cannot occur in Line 6, as otherwise the cloud has a set bit. It cannot occur in Line 8, since it contradicts the choice of $M_i$ as out of the $k/2$ machines from which $i$ received $1$ in Line 7, at least one of them is honest (say $j$). By definition it must be the case that $f(j) < f(i)$.

\textbf{Case 2}. $f(i) = u + 1$. In this case, the actual output is $1$ and $M_i$ outputs $0$. This implies $\Density \geq 1/n$. Consider phase $r = \lceil \log \InverseDensity \rceil$. By the property of the Large Set Expanders, at least $(1/2 + \byzfrac) k$ machines 
are instructed to query a set bit in Line 5 of phase $r$. Of these machines at least $k/2$ must be honest, i.e., at least $k/2$ honest machines set $y \gets 1$ in Line 6. In Line 8 then, all honest machines must set their values to $1$, which is a contradiction.

The round complexity is straightforward from the pseudo-code, there are $\log n$ phases and $3$ rounds per phase. The message complexity is maximum possible.

For the query complexity, which occurs only in Line 5 of the algorithm, observe that in round $r$, the query complexity is $O(\frac{n}{k \log{2/(2\byzfrac+1)}} + \frac{1/2 + \byzfrac}{2^{-r} \log{2/(2\byzfrac+1)}} + \log k)$ (See Lemma \ref{lem:large-set-expander}). 

Observe that after $\lceil \log \InverseDensity \rceil$ phases, all the honest machines should have either found a one (by the property of the expander), or $\Density = 0 \Rightarrow \lceil \log \InverseDensity \rceil = \lceil \log n \rceil$, which is the end of the algorithm. Summing over $r$ from $0$ to $\lceil \log \InverseDensity \rceil$, we get the desired result. For clarity, we replace certain log factors (such as $\log k$, $\log \InverseDensity$) with the more dominant $\log n$. 
\end{proof}

\subsubsection{Explicit $\ORp$ in $\Det(\byzfrac < 1)$}



Before presenting the algorithm for Explicit $\ORp$, we introduce \emph{Guaranteed Large Set Expanders} which will be used in our solution.
These are bipartite graphs that are like expanders providing us with slightly different properties that are useful in our setting.

\begin{defn}[Guaranteed Large Set Expander (GLSE)]
A bipartite graph $G(L,R)$ is an $(n, k, \byzfrac, \Density, s)$-guaranteed large set expander if $|L| = n$, $|R| = k$, $d=$ $\max_{v \in R} \deg(v) = \tilde{O}(ns/k)$ and for every $S \subseteq L$ with $|S| \geq \Density n$ and every $T \subseteq R$ with $|T| \leq \byzfrac k$, there exists at least one vertex $u \in S$ such that $|\Gamma(u) \setminus T| \geq s$.
\end{defn}

\inline GLSE vs LSE: In Large Set Expanders (Definition~\ref{def:lse}), the neighborhood of every large enough subset $S \subseteq L$, i.e., $|S| \geq \Density n$, cannot be covered fully by any subset of $\byzfrac k$ vertices, i.e., $|\Gamma(S)| > \byzfrac k$. The definition of GLSE imposes the stronger requirement that there exists no subset $T \subseteq R$ of size $\byzfrac k$, which when removed, decreases the degree of \emph{every} vertex in $S$ below~$s$.

In the following lemma, we show that we can construct such Guaranteed Large Set Expanders. (The proof is deferred to the appendix.)

\begin{lemma}\label{lem:glse}
There exists a $(n, k, \byzfrac, \Density, s)$-Guaranteed Large Set Expander (GLSE) for every integer $n, k$, $\byzfrac, \Density \in [0, 1]$ and $s < \min(k, n/2 - \frac{8}{\goodfrac}\left(1 + \frac{k \ln k}{n\Density}\right) )$. In particular, for constant $\byzfrac$ and $\Density$, taking $s = o(n)$ is sufficient to realize the required GLSE. 
\end{lemma}

\def\PROOFGLSE{
Consider the random experiment used to prove existence of ``good'' expanders. 
Start with an empty bipartite graph ($0$ edges) with vertices $L \cup R$ where $L = \{x_1, x_2, \dots x_n\}$ and $R = \{M_1, M_2, \dots M_k \}$. Choose an arbitrary integer $d$ and do the following.
\begin{itemize}
    \dnsitem For each $(i, t) \in [n] \times [d]$, sample an integer $j$ from $[k]$ u.a.r., and connect $x_i$ and $M_j$ by an edge.
\end{itemize}

For given $\byzfrac, \Density$, let $E$ be the event that at the end of the experiment, $\forall S \subset L$ with $|S| \geq \Density n$ and $T \subseteq R$ with $|T| \leq \byzfrac k$, we have there exists $u \in S$ such that $|\Gamma(u) \setminus T| \geq s$. 

We analyze only $\Prob(E)$ and find suitable $s$ so that $\Prob(E) > 1/2$. Recall that by the previous analysis (see Lemma \ref{lem:large-set-expander}), the degree is larger than $n d\log n / k$ with probability less than half. So if $\Prob(E) > 1/2$, then we must have $\Prob(\overline{E}) < 1/2$. By the union bound, the probability that we don't obtain the required expander is less than $1$, thus showing existence.

As usual, we first analyze the probability of failure for a fixed $S, T$ and apply the union bound over all choices of $S, T$.
We fix a set $S \subseteq L$ with $|S| = \Density n$ and a set $T \subseteq R$ with $|T| = k(1-\varepsilon)$, where $\byzfrac = 1-\varepsilon$.

On average, $d \varepsilon$ of the $d$ randomly sampled machines for every bit must fall outside $T$. We calculate the probability that less than half of them ($d \varepsilon/2$) fall outside $T$.
Let $p_{fail,S,i}$ be the probability that for bit $x_i$, at least $d \cdot (1 - \frac{\varepsilon}{2})$ 
out of the $d$ randomly sampled neighbors lie completely in $T$.

Observe that the count of number of neighbors of $x_i$ in $T$, denoted $X_i$, is the sum of $d$ i.i.d. Bernoulli r.v. each with mean $\byzfrac = 1-\varepsilon$. The sum, $\sum X_i$, has mean $\mu = d(1-\varepsilon)$.
We can thus bound the sum using the following Chernoff bound setting $\epsilon = \frac{\varepsilon}{2(1-\varepsilon)}$,
\begin{equation*}
p_{fail,S,T,i} ~=~ \text{Pr}\Big[X_i > d \left(1-\frac{\varepsilon}{2}\right) \Big] ~=~
\text{Pr}\Big[X_i > (1 + \epsilon) \mu \Big]
~\leq~ \exp{\left(-\frac{\epsilon^2 \mu}{2 + \epsilon}\right)} ~\leq~ \exp{\left(-\frac{d \varepsilon^2}{8 - 6\varepsilon}\right)}~.
\end{equation*}
For each of the $n\Density$ bits in $S$ 
the $d$ neighboring machines are sampled independently, so
$$p_{fail, S, T} ~\leq~ p_{fail,S,T,i}^{\Density n}  ~\leq~ \exp{\left(-\frac{\Density n d \varepsilon^2}{8 - 6\varepsilon}\right)}.$$
We now bound $p_{fail}$, the probability that $G$ generated does not satisfy property (1) with $s$ replaced by $d \varepsilon / 2$, as  
\begin{eqnarray*}
p_{fail} &\leq& \binom{n}{n\Density} \cdot \binom{k}{\byzfrac k} \cdot \exp{\left(-\frac{\Density n d \varepsilon^2}{8 - 6\varepsilon}\right)} \\
    &\leq& \exp {\left(n\Density \ln n + k(1-\varepsilon) \ln k - \Density n d \frac{\varepsilon^2}{8-6\varepsilon} \right)}.
\end{eqnarray*}

It is thus sufficient to choose $d > \frac{8 - 6\varepsilon}{\varepsilon^2} \Big[ \ln n + \frac{k \ln k}{\Density n} \Big]$ to bound $p_{fail}$ to below $1/e$.

Let the threshold given above be $d_{th}$. Observe that $d_{th} = \tilde{O}(\InverseDensity)$. Choosing $d = d_{th} + \frac{2s}{\varepsilon}$ is sufficient to get the desired result.
} 

We can now present the following theorem that utilizes Algorithm~\ref{alg:det:explicit_or} to obtain a deterministic algorithm for Explicit $\ORp$ for any $\byzfrac < 1$.

\begin{theorem}
\label{thm: det explicit OR + alg name}
In the $\Det(\byzfrac < 1)$ model, 
Algorithm $\GLSEExplicitDisjunct$ solves
Explicit $\ORp$ with $\Query = O(\frac{n \log n}{k} + \log k\cdot\InverseDensity + \sqrt{n})$, $\Time = O(n)$ and $\Message = O(n k^2)$
\end{theorem}

\begin{proof}
The intuition behind our solution is the fact that if we have a GLSE with parameters $(n, k, \byzfrac, \Density, s)$ and if we assign machine $M_i$ to query the indices corresponding to $\Gamma(M_i)$ of the bipartite graph, then we have the following: for every possible input and every possible behavior of the adversary, there exists at least one 1 that is mapped to at least $s$ honest machines.

In Algorithm~\ref{alg:det:explicit_or}, every machine initially computes (using the same deterministic algorithm) the bipartite expander with parameters $\byzfrac, \Density$. Then machine $M_i$ queries the set of bits in $\Gamma(M_i)$ to the cloud and sends all set bits to all other machines.

By the property of GLSE, at least $s$ machines must have queried the same index containing 1. Each machine then locally chooses one such index, queries this index to verify that the index is $1$. If it has found the index, then it simply returns, otherwise it marks all machines that sent this index as Byzantine and repeats with the remaining indices and machines until a $1$ is found. By the guarantee of GLSE, at least one such index will be found eventually. See Algorithm \ref{alg:det:explicit_or} for the pseudocode. We will now analyze this algorithm and choose $s$ so that query complexity is minimized.

\begin{algorithm}
    \begin{algorithmic}[1]
            \State Construct an $(n,k,\byzfrac,\Density,s)$-GLSE (See \ref{lem:glse}) with parameter $s = k\sqrt{\frac{\byzfrac\goodfrac}{n}}$
            \Statex \Comment{Phase 1: Cloud Queries}
            \State $\res[j] \gets \CloudQuery(j)$ for every $j \in \Gamma(M_l)$
            \State Send $j$ to every machine, for every $j$ with $\res[j] = 1$ 

            \Statex \Comment{Phase 2: Verification}
            \State $\BYZ \gets \emptyset$
            \While{$\exists \ j \in [1, n] $ such that at least $s$ machines not from $\BYZ$ sent $j$}
                \State $\res[j] \gets \CloudQuery(j)$
                \State \textbf{if} $\res[j] = 1$ \textbf{ then } \Return $j$
                \State $\BYZ \gets \BYZ \cup \{t \mid \text{ Machine } M_t \text{ sent } j\}$
            \EndWhile
            \State \Return 0
    \end{algorithmic}
    \caption{Deterministic Alg. $\GLSEExplicitDisjunct$, $\Det(\byzfrac < 1)$, Code for Machine $M_l$}
    \label{alg:det:explicit_or}
\end{algorithm}

The number of queries in Phase $2$ of Algorithm \ref{alg:det:explicit_or} is at most $\frac{\byzfrac k}{s}$ since 
there are at most $\byzfrac k$ Byzantine machines and for every query at least $s$ new Byzantine machines are found.

The query complexity is $\tilde{O}(\frac{n d_{th}}{k} + \frac{ns}{k \goodfrac} + \frac{\byzfrac k}{s})$. 
Optimal value for $s$ is $s = \frac{k\sqrt{\byzfrac \goodfrac}}{\sqrt n}$ and we have that the 
query complexity $\Query$ is as follows,

\begin{equation*}
    \Query = O\left(\frac{n}{k} \cdot \frac{\log n}{\goodfrac^2} + \InverseDensity \cdot \frac{\log k}{\goodfrac^2} + \sqrt{\frac{n}{\goodfrac}}  \right)
\end{equation*}

Note that for the round and message complexity guarantees of the above theorem is bottlenecked by Phase 1, where potentially every machine needs to send $O(n)$ indices to every other machine. This incurs round complexity of $O(n)$ due to congestion and the message complexity follows.

\inline When $\ORp$ is $0$: Suppose all the bits in the cloud are $0$, then in Algorithm \ref{alg:det:explicit_or} all honest machines break out of the while loop, since all candidates will fail verification in Line 8. They will all then conclude that the number of 1's is fewer than $n\Density$. Under the definition of the Explicit $\ORp$ problem, this promises that all  bits  in $\INPUT$ are indeed zero and the output is correct. The complexity measures are the same by similar arguments. 
If Algorithm $\ref{alg:det:explicit_or}$ is run on an input $\INPUT$ with some 1's but fewer than $n\Density$ 1's, then the Byzantine machines can force some honest machine to learn a 1 index and other honest machines to reach Line 10, deducing (correctly) that $\Density$ is not large enough.
\end{proof}

We remark here on the difficulties in adapting Algorithm \ref{alg:det:explicit_or} for $\ORp$. 
The key idea behind Algorithm \ref{alg:det:explicit_or} is that one can tradeoff between the query complexity by having an (suitably chosen) expander with larger degree and with the property that at least one index is queried by $s$ honest machines. This redundancy enables verification to be done more efficiently, saving upto a factor of $\sqrt{n}$ in the query complexity (when $k = n$). In order to ensure redundancy, every machine must crucially make all queries assigned to it. If suppose we were to (as in the previous deterministic algorithm) repeat the two phase procedures by trying out different $\Density$ in decreasing powers of two, the honest machines will not be able to find a suitable position to terminate. The Byzantine machines can enforce certain honest machines to learn of the answer early. In contrast Byzantine machines can also simply pretend that they do not know of the answer. Moreover the honest machines that know of the answer cannot participate in further iterations of the search in an idle manner (such as in Algorithm \ref{alg:or-small-beta-2}),i.e., by not making their assigned queries, since the redundancy property then won't be satisfied. These difficulties prohibit us from getting a similar (or better) query efficient algorithm as Algorithm \ref{alg:det:explicit_or} when $\byzfrac \in [1/2, 1)$ for $\ORp$ without knowledge of $\Density$.
 
\subsection{The $\XORp$ Problem in the Deterministic Model}
The algorithm, based on applying naive (deterministic) $\IDp$ (See Thm.~\ref{thm:naiv-ID}) followed by locally computing the parity of the input at each machine separately, requires 
query complexity $\Query=O(n)$,
with $\Time=\Message=0$.
The following theorem, which can be shown by a proof similar to that of Thm.~\ref{thm:LB-det-InputDistr}, establishes an asymptotically matching lower bound.

\begin{theorem}
\label{thm:LB-det-XOR}
In the $\Det(\byzfrac<1)$  model, any algorithm for $\XORp$ requires $\Query= \Omega(\byzfrac n)$.
\end{theorem}

\section{\CCA\ Algorithms in the Harsh Adversarial Model}
Our algorithms in this model will use private representative committees. 
Recall from Section~\ref{ssec:results} that a $\rho$-representative committee contains at least $\rho$ honest machines and a private committee is one whose member identities are not known to all.
We show below how to construct a private $\rho$-representative committee.

\subsection{Private representative committees}

To construct a private $\rho$-representative committee (i.e., one of quality at least $\rho$ (w.h.p), we have each machine add itself to the committee with probability $p$. For private committees, we only provide guarantees on the number of honest machines that are in the committee (i.e., the quality), since Byzantine machines can claim to be in the committee and there is no obvious mechanism for verification. We assume the $\Harsh(\byzfrac<1)$ model, where the upper bound on the number of Byzantine machines is $\byzfrac$.

Choosing an appropriate value of $p$, we can obtain high probability guarantees on the number of (honest) machines in a committee using standard Chernoff tail bounds. The analysis is as follows (with the proof deferred to the appendix). 

\begin{lemma}
\label{lem:coins}
Consider $k$ i.i.d Bernoulli random variables with bias 
    $p = \min(1, \frac{9 \ln n + 4\rho}{\goodfrac k})$,  $\byzfrac \in [0, 1)$, $n > 1$ and $\rho \le \goodfrac k$, we have with probability at least $1- 2n^{-3}$,
    \begin{itemize}
        \item for any subset of $\goodfrac k$ variables, at least $\rho$ of them are $1$.
        \item At most $(18 \ln n + 8\rho) / \goodfrac$ variables are $1$.
    \end{itemize}
\end{lemma}

\def\PROOFB{
    We argue that each property is satisfied independently and take the union bound.
    Consider any subset $S$ of $\goodfrac k$ machines and let $X_i$ be the r.v. corresponding to the value that machine $i^{th}$ machine from $S$ sampled a $1$. We have $\Prob(X_i = 1) = p$. Using Chernoff bound, we bound the total number of ones sampled by $S$ as follows:
    \begin{equation*}
        \Prob\left(\sum X_i \leq (1 - \epsilon) \mu \right) \leq \exp{\left(\frac{-\epsilon^2 \mu}{2}\right)}
    \end{equation*}
    We have $\mu = \goodfrac k p$ and we want to choose $\epsilon$ so that $(1-\epsilon)\mu = \rho$ and $\epsilon^2 \mu/2 = \tau \ln n$.
Here $\tau$ is a suitably chosen constant so that the probability of failure is at most $n^{-\tau}$. We use $\tau = 3$.
From the second equation we have $\epsilon = \sqrt{6 \ln n/\mu}$ and plugging into the first we get
$$ \mu - \sqrt{\mu} \sqrt{6 \ln n} - \rho = 0,$$
implying that
$$ \sqrt{\mu} = \frac{\sqrt{6 \ln n} + \sqrt{6 \ln n + 4\rho}}{2}~.$$

    It is thus sufficient to choose $p$ such that $p \geq \frac{6 \ln n + 4\rho}{\goodfrac k}$ so that property (1) is satisfied.

    To argue for property (2), we use Chernoff bounds again, but this time with an upper tail bound set with $\epsilon = 1$. We get
    \begin{equation*}
        \Prob\left(\sum X_i \geq 2 \mu \right) \leq \exp{\left(-\frac{\mu}{3} \right)}
    \end{equation*}

    It is thus sufficient to choose $\mu > 9 \ln n$ to obtain the desired probability of failure. Hence setting
    $p$ as in the lemma yields
    both properties.
} 

Thus to form a private $\rho$-representative committee, each machine should sample itself with probability $p = \frac{9 \ln n + 4\rho}{\goodfrac k}$ and the number of honest machines in the committee is $\Theta(\frac{\rho + \ln n}{\goodfrac})$  with probability at least $1 - 2n^{-3}$. We summarize the procedure and its properties below.

\begin{algorithm}
\begin{algorithmic}[1]
\State Every machine tosses a biased coin with probability of heads $p=\min(\frac{6\ln n+4\rho}{\goodfrac k},1)$
\State \Return $\cC=$ set of machines that tossed heads.
\end{algorithmic}
\caption{Procedure $\ElectPrivate$, $\Harsh(\byzfrac<1)$ , Code for machine $\Machine$}
\label{alg:private:comm:election}
\end{algorithm}

\begin{theorem}
\label{thm:private:comm}
In the $\Harsh(\byzfrac<1)$ model, 
for any positive integer $\rho \le \goodfrac k$, 
Procedure $\ElectPrivate$ w.h.p.
results in the election of a $\rho$-representative private committee with $\Query,\Time,\Message = 0$.
\end{theorem}

\subsection{The $\IDp$ Problem in the Harsh Model}
As a starting point, we observe that the $\IDp$ problem has two extreme solution points. 
The first is the naive deterministic algorithm mentioned in Thm. \ref{thm:naiv-ID}, 
which is tolerant to an arbitrary number of Byzantine failures. This solution does not require any communication, but it costs $\Query=\Omega(n)$ queries per machine.
The second extreme point is when there are no failures. In this case, one can distribute the responsibility for querying the data equally among the $k$ machines and obtain the rest by communication among them, thus achieving $\Query=O(n / k)$ queries per machine.

Our algorithms attempt to bridge the above two solutions, i.e., reduce the per machine (worst case) query complexity while tolerating arbitrary number of Byzantine failures in the $\Harsh$ model. 

\subsubsection{Randomized algorithm in $\Harsh(\byzfrac<1)$}
\label{sec:idist:sol1}

The problem posed by the $\Harsh$ model is that the adversary can choose the set of failed machines online, based on the progress of the algorithm. This implies that if the algorithm appoints some random machine $M$ to query a bit $x_i$ on some round $t$ of the execution, but communicate the bit to other machines at a \emph{later} round $t'$, then we cannot rely on the hope that the randomly selected $M$ will be honest, say, with probability $1-\byzfrac=\varepsilon$, since the adversary gets an opportunity to learn the identity of the chosen $M$ on round $t$ and subsequently corrupt it before round $t'$. Hence in order for us to benefit from the fact that some  machine $M$ is randomly chosen for some sub-task on round $t$, it is imperative that $M$ completes that sub-task \emph{on the same round}.

The idea used to overcome this difficulty is as follows. Sequentially in $n$ rounds $i = 1, 2, \dots n$, we 
query bit $x_i$ from the cloud to a private $\rho$-representative committee $\cC_i$, i.e., $x_i$ is queried by each (honest) machine in $\cC_i$. 
Then (still on the same round), each machine in $\cC_i$ sends the value of $x_i$ to every other machine. Machines not in $\cC_i$ might receive incorrect values from the Byzantine machines in $\cC_i$. However, as long as no more than $\rho$ incorrect values are received, each machine can be confident of the majority as the right answer (whp). In case more than $\rho$ machines sent an incorrect value, or more precisely, in case an honest machine receives more than $\rho$ zeros and more than $\rho$ ones, then machines resort to querying the cloud for the answer, forcing at least $\rho$ Byzantine machines to reveal themselves as being Byzantine. Choosing $\rho$ optimally results in a query complexity of $O(\frac{n \log n}{\goodfrac k} + \sqrt{n})$. See Algorithm \ref{alg:idist} for the pseudocode.

\begin{algorithm}
\caption{Algorithm $\DownloadBlacklist$, $\Harsh(\byzfrac<1)$ model, Code for machine $M$.}
\label{alg:idist}
\begin{algorithmic}[1]
\Statex \textbf{Output:} Array $\res$ such that $\res[i]=x_i$ for $i=1,2,...n$
\State $\BYZ \gets \emptyset$ \Comment{Machines known to be faulty}
\For{$i = 1, 2, \dots n$ (in separate rounds)}
\State Form a private $\rho$-representative committee $\cC_i$.\Comment{Parameter $\rho$ is fixed later from analysis}
\If{$M \in \cC_i$}  
\State $\res[i]\gets\CloudQuery(i)$, send $res[i]$ to all machines. \label{lno:idist-Q1}
\EndIf
\State $S_j \gets $ set of machines not in $\BYZ$ that sent $j$ for $j\in\{0,1\}$
\If{$\min(|S_0|, |S_1|) > \rho$} 
\State $res[i]\gets\CloudQuery(i)$.  \label{lno:idist-Q2}
\State $\BYZ \gets \BYZ \cup S_{1-\res[i]}$
\Else
~~ $\res[i]\gets \arg \max\limits_{j=0,1} |S_j|$.
\EndIf
\EndFor
\State \Return $\res$ 
\end{algorithmic}
\end{algorithm}

\begin{theorem}
\label{thm:ID harsh 1 + alg name}
In the $\Harsh(\byzfrac<1)$ model, 
Algorithm $\DownloadBlacklist$ w.h.p. solves the $\IDp$ problem with 
$\Query ~=~ {O}\left(\frac{n \log n}{\goodfrac k} + \sqrt{n}\right)$, $\Time= O(n)$ and $\Message= O(kn \log n + k^2\sqrt{\goodfrac n})$. 
\end{theorem}

\begin{proof}
The correctness follows from the observation that for each bit $x_i$, each honest machine either (i) queries for $x_i$ or (ii) fewer than $\rho$ machines claimed one of the two bit values. Since $\cC_i$ is $\rho$-representative (w.h.p), we can conclude that the correct bit value would have been reported by at least $\rho$ machines.

We analyze the query complexity of Algorithm \ref{alg:idist} and choose $\rho$ optimally.
Queries are made in lines~\ref{lno:idist-Q1} and~\ref{lno:idist-Q2} of Algorithm \ref{alg:idist}.
    
The expected number of queries in line \ref{lno:idist-Q1} is $n p$ where $p$ is the expected size of $\cC_i$. As $n > \goodfrac k$, we have $n p \geq 9 \ln n$. Using Chernoff bounds, we can bound the probability that there are more than $2n p$ queries, similar to the proof of Lemma \ref{lem:coins}. 

For every query in Line \ref{lno:idist-Q2} by a machine, the size of its local set $\BYZ$ increases by $\rho$. Therefore the total cost of these queries is at most $\frac{\byzfrac k}{\rho}$.
\david{Shachar pointed out to me that the last expression which was $\frac{\goodfrac k}{\rho}$ should probably be $\frac{\byzfrac k}{\rho}$. The implications are apparently that the optimization yields the changed bound below.}

The total number of queries (w.h.p.) is at most,
\begin{equation*}
    \Query = \frac{18 n \ln n }{\goodfrac k} + \frac{8n\rho}{\goodfrac k} + \frac{\byzfrac k}{\rho}
\end{equation*}
Choosing $\rho = \max\{ 1, k \sqrt{\frac{\goodfrac\byzfrac}{8n}}  \}$, 
we get 
$\Query = O(\frac{n \log n}{\goodfrac k} + \sqrt{\frac{\byzfrac}{\goodfrac} \cdot n})$.
\\
\david{The significance is only when $\byzfrac$ is nonconstant, but it might affect our attempts to prove a matching lower bound.}
    
From the description of the algorithm. it is clear that the time complexity $\Time$ is clearly $O(n)$, the number of iterations. The message complexity $\Message$ is calculated as the product of the number of honest machines that put themselves in each $\cC_i$ (which is $O(\log n + k/\sqrt{n})$) times $O(k)$ (i.e., the number of messages sent by each honest machine in $\cC_i$) time $n$ (the number of iterations). 
\end{proof}

Observe that the there is a trivial $\Omega(\frac{n}{\goodfrac k})$ lower bound on $\Query$ (in the case where Byzantine machines crash and do not participate in the execution).

The additional $\sqrt{n}$ term can be neglected whenever $k < \sqrt{n}$, since it is smaller than the lower bound in these cases. Thus, in a wide range of cases, the above algorithm is ``near-optimal''. It is also tolerant against the strongest form of Byzantine adversary, one that even has knowledge of random bits sampled up until the previous round.

\subsubsection{Randomized algorithm in $\Harsh(\byzfrac < 1/3)$}
\label{sec:idist:sol2-recursive}

The $\IDp$ algorithm of the previous section worked for the $\Harsh(\byzfrac<1)$ model, but fell short of yielding optimal query complexity. 
In this section we present a query-optimal algorithm for $\IDp$ in the 
$\Harsh(\byzfrac < 1/3)$ model.

\subsubsection*{The algorithm}

Let us now describe the algorithm in more detail.
The algorithm proceeds in  
$J_0 = \left\lceil\log_{1/\alpha} \frac{k}{c\log n}\right\rceil$ phases, 
whose goal is to reduce the number of unknown bits by a \emph{shrinkage factor}
$\alpha<1$.
The algorithm maintains a number of set variables, updated in each phase, including the following.
$\Known_M$ (respectively, $\Unknown_M$) is the set of indices $i$
whose value $\res[i]$ is already \emph{known} (resp., still \emph{unknown}) to $M$. At any time during the execution,
$\Known_M \cup \Unknown_M=$ $\{1,\ldots,n\}$.
$\res[i]=x_i$ is the Boolean value of $x_i$ for every $i\in \Known_M$.
(Slightly abusing notation for convenience, we sometimes treat $\Known_M$ as a set of \emph{pairs} $(i,\res[i])$,
i.e., we write $\Known_M$ where we actually mean $\Known_M\circ \res_M$.)
Each machine also identifies a set $\KTA_M$ of \emph{known-to-all} bits.

After phase $J_0-1$, the number of bits unknown to each honest machine $M$ reduces to 
$\alpha^{J_0}n = O(\frac{cn\log n}{k}) = {\tilde O}(n/k)$, 
thus allowing $M$ to cloud-query these remaining bits directly. 

We next detail the code of the algorithm, and of its main procedures.
(Hereafter, we omit the superscript $J$ when clear from the context.)

\begin{algorithm}
\caption{Algorithm $\DownloadGossip$, $\Harsh(\byzfrac<1/3)$, code for machine $M$}
\label{alg:idist-1 over 41}
\begin{algorithmic}[1]
\State $\Known_M\gets\emptyset$ \Comment{Indices of bits known to $M$}
\State $\KTA_M\gets\emptyset$ \Comment{Indices of bits that are known-to-all}
\State $\Unknown_M\gets \{1,\ldots,n\}$ 
\State $\INDEX_M \gets \Unknown_M$\
\State 
$\res\gets\emptyset$ \Comment{Indices of bits unknown to $M$ \& their values}
\State $\BYZ\gets\emptyset$ \Comment{Machines blacklisted by $M$ as Byzantine}
\State $c\gets \PARAM / \goodfrac$
\Comment{The parameter $\PARAM$ will be fixed later.}
\State 
$\alpha\gets \frac{(1+\epsilon)\byzfrac}{(1-\epsilon) (1-2\byzfrac)}$
\Comment{shrinkage factor, $\alpha < 1$. The parameter $\epsilon$ will be fixed later.}
\State $J_0 \gets \lceil\log_{1/\alpha} \frac{k}{c\log n}\rceil$ \Comment{Number of phases}
\For{ $J=0,1,2,\ldots, J_0-1$
{\bf (sequentially)} }
\State Invoke $\CommitteeWork$
\State Invoke $\Gossip$
\State Invoke $\KTAList$
\State Invoke $\CollectRequests$.
\EndFor
\For{every $i\in \Unknown_M$}
$\res[i] \gets \CloudQuery(i)$ \Comment{Querying the remaining unknown bits}
\EndFor
\State\Return$\res$
\end{algorithmic}
\end{algorithm}


\begin{algorithm}
\caption{Procedure $\CommitteeWork$, code for machine $M$}
\label{alg:committee-work}
\begin{algorithmic}[1]
\State $\hINDEX_M \gets \emptyset$\Comment{Set of indices whose committees $M$ joins}
\State $DV_M \gets \emptyset$\Comment{Set of \emph{directly} (cloud / comm) \emph{verified} indices}
\For{every $i=1,\ldots,n$ sequentially}
\Comment{Setting up committees}
\If{$i\in \INDEX_M$}
\State Join the private committee $\cC_i$ at random with probability
$\displaystyle p~=~\frac{c\log n}{\alpha^Jk}.$
\If{$M$ was selected to $\cC_i$}
\State $\hINDEX_M \gets \hINDEX_M \cup \{i\}$.
\If{$i\in \Unknown_M$ }
\State $\res[i] \gets \CloudQuery(i)$ and add $i$ to $DV_M$\Comment{Cloud-verification}
\EndIf
\Comment{If $i \not\in \Unknown_M$ then $M$ already knows it, i.e., $\res[i]=x_i$.}
\State Send the message $(i,\res[i])$ to every other machine.
\EndIf
\State Collect messages sent by members of $\cC_i$.\Comment{Ignore messages on bits $i\not\in \INDEX_M$.}
\EndIf
\EndFor
\State $\rho\gets (1-\epsilon)\PARAM\log n/\alpha^J$
\Comment{Also $\rho = (1-\epsilon)p\goodfrac k$}
\State 
\Comment{Blacklisting ``over-active'' Byzantine machines}
\State $\Wmax \gets (1+\epsilon)c\log n \cdot \frac{n}{k}$
\For{every other machine $M'$}
\State $\Work(M') \gets$ number of indices $i$ for which $M'$ reports belonging to $\cC_i$.
\If{$\Work(M') > \Wmax$} 
\State add $M'$ to $\BYZ$ and erase all its reports as committee member.
\EndIf
\EndFor
\For{every $i\in \INDEX_M$}
\State Let $\cC_i^M$ be the remaining
\emph{reduced committee}.\Comment{Possibly $\cC_i^M \ne \cC_i^{M'}$ for   $M \ne M'$.}
\EndFor
\For{every $i \in \Unknown_M$}\Comment{comm-verification}
\For{$b\in\{0,1\}$}
\State $\psi_b(i)\gets$ number of messages from $\cC_i^M$ members saying $x_i=b$.
\EndFor
\If{$\psi_0(i) \ge \rho$ and $\psi_1(i) < \rho$}
\State $\res[i]\gets 0$ and add $i$ to $DV_M$ 
\EndIf
\If{$\psi_1(i) \ge \rho$ and $\psi_0(i) < \rho$}
\State $\res[i]\gets 1$ and add $i$ to $DV_M$\Comment{If both $\psi_0(i) \ge \rho$ and $\psi_1(i) \ge \rho$, then
$i$ remains unknown} 
\EndIf
\EndFor
\State $\Known_M \gets \Known_M \cup DV_M$
\State $\Unknown_M \gets \Unknown_M \setminus DV_M$
\end{algorithmic}
\end{algorithm}

\begin{algorithm}
\caption{Procedure $\Gossip$, code for machine $M$}
\label{alg:gossip}
\begin{algorithmic}[1]
\For{every $i\in \Known_M$}
\State send the message $(i,\res[i])$ to all other machines.
\EndFor
\State Receive a list $\Known_{M'}$ from every other machine $M'$.
\State $GV_M\gets\emptyset$\Comment{Set of \emph{gossip-verified}  bits}
\For{every $i \in \Unknown_M$}
\State $\varphi_0(i)\gets |\{ M' \mid (i,0)\in \Known_{M'} \}|$.
\State $\varphi_1(i)\gets |\{ M' \mid (i,1)\in \Known_{M'} \}|$.
\If{$\varphi_0(i) \ge \byzfrac k + 1$ or $\varphi_1(i) \ge \byzfrac k + 1$} \State $GV_M\gets GV_M \cup \{i\}$
\EndIf
\EndFor
\State $\Known_M \gets \Known_M \cup GV_M$
\State $\Unknown_M \gets \Unknown_M \setminus GV_M$
\end{algorithmic}
\end{algorithm}

\clearpage
\begin{algorithm}
\caption{Procedure $\KTAList$, code for machine $M$}
\label{alg:KTA-list}
\begin{algorithmic}[1]
\For{every $i\in \Known_M$}
\State send the message $(i,\res[i])$ to all other machines.
\EndFor
\State Receive a list $\Known_{M'}$ from every other machine $M'$.
\For{every $i=1,\ldots,n$}
\State $\varphi_0(i)\gets |\{ M' \mid (i,0)\in \Known_{M'} \}|$.
\State $\varphi_1(i)\gets |\{ M' \mid (i,1)\in \Known_{M'} \}|$.
\If{$\varphi_0(i) \ge 2\byzfrac k + 1$ or $\varphi_1(i) \ge 2\byzfrac k + 1$}
\State $\KTA_M \gets \KTA_M \cup \{i\}$
\EndIf
\If{$\varphi_0(i) \ge \byzfrac k + 1$ or $\varphi_1(i) \ge \byzfrac k + 1$}
\State $\Known_M \gets \Known_M \cup \{i\}$
\State $\Unknown_M \gets \Unknown_M \setminus \{i\}$
\EndIf
\EndFor
\end{algorithmic}
\end{algorithm}

\begin{algorithm}
\caption{Procedure $\CollectRequests$, code for machine $M$}
\label{alg:collect-requests}
\begin{algorithmic}[1]
\State Set $\INDEX_M \gets \Unknown_M$
\State Send $\Unknown_M$ to all other machines.
\State Collect lists $\Unknown_{M'}$ from all other machines $M'$.
\For{every $i=1,\ldots,n$}
\State $R_U(i)\gets \{M' \mid i \in \Unknown_{M'} \}$.
\If{$i\in\KTA_M$} 
$\BYZ \gets \BYZ \cup R_U(i)$\Comment{Blacklisting for requesting known-to-all bits}
\EndIf
\EndFor
\State $\INDEX_M \gets \INDEX_M \cup \bigcup_{M'\not\in \BYZ} \Unknown_{M'}$\Comment{Indices to be learned, including $\Unknown_M$ of $M$ itself}
\end{algorithmic}
\end{algorithm}



\inline Remark: The communication performed in the various steps
of the algorithm take more than one time unit in the CONGEST model.
Hence, the algorithm must also ensure that the different steps are synchronized
and each step starts only after the previous step is completed by all machines.
Relying solely on reports by each machine concerning its progress might lead
to deadlocks caused by the Byzantine machines. Hence, the scheduling must be
based on the fact that the duration of each step is upper-bounded by
the maximum amount of communication the step involves.
We omit this aspect from the description of the algorithm.

\subsubsection*{Analysis}

\inline Sanity checks:
Let us start with the two sanity checks needed for ensuring the validity
of the random selection step and the convergence of the algorithm.

\begin{observation}
For $\byzfrac$ and $\epsilon$ satisfying
\begin{equation}
\label{eq:relate beta-epsilon-alpha}
\byzfrac < \frac{1-\epsilon}{3-\epsilon}~,
\end{equation}
(a) the chosen shrinkage factor satisfies $\alpha<1$,
and
\\
(b) the chosen probability satisfies $p<1$ for every $0\le J\le J_0-1$.
\end{observation}

\noindent\emph{Proof.}
Claim (a) follows directly from Eq. \eqref{eq:relate beta-epsilon-alpha} and the definition of $\alpha$.

The probability $p$ increases with $J$, so we only need to verify (b) for phase $J_0-1$. Indeed, 
$$\frac{c\log n}{\alpha^{J_0-1}k} 
~=~ \frac{c\log n}{k}\cdot\left(\frac{1}{\alpha}\right)^{J_0-1}
~<~ \frac{c\log n}{k}\cdot\left(\frac{1}{\alpha}\right)^{\log_{1/\alpha} \frac{k}{c\log n}} ~=~ 1.
\mbox{\hskip 3cm} \Box$$

\inline Progress tracking variables:
Next, we define notation for the values of the main variables of the algorithm
during the different phases.

\begin{itemize}
\item
Denote by $\Known_M^J$ (respectively, $\Unknown_M^J$)
the value of the set $\Known_M$ (resp., $\Unknown_M$)
at the beginning of phase $J$.
(Note that it is also the value of $\Known_M$ at the end of phase $J-1$)
\item
Denote by $\Known_M^{J,mid}$ (resp., $\Unknown_M^{J,mid}$)
the value of the set $\Known_M$ (resp., $\Unknown_M$)
at the end of procedure $\Gossip$ of phase $J$.
\item
Denote by $\INDEX_M^J$
the value of the set $\INDEX_M$
at the beginning of phase $J$.
\item
Denote by $\KTA_M^J$ the value of the set $\KTA_M$
at the end of Procedure $\KTAList$ of phase $J$.
\end{itemize}

Note that a bit $x_i$ can be unknown for $M$
and known for $M'$ for two honest machines $M$ and $M'$.
We say that $x_i$ is \emph{unknown} in phase $J$,
and the committee $\cC_i$ is \emph{necessary}, if 
$i\in \Unknown_M^J$ 
for \emph{some} honest machine $M$,
or equivalently, if $i \in \Unknown^J$, where
$$\Unknown^J=\bigcup_{M\in\cH} \Unknown_M^J$$
is the set of indices $i$ for which some honest machines request setting up a committee $\cC_i$ and querying in the current phase.
A bit $x_i$ is \emph{known} once $i\in \Known_M$ for every honest machine $M$.
Also let
\begin{eqnarray*}
\Unknown^{J,mid} &=& \bigcup_{M\in\cH} \Unknown_M^{J,mid} ~,
\\
\NKTA_M^J &=& \{1,\ldots,n\} \setminus \KTA_M^J.
\end{eqnarray*}

When $M$ joins (in Procedure $\CommitteeWork$) the committee $\cC_i$ for some $i \in \Unknown_M^J$,
$M$ is required to \emph{actively} query the cloud for the value of $x_i$.
We then say that $\cC_i$ is an \emph{active committee} for $M$.
(In contrast, when $M$ joins a committee $\cC_i$ for $i \in \Known_M^J$, it costs it nothing, since it already has the value of $x_i$ stored in $\res[i]$, so it does not need to spend another query.)
We define the following size variables.
\begin{itemize}
\item
Let $\tn_M^J$ denote the number of \emph{active} committees for $M$ in phase $J$.
\item
Let $\hn_M^J=|\hINDEX_M^J|$ denote the total number of committees that $M$ joins
by Procedure $\CommitteeWork$ in phase $J$. (Note that $\tn_M^J \le \hn_M^J$)
\item
Let $\alln_M^J=|\INDEX_M^J|$ denote the total number of requests received by $M$
by Procedure $\CollectRequests$ in phase $J$. 
\end{itemize}

\inline Bad events:
In an execution $\xi$ of the algorithm, there are
two types of \emph{bad events}, whose occurrence might fail the algorithm.
Our analysis is based on bounding the probability of bad events,
showing that with high probability, no bad events will occur in the execution,
and then proving that in a \emph{clean} execution, where none of the bad events occured, the algorithm succeeds with certainty.

\inline Bad event $\event_1(J,i)$:
In phase $J$, the committee $\cC_i$ selected for an unknown bit $x_i$
is not $\rho$-representative, for
$\displaystyle \rho = \frac{(1-\epsilon)\PARAM\log n}{\alpha^J}~$.
(If $x_i$ is already known, then this bad event does not affect
the correctness or query complexity of the honest machines, although
it might increase the time and message complexity.)

\inline Bad event $\event_2(J,M)$:
In phase $J$, an honest machine $M$ has $\hn_M^J > \Wmax$, namely, $M$ joins more than
$\displaystyle \Wmax = \frac{(1+\epsilon)cn\log n}{k}$ committees,
and subsequently gets blacklisted.

One bad event that we need not worry about is that of an honest machine mis-classifying a bit as \emph{known-to-all}. This is due to the following lemma.
 
\begin{lemma}
\label{lem:CK bits are KTA}
If some honest $M$ adds $i$ to its set $\KTA_M^J$ of known-to-all bits at the end of Procedure $\KTAList$ of phase $J$, then $i\in \Known_{M'}^{J+1}$ for every honest $M'$.
\end{lemma}

\inline Note: the sets $\KTA_M$ might not be all equal, namely, every honest machine might be aware of a different subset of the known-to-all bits. Note, however, that as shown later in Lemma \ref{obs:Core is CK}, the sets $\KTA_M$ of all honest machines contain the set $\CORE$ discussed in the high-level overview, and the fast growth of $\CORE$ is essentially the cause for the fast shrinkage of the set of unknown bits.

\begin{proof}
Suppose $i\in \KTA_M^J$ for some honest $M$. Then in Procedure $\KTAList$  of phase $J$, $M$ counted at least $2\byzfrac k+1$ messages containing $(i,b)$ (for $b\in\{0,1\}$). 
At least $\byzfrac k+1$ of these messages were sent by honest machines, and therefore in Procedure $\KTAList$ of phase $J$, all honest machines will count at least $\byzfrac k+1$ messages containing $(i,b)$. Consequently, every honest machine $M'$ will move $i$ to $\Known_{M'}$ at that step, so $i\in \Known_{M'}^{J+1}$.
\end{proof}

\inline Properties of clean executions:
For an integer $J\ge 0$, call the execution \emph{$J$-clean} if none of the bad events $\event_1(j,i)$ or $\event_2(j,M)$ 
occurred in it for $0\le j\le J$.
Clearly, an unclean execution of the algorithm might fail.
However, we show the probability that this happens is very small.
First, we establish some desirable properties of $J$-clean executions.

\begin{observation}
\label{lem:rho-rep reduced comm}
In a $J$-clean execution, if $i\in \Unknown^{J}$
(i.e., $x_i$ is still unknown
in phase $J$), then for every honest machine $M$,
the reduced committee $\cC_i^M$ is $\rho$-representative.
\end{observation}

\begin{proof}
Observe that
(a) in the absence of bad events of type $\event_1$,
the committee $\cC_i$ is $\rho$-representative, and
(b) no honest machine was erased from $\cC_i$ when creating $\cC_i^M$
due to the absence of bad events of type $\event_2$.
\end{proof}

\inline Remark: Note that once a committee is selected, the adversary can corrupt all of its members in the very next round. By then, however, the committee has completed its querying and communication actions, so the fact that it is no longer representative does not harm the execution. 

Note also that the need to complete all committee actions in a single round is the reason why it is required to perform the querying sequentially, spending a round for each bit $x_i$. The querying operations of all committees could in principle be parallelized, but the subsequent communication step might require more than a single round in the CONGEST model, giving the adversary an opportunity to intervene and corrupt an entire committee before it has completed sending its messages. 

Note that those bits that were not moved from $\Unknown_M$ to $\Known_M$
during the main phases $J$ of the algorithm were cloud-verified
in the final step of the algorithm. This implies the following.

\begin{observation}
\label{obs:all bits learned}
By the end of the execution, every honest machine has the value $\res[i]$
for every bit $x_i$.
\end{observation}
It remains to show that for every $x_i$, the $\res[i]$ value
obtained by each honest machine is correct. 

\begin{lemma}
\label{lem:correct-bits-J}
In a $J$-clean execution, whenever an honest machine learns an input bit $x_i$
in phases $0$ to $J$, the learned value $\res[i]$ is correct.
\end{lemma}
\begin{proof}

Consider an input bit $x_i$. Order the honest machines that learned $x_i$
during phases 0 to $J$ according to the time
by which they acquired $x_i$. the proof is  by induction on this order. 

For the induction basis, note that the first machine to acquire $x_i$
must have cloud-verified it, so the value it has obtained is clearly correct.

Now consider the $t$th machine $M$ in this order, and suppose $M$ knows that
$x_i=b$. there are several cases to consider.
\inline Case 1:
$M$ cloud-verified $x_i$,
either on the last step of the algorithm or in Procedure $\CommitteeWork$ during some phase $J$. Then again $\res[i]$ is clearly correct.
\inline Case 2:
$M$ comm-verifies $x_i$, in Procedure $\CommitteeWork$.
Then $M$ found $\psi_b(i) \ge \rho$ and $\psi_{1-b}(i) < \rho$.
Since the execution is $J$-clean, $\cC_i^M$ is $\rho$-representative
by Lemma \ref{lem:rho-rep reduced comm}.
This implies that if $x_i=1-b$ then all the honest machines in $\cC_i^M$
would return $1-b$, and $M$ would find $\psi_{1-b}(i)\ge\rho$,
which did not happen. Hence, $x_i=b$.
\inline Case 3:
$M$ gossip-verifies $x_i$, in the $\Gossip$ step or $\KTAList$ Procedure.
Then $M$ has received messages from $\byzfrac k +1$ or more machines
stating that they already know that $x_i=b$.
At least one of those machines, $M'$, is honest,
and it acquired $x_i$ prior to $M$, hence the inductive hypothesis
applies to it, yielding that indeed $x_i=b$.
\end{proof}



\begin{lemma}
\label{lem:structure}
In a $J$-clean execution, assuming $\byzfrac<1/3$, for every honest $M$,
$$\INDEX_M^{J+1} ~\subseteq~ \NKTA_M^{J} ~\subseteq~ \Unknown^{J,mid} ~\subseteq~ \Unknown^J ~\subseteq~ \INDEX_M^J.$$
\end{lemma}

\begin{proof}
Consider a bit index $i \notin \NKTA_M^J$. Then $x_i$ is marked known-to-all by $M$ in Procedure $\KTAList$
Consequently, $M$ ignores $x_i$ in phase $J$ even if it receives it in some request message in Procedure $\CollectRequests$. Hence $i \notin \INDEX_M^{J+1}$. The first containment follows

Consider a bit index $i \in \NKTA_M^{J}$.
Then $x_i$ is not listed as known-to-all in $M$, i.e., $i\notin\KTA_M$, so $M$ had $\varphi_0(i) \le 2\byzfrac k$ and $\varphi_1(i) \le 2\byzfrac k$.

Let $b = x_i$ i.e the correct value of $x_i$ and let $0 \le \delta \le 1$ be a the fraction of faulty machines that reported knowing $i$.
Since the execution is $J$-clean,by Lemma \ref{lem:correct-bits-J}, we know that $\varphi_{1-b}(i) \le \delta \byzfrac k $.
Therefore $\varphi_0(i) + \varphi_1(i) \le (2+\delta)\byzfrac k$.
Hence the number of machines that informed $M$ that they do not know $x_i$
satisfies $k - (\varphi_0(i) + \varphi_1(i)) \ge (1-(2+\delta)\byzfrac)k > (1-\delta)\byzfrac k$,
where the second inequality follows since $\byzfrac < 1/3$. 

Hence, there is at least one honest machine $M'$ that did not send $x_i$
as part of its $\Known_{M'}^{J,mid}$, so 
$i\in \Unknown_{M'}^{J,mid}$,
and hence
$i\in \Unknown^{J,mid}$.
The second containment follows.

The next containment follow from the fact that for an honest machine $M$ , $\Unknown_M$ is monotone decreasing in time.

Consider an index $i \in \Unknown^{J}$. Then some honest $M' \in \cH$ has $i \in \Unknown_{M'}^{J}$. This has two implications when $J \ge 1$.
First, $M'$ will send a request to learn $i$ in Procedure $\CollectRequests$ of phase $J-1$. Second, by Lemma \ref{lem:CK bits are KTA} $i \notin \KTA_M^{J-1}$(otherwise $i \in \Known_{M'}^J$). Hence $M$ will respect the request by $M'$, and add $i$ to $I_M^J$.
When $J=0$, $\Unknown^J = \{1,\dots, n\} = \Unknown_M^J = \INDEX_M^J$.
The fourth containment follows.
\end{proof}
We remark that if $x_i$ is known, hence $\cC_i$ is not necessary, then the
inviting machine is Byzantine, so it may invite only a few honest machines (or none), hence the constructed $\cC_i$ is not guaranteed to be $\rho$-representative,
but this will not hurt any honest machine,
since in this case the honest machines already know $x_i$
and will not listen to the committee.)

Define the \emph{core of \TwoCK} after phase $J$ as follows.
For every index $i$, let $num_{DV}^J(i)$ denote the number of honest machines $M$
that comm-verified $i$ and added it to $DV_M$ in Procedure $\CommitteeWork$ of phase $J$. Then 
$$\CORE^J ~=~ \{ i \mid  num_{DV}^J(i) \ge \byzfrac k+1 \}.$$
The name is justified by the following lemma.

\begin{lemma}
 \label{obs:Core is CK}
For every $i\in\CORE^J$, $i \in \Known_M^{J,mid}$ and $i \in \KTA_M^{J}$, for every honest machine $M$ 
\end{lemma}

\begin{proof}
Consider an index $i\in\CORE^J$. By definition, $x_i$ was comm-verified by
at least $\byzfrac k+1$ honest machines during Procedure $\CommitteeWork$ of phase $J$. Each of these machines will send $i$ (along with its value) to every other machine during the $\Gossip$ round. Subsequently, at the end of this round, 
$i\in\Known_M^{J, mid}$ for every honest $M$.
Consequently, in procedure $\KTAList$ of phase $J$, \emph{all} honest machines will report knowing $x_i$ , so every honest machine $M$ will add it to $\KTA_M^{J}$
\end{proof}

\begin{lemma}
\label{obs:ID-opt-3}
In a $J$-clean execution, for every $J\ge 0$ and every honest $M$,
\\
(1) $|\Unknown^{J, mid}|\le \alpha^{J+1} n$, 
\\
(2) $\alln_M^J \le \alpha^{J} n$,
\\
(3) $|\Unknown^J|\le \alpha^J n$.
\\
(4) $|\Unknown_M^J|\le \alpha^J n$.
\end{lemma}
\begin{proof}
We first prove part (1), 
by considering iteration $J\ge 0$ and bounding $|\Unknown^{J,end}|$ at its end.

The purpose of blacklisting Byzantine machines that claim to participate
in too many committees, via defining reduced committees, is to curb
the influence of the Byzantine machines on votes,
by bounding the extent of \emph{Byzantine infiltration} into committees.
For every honest machine $M$ and Byzantine machine $M'$,
denote by $\BI_M(M')$ the number of reduced committees $\cC_i^M$ that $M'$
claimed to belong to. (Note that for machines $M'$ that were not blacklisted,
this value is the same as $\Work(M')$.) 
Denote the total number of Byzantine infiltrations
into reduced committees of $M$ by $\BI_M = \sum_{M'\in\BYZ} \BI_M(M')$.
Denote the total number of Byzantine infiltrations into reduced committees
of honest machines by $\BI = \sum_{M\in\cH} \BI_M$.
By the way $M$ constructs the reduced committees in Procedure $\CommitteeWork$,
every machine appears in at most $\Wmax$ reduced committees of $M$, hence
$\BI_M \le \byzfrac k \cdot \Wmax$, and therefore
$$\BI ~\le~ \goodfrac k\cdot\BI_M ~\le~ \goodfrac k\cdot \byzfrac k \cdot \Wmax
~=~ (1+\epsilon)c\byzfrac\goodfrac\cdot kn\log n.$$
Consider a bit $x_i \in \Unknown^J$. By the fourth containment of Lemma
\ref{lem:structure}, $x_i \in \INDEX_M^{J}$ for every honest machine $M$.
Hence every honest $M$ will set up a committee $\cC_i$, which will be
$\rho$-representative since the execution is $J$-clean.

A necessary condition for $x_i$ to remain in $\Unknown^{J, mid}$ is that at most
$\byzfrac k$ honest machines $M$ add it to their set $DV_M$ of cloud- or
comm-verified bits in Procedure $\CommitteeWork$ of phase $J$. 
This is because otherwise,$i\in\CORE^{J}$ and by lemma \ref{obs:Core is CK},
 it will belong to $\Known_M^{J,mid}$ for every honest $M$.

Hence, in order to keep $i$ in $\Unknown^{J, mid}$,
the adversary must prevent at least $(1-2\byzfrac)k$ honest machines from
cloud- or comm-verifying $x_i$.
To achieve that, at least $\rho$ Byzantine machines
must infiltrate the reduced committee $\cC_i^M$ for at least $(1-2\byzfrac)k$
honest machines. This incurs at least $(1-2\byzfrac)k\rho$ work.
Hence, the number of bits $x_i$ for which this can happen is at most 
$$|\Unknown^{J,mid}| ~\le~ \frac{\BI}{(1-2\byzfrac)k\rho} 
~\le~ \frac{(1+\epsilon)c\byzfrac\goodfrac \cdot kn\log n}{(1-2\byzfrac)k\cdot
(1-\epsilon)\PARAM\log n/\alpha^{J}}
~=~ \frac{(1+\epsilon)\byzfrac\PARAM \cdot \alpha^{J} n}{(1-2\byzfrac)(1-\epsilon)Z}
~=~ \alpha\cdot\alpha^{J} n~,$$
where the last equality is by the definition of $\alpha$.
This yields Part (1). 

Part (2) follows from part (1) upon noting that
$\alln_M^J = |\INDEX_M^J| \le |\Unknown^{J-1,mid}|$ by Lemma \ref{lem:structure}.
Part (3) follow from part (2) by Lemma \ref{lem:structure} again.
Part (4) follows from part (3), noting that $\Unknown_M^J \subseteq \Unknown^J$.
\end{proof}
\inline High probability of clean executions:
We now argue that the probability for the occurrence of any of the bad events
is low.


\begin{lemma}
\label{lem:bad events type 1} 
For any $J\ge 0$, if the execution $\xi$ is $(J-1)$-clean,
and the parameters $\epsilon$ and $\PARAM$ satisfy
\begin{equation}
\label{eq:relate PARAM-epsilon-EV1}
\epsilon^2\PARAM/2 \ge 2+\lambda
\end{equation}
for some constant $\lambda>0$,
then the probability that any of the bad events $\event_1(J,i)$
occurred in $\xi$ is at most 
$O(\frac{1}{n^{1+\lambda}})$.
\end{lemma}

\begin{proof}
We first show that for every bit $x_i$, $\Prob[\event_1(J,i)] \le 1/n^{2+\lambda}$.

Consider an index $i\in \Unknown^J$. 
By Lemma \ref{lem:structure}, $\Unknown^J \subseteq \INDEX_M^{J}$, and hence $i \in \INDEX_M^{J}$, 

Therefore, all honest machines join the committee $\cC_i$
with probability $p$. Hence, denoting the number of honest machines in $\cC_i$
by $X$, 
$$\Exp[X] ~=~ p|\cH| ~\ge~ p\cdot \goodfrac k ~=~ \frac{\goodfrac c\log n}{\alpha^J} = \frac{\PARAM\log n}{\alpha^J}.$$

\begin{equation}
\label{eq:chernoff1}
\Prob[\event_1(J,i)] ~=~ \Prob[X<\rho] ~=~ \Prob[X\le (1-\epsilon) \cdot \PARAM\log n/\alpha^J]
\le \Prob[X\le (1-\epsilon) \Exp[X]]
\end{equation}
By Chernoff's bound,
\begin{equation}
    \Prob[X\le (1-\epsilon) \Exp[X]] \le \exp \left( -\frac{\epsilon^2 \Exp[X]}{2} \right)
    \le \exp\left(-\frac{\epsilon^2}{2} \cdot \frac{\PARAM\log n}{\alpha^J}\right),
\end{equation}
and by Eq. \eqref{eq:relate PARAM-epsilon-EV1} it follows that 
\begin{eqnarray*}
\Prob[\event_1(J,i)] 
~\le~ \exp\left(-(2+\lambda)\cdot\frac{\log n}{\alpha^J}\right) ~\le~ n^{-2-\lambda}.
\end{eqnarray*}
By the union bound, the probability that \emph{any} bad event of type $\event_1$
occurred in the execution is at most $O(\frac{1}{n^{1+\lambda}})$.
\end{proof}

\begin{lemma}
\label{lem:bad events type 2} 
For any $J\ge 0$, if the execution $\xi$ is $(J-1)$-clean, and the parameter $\PARAM$ satisfies
\begin{equation}
\label{eq:relate PARAM-EV2}
\frac{\epsilon^2}{2+\epsilon}\cdot \PARAM \ge 2+\lambda
\end{equation}
for some constant $\lambda>0$,
then the probability that any of the bad events $\event_2(J,M)$
occurred in $\xi$ is at most $O(\frac{1}{n^{1+\lambda}})$.
\end{lemma}

\begin{proof}
We first show that for every honest machine $M$,
$\Prob[\event_2(J,M)] \le 1/n^{8/3}$.
The bad event $\event_2(J,M)$ occurs if $\hn_M^J >\Wmax$ in phase $J$.
In Procedure $\CommitteeWork$, $M$ tries (randomly) to join the committee $\cC_i$ for every
$x_i\in \INDEX_M^J$, hence
$\Exp[\hn_M^J] = p \alln_M^J$.
Applying Lemma \ref{obs:ID-opt-3}(2), we get that
$$
\Exp[\hn_M^J] ~\le~ p\alpha^{J} n
$$
We introduce a variable $X \in (0, 1]$ such that
$$\Exp[\hn_M^J] ~=~ X\cdot p\alpha^{J} n ~=~ X \cdot c \log n \cdot \frac{n}{k}
~=~ \frac{X\cdot \Wmax}{1+\epsilon}.$$
We can see now that
\begin{eqnarray*}
\Prob[\event_2(J,M)] &=& \Prob[\hn_M^J >\Wmax]
~\le~ \Prob\left[\hn_M^J >\frac{1+\epsilon}{X} \cdot \Exp[\hn_M^J]\right]
\end{eqnarray*}
Using the variation of Chernoff's bound that says that, for $\delta > 0$,
\begin{eqnarray*}
\Prob\left[A > (1 + \delta) \Exp[A]\right] \le \exp \left (-\frac{\delta^2}{2+\delta} \cdot \Exp[A] \right)
\end{eqnarray*}
and setting $\delta = \frac{1+\epsilon}{X} -1$, we get
\begin{eqnarray*}
\Prob[\event_2(J,M)] &\le&
\exp \left (-\frac{(\frac{1+\epsilon-X}{X})^2}{2+\frac{1+\epsilon}{X} -1} \cdot \Exp[\hn_M^J] \right)
~=~ \exp \left (-\frac{(1+\epsilon-X)^2}{X^2(\frac{1+\epsilon}{X} +1)} \cdot X c\log n \cdot \frac{n}{k} \right)
\\ &=& 
\exp \left (-\frac{(1+\epsilon-X)^2}{X+1+\epsilon} \cdot c\log n \cdot \frac{n}{k} \right)
~=~ \exp \left (-f(X) \cdot c\log n \cdot \frac{n}{k} \right),
\end{eqnarray*}
where $f(x) = \frac{(1+\epsilon-x)^2}{x+1+\epsilon}$. 
It is easily verifiable that $f(x)$ is monotone decreasing in the range $[0,1]$, attaining a minimum value of $\frac{\epsilon^2}{2+\epsilon}$, i.e, $f(x) \ge \epsilon^2/2+\epsilon$ for every $x \in [0,1]$. Therefore, we get
\begin{eqnarray*}
\Prob\left[\hn_M^J >\frac{1+\epsilon}{X} \cdot \Exp[\hn_M^J]\right]
&\le& \exp\left(-\frac{\epsilon^2}{2+\epsilon}\cdot c\log n \cdot \frac{n}{k}\right)
~\le~ n^{-c \epsilon^2 / (2+\epsilon)}~\le~ n^{-\PARAM \epsilon^2 / (2+\epsilon)}
~\le~ \frac{1}{n^{2+\lambda}}~,
\end{eqnarray*}
where the last inequality follows by Eq. \eqref{eq:relate PARAM-EV2}.
The lemma now follows by the union bound.
\end{proof}

Call an execution \emph{clean} if none of the bad events of type $\event_1$ or $\event_2$ occurred in it. The above three lemmas yield the following:

\begin{corollary}
\label{cor:clean whp}
Consider an execution $\xi$. 
If the parameters $\epsilon$ and $\PARAM$ satisfy
\begin{equation}
\label{eq:relate PARAM-epsilon-EV1+2}
\PARAM\cdot\min\{\epsilon^2/2~,~\epsilon^2/(2+\epsilon)\} ~\ge~ 2+\lambda 
\end{equation}
for some constant $\lambda>0$,
then the probability that $\xi$ is clean is at least 
$1-O(\frac{\log n}{n^{1+\lambda}})$.
\end{corollary}

\begin{proof}
By induction on $J$, we prove that the probability that $\xi$ is $J$-clean 
is at least $1-\frac{2J}{n^{1+\lambda}}$. in other words, we show that
the probability that \emph{any} of the events $\event_1(j,i)$ or $\event_2(j,M)$
occurred in it for any $0\le j\le J$ is at most
$\frac{J}{n^{1+\lambda}}$.

For $J=0$, $\xi$ is vacuously $(J-1)$-clean, so the claim holds
by Lemmas \ref{lem:bad events type 1} and \ref {lem:bad events type 2}. 

Now suppose the claim holds for $J-1\ge 0$ and consider phase $J$.
By the inductive hypothesis, the overall probability for all the events $\event_1(j,i)$ and $\event_2(j,M)$ 
for any $0\le j\le J-1$ is at most 
$\frac{J-1}{n^{1+\lambda}}$.
So with probability at least $1-\frac{J-1}{n^{1+\lambda}}$, the execution is
$(J-1)$-clean. 
By Lemmas \ref{lem:bad events type 1} and \ref {lem:bad events type 2}, 
the overall probability for all the events $\event_1(J,i)$ and $\event_2(J,M)$ 
given that the execution is $(J-1)$-clean is at most
$\frac{1}{n^{1+\lambda}} \cdot (1-\frac{J-1}{n^{1+\lambda}}) \le \frac{1}{n^{1+\lambda}}$.
The claim now follows by the union bound.
\end{proof}

Hereafter we assume that the execution $\xi$ at hand is clean
(noting that the probability for this is at least $1-1/n$).

\inline Correctness:
Combining Obs. \ref{obs:all bits learned}, Lemma \ref{lem:correct-bits-J}
and Cor. \ref{cor:clean whp}, we get the following.

\begin{lemma}
\label{lem:correct-bits}
In a clean execution, every honest machine learns all the input bits
correctly.
\end{lemma}

Let us next justify the blacklisting done in Procedures $\CollectRequests$  and $\CommitteeWork$.

\begin{lemma}
\label{lem:blacklisted-byz-step1} 
Consider Procedure $\CollectRequests$ performed by machine $M$ in phase $J$.
If $M$ blacklists $M'$ as Byzantine and ignores its request to query $x_i$,
then this decision is justified, namely, $M'$ is indeed Byzantine,
and $x_i$ is indeed known already to all honest machines.
\end{lemma}

\begin{proof}
Consider the scenario described in the lemma.
$M$ ignores $x_i$ because it is marked as known-to-all.
This marking indicates that $M$ counted at least $2\byzfrac k+1$ received messages
containing $(i,b)$
(for $b\in\{0,1\}$) in Procedure $\KTAList$ of phase $J$.
At least $\byzfrac k+1$ of these messages were sent by honest machines,
and those senders have sent the same message to all other machines.
Therefore, every honest machine $M''$ has received at least $\byzfrac k+1$
messages containing $(i,b)$,
and subsequently moved $x_i$ to $\Known_{M''}^{J+1}$.
Hence $x_i$ is already known for all honest machines, and the requesting
machine $M'$ must be Byzantine.
\end{proof}

\begin{lemma}
In a clean execution, if machine $M'$ is blacklisted as Byzantine
by machine $M$ in Procedure $\CommitteeWork$, then $M'$ is indeed Byzantine.
\end{lemma}

\begin{proof}
Since the execution is clean, bad events $\event_2$ do not happen, 
so $\hn_M^J \le \Wmax$ for every honest machine $M$.
Hence, the fact that $M'$ claimed to more than $\Wmax$ committees implies that it is Byzantine.
\end{proof}

\inline Complexity:
We now discuss the complexities of the algorithm.

\begin{lemma}
In a clean execution, the 
query complexity of the algorithm
is $\Query=O(\frac{n}{k} \cdot \log^2 n) = {\tilde O}(n/k)$.
\end{lemma}
\begin{proof}
In phase $J$, an honest machine $M$ participates actively in $\tn_M^J$
committees, and hence performs this many queries.
In a clean execution, bad events $\event_2$ did not happen, 
so $\hn_M^J \le \Wmax$ for every honest machine $M$.
As $\tn_M^J \le \hn_M^J$, it follows that also $\tn_M^J \le 2c\log n \cdot n/k$ for every honest machine $M$.
Hence, over all $\log(n/k)$ phases, $M$ performs 
$O\left(\frac{n\log^2n}{\goodfrac k}\right) = {\tilde O}(n/k)$ queries.

In addition, by Lemma \ref{obs:ID-opt-3}(4), after performing phase 
$J_0-1=\lfloor\log_{1/\alpha}(\frac{k}{c\log n})\rfloor-1$, 
each honest machine is left with a set $\Unknown_M$ of at most ${\tilde O}(n/k)$
unknown bits,
which it then acquires using $\tilde{O}(n/k)$ 
additional queries.
\end{proof}

It remains to select suitable parameters $\epsilon$, $\PARAM$ and $\lambda$, satisfying the constraints of Eq. \eqref{eq:relate beta-epsilon-alpha} and \eqref{eq:relate PARAM-epsilon-EV1+2}.
Note that $g(\epsilon)=\frac{1-\epsilon}{3-\epsilon}$ is monotone decreasing for $0<\epsilon\le 1$, with $g(0)=1/3$, so satisfying Eq. \eqref{eq:relate beta-epsilon-alpha} when $\beta=1/3-\epsilon'$ for small $\epsilon'$ requires taking a sufficiently small $\epsilon$. Selecting a sufficiently large $\PARAM$ will satisfy Eq. \eqref{eq:relate PARAM-epsilon-EV1+2}.
For a concrete example, taking $\epsilon=1/10$ and $\PARAM=445$ will satisfy the constraints for $\beta \le 1/3-2/87$ and $\lambda=1/10$.

\begin{theorem}
\label{thm:ID harsh small beta + alg name}
In the 
$\Harsh(\byzfrac<1/3)$ model, Algorithm $\DownloadGossip$ w.h.p.
solves the $\IDp$ problem with 
$\Query = O\left(\frac{n\log^2 n}{\goodfrac k}\right)$, 
$\Time = O\left(n\log_{\frac{1}{\byzfrac}}\left(\frac{\goodfrac k}{\log n}\right)\right)$ and 
$\Message = O\left(nk^2\log_{\frac{1}{\byzfrac}}\left(\frac{\goodfrac k}{\log n}\right)\right)$. 
\end{theorem}

We remark that our focus was on optimizing query complexity. The $\Time$ and $\Message$ complexities can be improved further. For example the current algorithm requires the machines to send the entire set of known bits in each iteration, but clearly it suffices to send the updates. 


\subsection{The $\ORp$ Problem in the Harsh Model}



We now consider improved randomized algorithms for the $\ORp$ problem.

\subsubsection{Randomized algorithm in $\Harsh(\byzfrac<1)$}
We describe an efficient approach for the $\ORp$ that works for any $\byzfrac = 1-\varepsilon$ and an adaptive adversary. First we describe a useful primitive called \emph{verification spreading} below. Once this primitive is established, we can solve $\ORp$ in a straightforward manner. We say that the index $i$ is a \emph{verified index} for the machine $M$ if $M$ has already queried the cloud and discovered that $\INPUT[i]=1$. At the outset, every machine $M$ has a (possibly empty) set of verified indices,
$$\Known_M =\{i \mid \resM[i]=1\}$$


\textbf{Verification Spreading}
In the verification spreading problem, we want all honest machines to know some verified index given at least one honest machine $\Machine$ knows a verified index i.e. $\Known_M \ne \emptyset$. We provide Procedure $\Spread$ that solves the verification spreading problem.

\begin{algorithm}
\begin{algorithmic}[1]
    \For{$r=1,2\dots \lceil \log k \rceil$ phases (sequentially)} 
        \If{$\Known_M\ne\emptyset$}
        \State Pick some arbitrary index $b(M)\in \Known_M$.
        \State Send $b(M)$ to all machines. 
        \EndIf
        \State $H \gets \{M' \mid \mbox{a single index}~ b(M') ~\mbox{was received from}~ M'\}$.
        \State $S \gets \{b(M') \mid M' \in H\}$ \Comment{array of indices received} 
        \State{Sample a set $\INDEX$ of $\lceil \frac{8\ln n}{\goodfrac} \rceil$ indices uniformly at random from $S$}
        \State {$\res[i] \gets \CloudQuery(i)$ for every $i\in\INDEX$} 
        \State {$\Known_M \gets \Known_M \bigcup \{ i \mid \res[i]=1 \}$}
    \EndFor
\end{algorithmic}
\caption{Procedure $\Spread$, $\Harsh(\byzfrac<1)$ , Code for machine $M$}
\label{alg:infospread}
\end{algorithm}

Let us make the following two remarks about code efficiency.
First, note that once an honest machine has found a verified bit at some phase $r$, it does not need to continue going through the entire process and query the bits it receives in the following rounds.
Secondly, sampling indices \emph{without} replacement in the final step of the procedure may yield slightly superior performance (although it may be slightly harder to implement and analyze). In both cases, the worst-case analysis does not change.

\begin{theorem}
\label{lem:inf:spread}
In the $\Harsh(\byzfrac<1)$ model, Procedure $\Spread$  w.h.p. solves the Verification Spreading problem with $\Query = O(\frac{\log k \ln n}{\goodfrac})$, $\Time = O(\log k)$ and $\Message = O(k^2\log k)$
\end{theorem}

\begin{proof}
Let $f(r)$ denote the number of honest machines $M$ that have $\Known_M \neq \emptyset$ (i.e., at least one known index) at the end of phase $r$ of the execution. Let $\byzfrac_r$ denote the fraction of machines corrupted by the adversary at the end of phase $r$. Observe that $\byzfrac_r$ must be non-decreasing in $r$. 
    
The premise of the lemma is that $f(0) > 0$. To prove the lemma we need to show that $f(r) = k(1-\byzfrac_r)$ when $r = \lceil \log k \rceil$. 
    
We prove this in two stages. First we show that as long as $f(r) < \goodfrac k / 4$, it doubles in the next phase, i.e., $f(r + 1) \geq 2 f(r)$. These phases comprise the first stage. In the second stage, when $f(r) \geq \goodfrac k / 4$, then within a single phase all honest machines receive a correct index, i.e., $f(r + 1) = \goodfrac k$.

\textbf{Stage 1} ($f(r) < \goodfrac k / 4$).  Consider any set of $\goodfrac k - f(r)$ honest machines that have not yet seen a set bit at the beginning of phase $r + 1$. Partition them into $f(r)$ (almost) equal disjoint parts, each of size no less than $\goodfrac k / f(r) - 2$. Consider any one part $P$. Let $p$ be the probability that none of the machines in $P$ queried a set index. We have,
$$p ~\leq~ 
\left(1 - \frac{f(r)}{k}\right)^{(\frac{\goodfrac k}{f(r)} - 2) \cdot \frac{8 \ln n}{\goodfrac}} 
~\leq~ \left( 1 - \frac{f(r)}{k}\right)^{\frac{\goodfrac k}{2f(r)} \cdot \frac{8 \ln n}{\goodfrac}} 
~\leq~ e^{-4\ln n} 
~=~ n^{-4}, $$
where the last inequality uses $1 - x \leq e^{-x}$.

Now applying the union bound over all such parts $P$, we get that w.h.p, each of the $f(r)$ parts have at least one machine querying a set index. Therefore, at the end of phase $r+1$, $f(r)$ new machines have queried a set index.

\textbf{Stage 2} ($f(r) \geq \frac{\goodfrac k}{4}$). At this stage, of the $k$ possible messages $(M',b(M'))$ received by an honest machine in  phase $r+1$, at least $\goodfrac k / 4$ must be valid. Note that this is a constant fraction $(\goodfrac / 4)$. Sampling $8 \ln n / \goodfrac$ is more than sufficient to hit one of the valid bits w.h.p. The probability of failure is $\left(1-\goodfrac/4\right)^{8\ln n / \goodfrac} \leq n^{-2}$. Taking union bound over all honest machines, we get that each of them hits at least one verified bit (w.h.p.).

Finally, observe that after $r = \lceil \log k \rceil - 1$ phases, the execution must reach Stage $2$, thus $\lceil \log k \rceil$ phases are sufficient for termination.
$\Query=O(\frac{\log k \ln n}{\goodfrac})$ and $\Time=O(\log k)$ follow directly from the algorithm.
\end{proof}

\textbf{Algorithm for $\ORp$}

Let $\byzfrac < 1-\varepsilon$ for some $\varepsilon > 0$. We present an efficient algorithm for $\ORp$ whose query complexity is near optimal.

\begin{algorithm}
    \begin{algorithmic}[1]
        \For{$r=1,2,\dots \lceil \log{n} \rceil$} 
        \State $\Density_r \gets \frac{1}{2^r}$  \Comment{Guess for the value of $\Density$ at round $r$}  
        \State Sample a set $\INDEX$ of 
        $\lceil \frac{1}{\Density_r} \cdot \frac{1}{\goodfrac k} \cdot \ln n \rceil$ 
        indices uniformly at random from $\{1,2\ldots n\}$
        
        \State $\res[j]\gets \CloudQuery(j)$ for every $j\in \INDEX$
        \State $\Known_M \gets \{j \mid res[j]=1\}$

    
        \State $\Call{Spread}{ \Known_M}$
        \If{$\res[j]=1$ for some $j$} 
        ~~ \Return 1
        \EndIf
        \EndFor
        \State \Return 0.
    \end{algorithmic}
    \caption{$\RandomizedDisjunctionAlgo$, $\Harsh(\byzfrac<1)$, Code for machine $\Machine$}
    \label{alg:or:gamma}
\end{algorithm}

\begin{theorem}
\label{thm: OR Harsh + alg name}
In the $\Harsh(\byzfrac<1)$ model, Algorithm $\RandomizedDisjunctionAlgo$ w.h.p. solves the $\ORp$ problem with 
$\Query = O\left( \frac{\log n}{\goodfrac k}\cdot\InverseDensity + \frac{ \log k \log n \log(1/\ModifiedDensity)}{\goodfrac} \right)$, 
$\Time = O\left(\log k \cdot \log \InverseDensity\right)$ and $\Message = O(k^2\log k \cdot \log \InverseDensity)$.
\end{theorem}




\begin{proof}
We consider three cases.
    
\textbf{Case 1}: ($\Density = 0$). In this case, none of the honest agents would've seen a set bit till the end of the execution ($\log n$ phases).

\textbf{Case 2}: (No honest machine sees a set bit before phase $J=\lceil -\log \Density \rceil$). 
In this case, during phase $J$, the honest machines collectively query $\ln n\cdot\InverseDensity$ bits. The probability that none of them are set is at most $\left(1 - \Density \right)^{\ln n\cdot\InverseDensity} \leq 1/n$.

By Lemma \ref{lem:inf:spread}, by the beginning of phase $J+1$, all honest machines would've seen a set bit (w.h.p). The honest machines that directly queried for a set bit would've terminated at end of phase $J$. The remaining machines (that received correct index through communication) terminate at end of phase $J$.

The query complexity per machine consists of two parts. The queries due to Line 2 for a geometric series with highest term being $O(\frac{\ln n}{\goodfrac k}\cdot\InverseDensity)$. The queries due to Line 2 are fixed to $\frac{\ln n}{\goodfrac}$ per phase for at most $\lceil -\log{\Density} \rceil + 2$ phases.

\textbf{Case 3}: (Some honest machine sees a set bit at phase $r < J$). By Lemma \ref{lem:inf:spread}, all honest machines would've sampled a correct bit either during or before phase $r+1$. Regardless of what happens in phase $r+2$, by the end of phase $r+2$ all honest machines terminate.
\end{proof}

\section{\CCA\ Algorithms in the Benign Adversarial  Model}

\subsection{Primitives}

\subsubsection{Public (representative and majorizing) committees}
\label{sec:public:comm}

In this section we describe algorithms for electing public representative and public majorizing committees. In the $\Benign$ model, the cloud has the ability to generate global random bits. Therefore, a simple way of electing a public committee of size $\size$ is to query $O(\size\log n)$ global random bits which represent the IDs of the machines in the committee. If multiple committees are desired, say $\num$, then repeating the election for each committee can become expensive as the query complexity increases by a factor of $\num$. We therefore design more query efficient algorithms for electing public committees by using $c-$ wise independent family of hash functions.

\begin{defn}(c-wise Independent Hashing)
For $N, L, c \in \mathbb{N}$ such that $c \le N$, $H = \{h : [N] \rightarrow [L]\}$ is a family of \emph{$c$-wise independent Hash functions} if for all distinct $x_1,x_2,\dots x_c \in [N]$, whenever $h$ is chosen uniformly at random from $H$, the random variables $h(x_1), h(x_2), \dots h(x_c)$ are independent and uniformly distributed in $[L]$.
\end{defn}

\begin{lemma}
\label{lem:hashfamily}
For every $a,b,c$, there is a family of $c$-wise independent hash functions $H$ $=\{h : \{0, 1\}^a \rightarrow \{0, 1\}^b\}$ such that choosing a random function from $H$ takes $c \cdot \max\{a, b\}$ random bits, and evaluating a function from $H$ takes $\poly(a, b, c)$ computation steps.
\end{lemma}

The algorithm for electing $\num$ public committees of size $\size$ works as follows. Let $\cH=\{h:[\num]\rightarrow[k]\}$ be a $\size$-wise independent family of hash functions. The machines sample $\size$ hash functions $h_1,h_2,\ldots, h_{\size}$, which requires $\tilde{\Theta}(\size^2)$ global random bits. The $i^{th}$ committee $\cC_i$ is given by $\cC_i = \{h_j(i) \mid 1\leq j \leq \size\}$, for $i = 1, 2, \dots, \num$.

\begin{algorithm}
\begin{algorithmic}[1]
\Statex \textbf{Input}: Positive integers $\size, \num$, family of  $\size$-wise independent hash function $\cH=\{h:[\num]\rightarrow[k]\}$
\Statex \textbf{Output}: Return $\num$ committees $\cC_1, \cC_2, \dots, \cC_\num$ of size $\size$ each.
\State $\rand[j]\gets\CloudRG(j)$ for $j=1,2,\ldots,\max(\ceil{\log\num,\log k})$
\State Sample $h_1, h_2, \dots, h_\size$ independently from $\cH$ using $\rand$.
\State \Return $\cC_i = \{h_j(i) \mid j=1,2, \dots, \size \}$ for $i = 1, 2, \dots, \num$.
\end{algorithmic}
\caption{Procedure $\ElectPublic(\size, \num)$ , $\Benign(\byzfrac < 1)$, code for machine $\Machine$}
\label{alg:public:comm:election}
\end{algorithm}

The following lemmas summarize the algorithm's properties
(proofs are deferred to the appendix).

\begin{lemma}
\label{lem:pub:comm:num}
In the $\Benign(\byzfrac<1)$ model, for $\size \ge 9\log n$, at the end of procedure $\ElectPublic$, every machine belongs to at most $\frac{\size\num}{k}+\size^2$ committees w.h.p.
\end{lemma}

\def\ProofFourThree{
Consider machine $\Machine_j$. The number of committees that $\Machine_j$ belongs to is given by the number of solutions to $\sum\limits_{u=1}^{\size}\sum\limits_{v=1}^{\num}h_u(v)=j$. Let $X_{u,v}=\{1\mid h_u(v) = j\}$. We know that $p=\Prob(X_{u,v}=j)=\frac{1}{k}$. Applying Chernoff bounds with $\epsilon = \frac{\size k}{\num}$ and $\mu = \num p=\frac{\num}{k}$, we get $\Prob\left(\sum\limits_{v=1}^\num X_{u,v} > \frac{\num}{k} + \size\right) \le \exp{\left(\frac{-\size}{3}\right)}$. Since $h_1,h_2,\ldots, h_\size$ are independently chosen from $\cH$, $\Prob\left(\sum\limits_{u=1}^{\size}\sum\limits_{v=1}^\num X_{u,v} > \frac{\size\num}{k} + \size^2\right) \le \size\exp{\left(\frac{-\size}{3}\right)}$. When $\size \ge 9\log n$, $\Prob\left(\sum\limits_{u=1}^{\size}\sum\limits_{v=1}^\num X_{u,v} > \frac{\size\num}{k} + \size^2\right) \le \frac{1}{n^2}$. The lemma follows from a simple application of union bound over the $k$ machines.
} 

\begin{lemma}
\label{lem:pub:comm:maj}
In the $\Benign(\byzfrac<1/2)$ model, for $\size \ge \frac{2\log n}{(1/2-\byzfrac)^2}$ and $\num \in O(n)$, at the end of Procedure $\ElectPublic$, the committees formed are public majorizing committees w.h.p.
\end{lemma}

\def\ProofFourFour{
Consider committee $\cC_i$. Let $X_j=\{1\mid h_j(i)$ is honest$\}$ for $j=\{1,2,\ldots,\size\}$. $\Prob(X_j=1)=\frac{1}{2}-\epsilon$. 
We have $\mu = \Exp[\sum X_j]=\size(\frac{1}{2}-\varepsilon)$. Choosing $\epsilon=\frac{2\varepsilon}{1-2\varepsilon}$, we get that 
$$\Prob\left(\sum X_i > \size/2 \right) ~\leq~ \Prob\left(\sum X_i > (1 + \epsilon) \mu \right) ~\leq~ \exp{\left( -\frac{\varepsilon^2 \size}{1-\varepsilon} \right)} ~\leq~ \exp(-\varepsilon^2\size).$$
When $\size\ge\frac{2\log n}{\varepsilon^2}$, we have $\Prob\left(\sum X_i > \size/2 \right) \le \frac{1}{n^2}$. The lemma now follows by a simple application of union bound over the $\num$ committees.
} 

\begin{lemma}
\label{lem:pub:comm:weak}
In the $\Benign(\byzfrac<1)$ model, 
for $\size \ge \frac{2\log n}{\goodfrac}$ and $\num \in O(n)$, at the end of Procedure $\ElectPublic$, the committees formed are weak (\emph{1-representative public}) committees w.h.p.
\end{lemma}

\def\ProofFourFive{
Consider committee $\cC_i$. The probability that all machines in $\cC_i$ are bad is at most $(1-\goodfrac)^\size \leq \exp(-\goodfrac\size)\leq \frac{1}{n^2}$ for $\size\ge \frac{2\log n}{\goodfrac}$. The lemma follows by applying the union bound over the $\num$ committees.
} 

Lemmas \ref{lem:pub:comm:num} and \ref{lem:pub:comm:maj} directly imply the following.

\begin{theorem}
\label{thm:pub:major:comm}
In the $\Benign(\byzfrac<1/2)$ model, for $\size > \frac{9\log n}{(1/2-\byzfrac)^2}$ and $\num \in O(n)$,
Procedure $\ElectPublic$ w.h.p. results in the election of $\num$ public majorizing committees each of size $\size$ with $\Query = O((1/2-\byzfrac)^{-4} \log^2 n \max(\log \num, \log k))$, $\Time = O(1)$ and $\Message = 0$. Further, each machine belongs to $O(\frac{\size\num}{k}+\size^2)$ committees.
\end{theorem}


Lemmas \ref{lem:pub:comm:num} and \ref{lem:pub:comm:weak} directly imply the following.

\begin{theorem}
\label{thm:pub:rep:comm}
In the $\Benign(\byzfrac<1)$ model, for $\size > \frac{9\log n}{\goodfrac}$ and $\num \in O(n)$,
Procedure $\ElectPublic$ w.h.p. results in the election of $\num$ weak committees each of size $\size$ with $\Query = O(\goodfrac^{-2} \log^2 n \max(\log \num, \log k))$, $\Time = O(1)$, and $\Message = 0$. Further, each machine belongs to $O(\frac{\size\num}{k}+\size^2)$ committees.
\end{theorem}


\subsubsection{The Weak Resolution Problem in $\Benign(\byzfrac<1)$}

In this problem, it is given that a weak committee $\cC$ knows the bits $x_i$ present in the cloud for indices $i\in\Known\subseteq[1,n]$. The problem requires a machine $\Machine$ to learn these bits from $\cC$. (Both $\cC$ and $\Known$ are known to $\Machine$.)

\begin{algorithm}
\begin{algorithmic}[1]
\Statex \textbf{Input}: Weak committee $\cC$, set of indices $\Known$, $\cC$ knows bits $x_j$ for every $j\in \Known$.
\Statex \textbf{Output}: array $\res$ such that $\res[j] = x_{j}$ for every $j \in \Known$
\For{every machine $M'\in\cC$}
\State $M'$ sends $\res_{M'}[j]$ to $M$ for every index $j\in\Known$
\EndFor
\label{weak:resolve:message}
\State $\BYZ \gets \emptyset$
\Comment{Machines blacklisted by $M$}
\For{every $j$ in $\Known$ {\bf sequentially}}
\If{machines in $\cC\setminus B$ sent different bits for index $j$}
\State $\res[j]\gets \CloudQuery(j)$
\State Add machines that sent $1-\res[j]$ to $\BYZ$. \Comment{Those that sent $1-\res[j]$}
\Else
~~ $\res[j]\gets$ bit received from $\cC\setminus B$
\EndIf
\EndFor
\end{algorithmic}
\caption{Procedure $\WeakResolve(\cC,\Known)$, $\Benign(\byzfrac<1)$, code for machine $\Machine$}
\label{alg:weak:resolve}
\end{algorithm}

\begin{theorem}
\label{thm:weak:resolve}
In the $\Benign(\byzfrac<1)$ model, Procedure $\WeakResolve$
w.h.p. solves the Weak Resolution problem with 
$\Query = O(|\cC|)$, $\Time = O(|\Known|)$ and $\Message = O(|\Known|\cdot|\cC|)$. 
\end{theorem}

\begin{proof}
The algorithm ensures that whenever $\Machine$ queries the cloud, at least one machine in $\cC$ is blacklisted, hence $\Query=O(|\cC|)$. The round complexity is $O(|\Known|)$ since $\Machine$ needs to receive the bits corresponding to the indices in $\Known$ from $\cC$ (line \ref{weak:resolve:message}), and it takes another $O(|\Known|)$ rounds to query the cloud sequentially and blacklist Byzantine machines. 
The communication performed in line \ref{weak:resolve:message} is also responsible for the message complexity, $\Message=O(|\Known|\cdot|\cC|)$. Correctness follows from the fact that for every index $j \in \Known$, $\Machine$ either learns $x_j$ directly from the cloud,
or all the non-blacklisted machines remaining in $\cC$ have sent the same value $b$ for it. In the latter case, $b$ must be correct since $\cC$ is 1-representative, implying that least one honest machine in $\cC$ sent $b$.
\end{proof}

\textbf{Fast Weak Resolution sensitive to $\Density$ in the $\Benign(\byzfrac<1)$ model}

We now discuss a slightly modified algorithm for the weak resolution problem, called Procedure $\FastWeakResolve$, that has improved round complexity compared to that of Procedure $\WeakResolve$
when an overwhelming majority of the bits at indices in $\Known$ are zeros (or alternatively, an overwhelming majority are ones). 
In this procedure, machine $\Machine$ asks the committee $\cC$ for the next smallest index $j \in\Known$ such that $x_j=b$, alternating between $b=0$ and $b=1$. This continues until there are no further indices satisfying this condition for some value of $b$. At this point, $\Machine$ terminates its execution.

\begin{theorem}
\label{thm:fast:weak:resolve}
In the $\Benign(\byzfrac<1)$ model, Procedure $\FastWeakResolve$
w.h.p. solves the weak resolution problem with $\Query = O(|\cC|)$, $\Time = O(|\cC|+\min(\Density_\Known,1-\Density_\Known)|\Known|)$ and message complexity $\Message = O\left((1+\min(\Density_\Known,1-\Density_\Known)|\Known|)|\cC|\right)$ where $\Density_k$ is the fraction of ones at indices in $\Known$
\end{theorem}

\begin{proof}
The algorithm guarantees that at least one machine is blacklisted whenever a query is made to the cloud. Hence, $\Query=O(|\cC|)$. Let us consider the case when $\Density_{\Known}<1-\Density_{\Known}$. When $i=\Density_\Known|\Known|+1$ (in line \ref{fast:weak:resolve:i}) and $b=1$ (in line \ref{fast:weak:resolve:b}), $\cC\setminus\BYZ$ sends the $(\Density_{\Known}|\Known|+1)^{th}$ smallest index $j\in\Known$ where $x_j=1$. Since such an index does not exist, the honest machines will send $n+1$ as per the algorithm. Thus $\Machine$ will eventually execute line \ref{fast:weak:resolve:ret}. At this point $\Machine$ has knows all the indices $j\in\Known$ that satisfy $x_j=1$. Therefore the remaining indices $j\in \Known$ satisfy $x_j=0$ and $\Machine$ terminates its execution. We know that $\Machine$ spends $O(|\cC|)$ rounds to blacklist machines and query the cloud. Therefore $\Time=O(\Density_{\Known}|\Known|+|\cC|)$. A similar argument holds when $\Density_{\Known}\ge1-\Density_{\Known}$. Hence $\Time=O(|\cC|+\min(\Density_\Known,1-\Density_\Known)|\Known|)$. Similarly $\Message=O((1+\min(\Density_\Known,1-\Density_\Known)|\Known|)|\cC|)$.
\end{proof}

\begin{algorithm}
\begin{algorithmic}[1]

\Statex \textbf{Input}: Weak committee $\cC$,  set of indices $\Known$. $\cC$ knows bits $x_j$ for every $j\in \Known$.
\Statex \textbf{Output}: Array $\res$ such that $\res[j]=x_j$ for every $j\in \Known$.
\Statex \textbf{Notation}: We use $\Known_{\bitval,i}$ where $\bitval \in \{0,1\}$ to denote the $i^{th}$ smallest index $j$ in $\Known$ 
such that $x_j=\bitval$. If $j$ doesn't exist, then  $\Known_{\bitval,i}=n+1$.
\State $\res[j]\gets\nul$ for every $j\in \Known$
\State $\BYZ\gets\emptyset$
\For{$i=1,2....n$ sequentially}\label{fast:weak:resolve:i}
\For{$b=0,1$ sequentially}\label{fast:weak:resolve:b}
\For{every machine $M'\in\cC\setminus\BYZ$}
\State $M'$ sends $\Known_{b,i}$ to $M$.
\EndFor
\State $\INDEX \gets$ set of distinct indices received
\For{every $j\in \INDEX$ such that $\res[j]=\nul$ sequentially in sorted order}
\If{$j=n+1$}
\State $\res[j]=1-b$ for every $j$ such that $j\in \Known$ and $\res[j]=\nul$\label{fast:weak:resolve:ret}
\State\Return
\EndIf
\If{$j$ is largest index in $\INDEX$}~~
$\res[j]\gets b$\Comment{No need to verify}
\Else~~
$\res[j]\gets \CloudQuery(j)$
\EndIf
\If{$\res[j]=b$}
\State Add machines that sent indices different from $j$ to $\BYZ$
\State $\INDEX\gets\emptyset$
\Else ~~
Add machines that sent $j$ to $\BYZ$
\EndIf
\EndFor
\EndFor
\EndFor
\end{algorithmic}
\caption{Procedure $\FastWeakResolve(\cC,\Known)$ , $\Benign(\byzfrac<1)$, code for machine $M$}
\label{alg:weak:resolve-density-based}
\end{algorithm}

\subsubsection{Input Convergecasting in $\Benign(\byzfrac<1)$}

In this section, 
we introduce the \emph{input convergecasting} problem and provide a procedure called $\InputConvergecast$ to solve it. This will be used later for solving the $\IDp$ and $\XORp$ problems with improved round complexities.

We make use of the following structures. 
A $d$-ary tree in which every level $i\ge 0$, except possibly the deepest, is entirely filled (i.e., has $d^i$ nodes) is called a \emph{complete tree}.
A \emph{public convergecast tree} $T_{(n,k)}$ is an $(\ceil{\frac{n}{k}+1})$-ary complete rooted tree with exactly $\ceil{\frac{k^2}{n}}$ leaves such that each node in the tree represents a weak committee.
A \emph{public convergecast forest} $\cF_{(n,k)}$ is a forest containing $\min(\ceil{\frac{n}{k}},k)$ public convergecast trees $T_{(n,k)}$.
The following lemma summarizes some straightforward properties of 
public convergecast trees and forests.
\begin{lemma}
\label{thm:Benign:CT}
In a public convergecast tree $T_{(n,k)}$,
the degree of every non leaf node is at least 2,
the number of nodes is $\Theta(1+\frac{k^2}{n})$,
the height is $O(\log k)$,
and when $k^2\le n$, the tree has exactly one node which is also the root.
%
%
In a public convergecast forest $\cF_{(n,k)}$,
the number of roots is $\min(\ceil{\frac{n}{k}},k)$,
the number of leaves is $\Theta(k)$,
and the total number of nodes is $\Theta(k)$.
\end{lemma}

Consider a public convergecast forest $\cF=\cF_{(n,k)}$ and number its nodes from 1 to $\num$ in some predefined order. Denote the $i^{th}$ node of $\cF$ by $\cF_i$. All we need to construct $\cF$ is a set of $\num$ weak committees. Hence, we represent $\cF$ as $\cF=\{\cC_1,\cC_2,\ldots,\cC_{\num}\}$, where $\cC_i$ is a weak committee and $\cF_i=\cC_i$ for $i=1,2,\ldots,\num$. Also, we use $\cF_i^p$ to denote the parent node of $\cF_i$ and $\{\cF_i^c\}$ to denote the set of children of $\cF_i$. 
The level of a node in $\cF$ is equal to the height of its subtree. Thus all leaves of $\cF$ are at level 1. The height of $\cF$ denoted by $\cHt$ is the maximum height of a tree in $\cF$.

Theorem \ref{thm:pub:rep:comm} directly implies the following.

\begin{theorem}
In the $\Benign(\byzfrac<1)$ model, Procedure $\ConvergecastForest$
w.h.p constructs a public convergecast forest with the following complexities: (i) $\Query = O(\goodfrac^{-2} \log^2 n \log k)$, (ii) $\Time = O(1)$, (iii) $\Message = 0$. Also, every machine belongs to $O(\frac{\log ^2n}{\goodfrac^2})$ nodes.
\end{theorem}

\begin{algorithm}
\begin{algorithmic}[1]
\State $\num\gets$ number of nodes in an $(n,k)$ public convergecast forest\Comment{$\num\in\Theta(k)$}
\State $\size\gets\ceil{\frac{9\log n}{\goodfrac}}$
\State $\cC_1,\cC_2,\ldots,\cC_{\num}\gets\ElectPublic(\size,\num)$
\State \Return $\cF=\{\cC_1,\cC_2,\ldots,\cC_{\num}\}$
\end{algorithmic}
\caption{Procedure $\ConvergecastForest$, $\Benign(\byzfrac<1)$, code for all machines}
\label{alg:Benign:CF}
\end{algorithm}


\begin{defn}(Input Convergecasting Problem) 
Every bit is known by some root of a public convergecast forest $\cF_{(n,k)}$.
\end{defn}
We now describe Procedure $\InputConvergecast$ that solves the Input Convergecasting problem. We first construct a convergecast forest $\cF=\cF_{(n,k)}$. Partition the indices in $[1,n]$ among the $\Theta(k)$ leaves of $\cF$ such that each leaf is assigned $\Theta(\frac{n}{k})$ indices. Every nonleaf node is assigned the union of the indices assigned to its children. Every leaf learns the bits at indices assigned to it by directly querying the cloud. Every nonleaf node learns the bits at indices assigned to it by invoking Procedure $\WeakResolve$ on its children.
The procedure terminates when all roots have learned the bits at indices assigned to it. 

\begin{algorithm}
\begin{algorithmic}[1]
\State $\cF\gets \ConvergecastForest$
\State Partition the indices in $[1,n]$ into disjoint parts such that every machine in leaf $v$ in $\cF$ is assigned the set $\INDEX_v$ such that $|\INDEX_v|=O(\frac{n}{k})$ 
\State Every machine in nonleaf node $v$ in $\cF$ is assigned the set $\INDEX_v$ containing the union of indices assigned to the leaves in its subtree.
\For{every leaf $v$ in $\cF$ in parallel }
\For{every machine $M\in v$}
\State $\res_M[j]\gets\CloudQuery(j)$ for every $j\in\INDEX_v$
\label{line:IC:Benign:leaf}
\EndFor
\EndFor
\For{$i=2,3,\ldots,\cHt$ sequentially}
\For{every nonleaf node $v$ in $\cF$ at level $i$ in parallel}
\For{every child $\cC$ of $v$ in parallel}
\State Every machine in $v$ invokes Procedure $\WeakResolve(\cC,\INDEX_{\cC})$
\label{line:IC:Benign:weak}
\EndFor
\EndFor
\EndFor
\end{algorithmic}
\caption{Procedure $\InputConvergecast$, $\Benign(\byzfrac<1)$, code for all machines}
\label{alg:IC:benign}
\end{algorithm}

\begin{theorem}
In the $\Benign(\byzfrac<1)$ model, 
Procedure $\InputConvergecast$ w.h.p.
solves the Input Convergecasting problem 
with the following complexities.
\\
- When $k^2>n$, (i) $\Query = O(\frac{n\log^3n}{\goodfrac^3 k})$, (ii) $\Time = O((\frac{n}{k}+\frac{k^2}{n})\frac{\log^2n}{\goodfrac^2})$, (iii) $\Message = O(\frac{n\log k\log^2 n}{\goodfrac^2})$.
\\
- When $k^2\le n$, (i) $\Query = O(\frac{n\log^3n}{\goodfrac^3 k})$, (ii) $\Time = O(1)$, (iii) $\Message =0$.
\label{thm:IC:benign}
\end{theorem}
\begin{proof}
A machine belongs to 
$O((\log ^2n)/\goodfrac^2)$ 
leaves, so the number of queries made by a machine in line \ref{line:IC:Benign:leaf}  is 
$O((n\log^2n)/(\goodfrac^2 k))$. 
Further, a machine belongs to 
$O((\log ^2n)/\goodfrac^2)$ 
nonleaf nodes. Hence, a machine has $O((n\log^2n)/(\goodfrac^2 k))$
children. By Thm. \ref{thm:weak:resolve}, the total number of queries spent over all invocations of Procedure $\WeakResolve$ in line \ref{line:IC:Benign:weak} is 
$\Query=O((n\log^3n)/\goodfrac^3)$. 

When $k^2\le n$, all nodes are leaves, so every machine terminates after directly querying the cloud for bits at indices assigned to it, hence $\Time=O(1)$ and $\Message=0$.
When $k^2>n$, there are 
$\ceil{n/k}$ 
roots and each root is assigned $O(\frac{n}{n/k})=O(k)$ indices. If the root is filled, then each child of the root would be assigned $O(\frac{k}{n/k})=O(\frac{k^2}{n})$ indices. If the root is not filled, then each child of the root is also a leaf. In such a case, the child is assigned $O(\frac{n}{k})$ indices. Therefore a child of a root would be assigned $O(\frac{k^2}{n}+\frac{n}{k})$ indices. Also, a machine in a child of a root could belong to $O(\frac{\log^2 n}{\goodfrac^2})$ nodes. Therefore a root would spend $O((\frac{n}{k}+\frac{k^2}{n})\frac{\log^2n}{\goodfrac^2})$ rounds during Procedure $\WeakResolve$. This is also the slowest phase of the procedure and therefore $\Time=O((\frac{n}{k}+\frac{k^2}{n})\frac{\log^2n}{\goodfrac^2})$.
Every bit is sent by $O(\log k)$ committees to its parent committee. Therefore $\Message=O(\frac{n\log k\log^2 n}{\goodfrac^2})$
\end{proof}

\CommentedStart
\textbf{Faster convergecasting.}
Procedure $\WeakResolve$, used in Procedure $\InputConvergecast$, could be replaced with Procedure $\FastWeakResolve$ to improve the round and message complexities. However, when the density of ones and zeros is not uniform throughout the input,
this may not provide a speedup by a factor of $\gamin$. This is because the range $[1,n]$ is partitioned among the leaves of $\cF$ in a predefined manner, so the density of ones and zeros directly queried by a machine in a leaf in line \ref{line:IC:Benign:leaf} could be worse than the true density of ones and zeros in the input. To fix this, we need a technique by which the machines in the leaves can query random indices. We can use the cloud's \emph{random bit generator service} to generate a new array $\Random$ of size $O(n\log n)$ whose elements are chosen uniformly at random from $[1,n]$. This guarantees that every index $i\in[1,n]$
will appear in $\Random$ w.h.p. We now partition the array $\Random$ among the leaves in $\cF$. Every machine applies two rounds of queries to learn a bit. The first involves $\ceil{\log n}$ queries to obtain $\ceil{\log n}$ random bits in order to generate a random index $i\in[1,n]$. The second round consists of a query to learn $x_i$. This technique, along with Procedure $\FastWeakResolve$, allows us to improve the round and message complexities for input convergecasting by a factor of $O(\gamin\log n)$.
\CommentedEnd

\subsection{The $\IDp$ Problem in the Benign Model}

\subsubsection{Randomized algorithm in $\Benign(\byzfrac<1)$}

The approach taken in Sect. \ref{sec:idist:sol1} for performing $\IDp$ in the $\Harsh(\byzfrac<1)$ model fell short of achieving (near) optimal query complexity. Subsequently, in this section we describe an approach that achieves optimal query complexity up to logarithmic factors (whp) in the $\Benign$ model.

Algorithm $\LinearDownload$, presented next, has large time and message complexities, and later on we present an improved algorithm called $\FastLinearDownload$. Nevertheless, we find it instructive to present this algorithm first, as it illustrates the key idea behind our approach in a simpler setting, hence serving as a useful introduction to the more complex Algorithm $\FastLinearDownload$.
The idea is as follows. Construct $k$ weak committees $\cC_1,\cC_2...\cC_k$. The indices in $[1,n]$
is partitioned into $k$ disjoint parts of about $n/k$ bits such that the $i^{th}$ committee $\cC_i$ has the $i^{th}$ part. Now, the information is spread from $\cC_1$ to $\cC_k$ in a linear fashion. $\cC_i$ learns the bits in first $i-1$ parts from $\cC_{i-1}$ by invoking Procedure $\WeakResolve$, and then learns the bits in the $i^{th}$ part by directly querying the cloud. This way, $\cC_k$ learns the bits in all $k$ parts, and subsequently every machine learns all the bits from $\cC_k$ by invoking Procedure $\WeakResolve$.

\begin{algorithm}
\label{alg:ID:Benign:Linear:Slow}
\begin{algorithmic}[1]
    \State $\size\gets\ceil{\frac{9\log n}{\goodfrac}}$
    \State $\cC_1,\cC_2,\ldots,\cC_{k}\gets\ElectPublic(\size,k)$
    \State Partition indices in $\INDEX=[1,n]$ into $k$ disjoint subsets $\INDEX_1,\INDEX_2,\ldots,\INDEX_k$ each of about $n/k$ bits.
    \For{$i=1,2,\ldots, k$ in parallel}
    \label{line:ID:Benign:linear:slow:cloud}
    \For{every machine $M\in\cC_i$}
    \State $\res_M[j]\gets\CloudQuery(j)$ for every $j\in\INDEX_i$ 
    \EndFor
    \EndFor
    \For{$i=2,3,\ldots, k$ sequentially}
    \label{line:ID:Benign:linear:slow:weak}
    \State Every machine in $\cC_i$ invokes Procedure $\WeakResolve(\cC_{i-1},\bigcup_{j=1}^{i-1}\INDEX_j) $
    \EndFor
    \State Every machine invokes Procedure $\WeakResolve(\cC_{k},\INDEX)$
    \label{line:ID:Benign:linear:slow:last}
    \end{algorithmic}
    \caption{Algorithm $\LinearDownload$ , $\Benign(\byzfrac<1)$, code for all machines}
\end{algorithm}

\begin{theorem}
\label{thm:ID:Benign:Linear:Slow}
In the $\Benign(\byzfrac<1)$ model, Algorithm 
$\LinearDownload$ w.h.p.
solves the $\IDp$ problem with 
$\Query=O\left(\frac{n\log^2 n}{\goodfrac^2 k}+\frac{\log^3 n}{\goodfrac^3}\right)$, $\Time=O(nk)$ and $\Message=O\left(\frac{nk\log^2 n}{\goodfrac^2}\right)$.
\end{theorem}


\begin{proof}
By Thm. $\ref{thm:pub:rep:comm}$, a machine belongs to 
$O((\log^2 n)/\goodfrac^2)$
committees. Therefore the number of queries spent in line \ref{line:ID:Benign:linear:slow:cloud} is 
$O((n\log^2 n)/(\goodfrac^2 k))$.
By Thm. $\ref{thm:weak:resolve}$, the number of queries spent in line \ref{line:ID:Benign:linear:slow:weak} is 
$O((\log^3 n)/\goodfrac^3)$ and 
$O((\log n)/\goodfrac)$ 
in line \ref{line:ID:Benign:linear:slow:last}. Hence, $\Query=O\left(\frac{n\log^2 n}{\goodfrac^2 k}+\frac{\log^3 n}{\goodfrac^3}\right)$.
Since the $k$ committees learn bits sequentially in line  \ref{line:ID:Benign:linear:slow:weak}, $\Time=O(nk)$.
Since each committee is of size $O(\frac{\log n}{\goodfrac})$, from line \ref{line:ID:Benign:linear:slow:weak}, $\Message=O(\frac{nk\log^2 n}{\goodfrac^2})$ messages. 
\end{proof}

Again, the round and message complexities of the algorithm can be improved by replacing Procedure $\WeakResolve$ with $\FastWeakResolve$.
Theorems \ref{thm:ID:Benign:Linear:Slow} and \ref{thm:fast:weak:resolve} yield the following.

\begin{algorithm}
\label{alg:ID:Benign:linear:fast}
\begin{algorithmic}[1]
    \State $\size\gets\ceil{\frac{9\log n}{\goodfrac}}$
    \State $\cC_1,\cC_2,\ldots,\cC_{k}\gets\ElectPublic(\size,k)$
    \State Partition indices in $\INDEX=[1,n]$ into $k$ disjoint subsets $\INDEX_1,\INDEX_2,\ldots,\INDEX_k$ each of about $n/k$ bits.
    \For{$i=1,2,\ldots, k$ in parallel}
    \For{every machine $M\in\cC_i$}
    \State $\res_M[j]\gets\CloudQuery(j)$ for every $j\in\INDEX_i$ 
    \EndFor
    \EndFor
    \For{$i=2,3,\ldots, k$ sequentially}
    \State Every machine in $\cC_{i}$ invokes Procedure $\FastWeakResolve(\cC_{i-1},\bigcup_{j=1}^{i-1}\INDEX_j) $
    \EndFor
    \State Every machine invokes Procedure $\FastWeakResolve(\cC_{k},\INDEX)$
    \end{algorithmic}
    \caption{Algorithm $\FastLinearDownload$ , $\Benign(\byzfrac<1)$, code for all machines}
\end{algorithm}

\begin{theorem}
\label{thm: fast linear download}
In the $\Benign(\byzfrac<1)$ model, Algorithm
$\FastLinearDownload$ w.h.p.
solves the $\IDp$ problem with
$\Query=O\left(\frac{n\log^2 n}{\goodfrac^2 k}+\frac{\log^3 n}{\goodfrac^3}\right)$, 
$\Time=O(\gamin nk+\frac{n\log n}{\goodfrac})$ and $\Message=O\left(\frac{(\gamin k+1)n\log^2 n}{\goodfrac^2}\right)$.
\end{theorem}

We now describe a parallel algorithm for $\IDp$ that makes use of Procedure $\InputConvergecast$. At the end of Procedure $\InputConvergecast$, every bit is present in some root of the convergecast forest $\cF$. Every machine can learn all bits by invoking Procedure $\WeakResolve$ on each of the $\min(\frac{n}{k},k)$ roots of $\cF$ with the respective set of indices assigned to it.

\begin{algorithm}
\begin{algorithmic}[1]
\State Invoke Procedure $\InputConvergecast$.
\State Let $\cF$ denote the convergecast forest.
\For{every root $v$ in $\cF$ in parallel}
\State Let $\INDEX_v$ denote the set of indices assigned to $v$ during Procedure $\InputConvergecast$
\State Every machine invokes Procedure $\WeakResolve(v,\INDEX_v)$
\label{line:ID:Benign:Parallel:Weak}
\EndFor 
\State\Return
\end{algorithmic}
\caption{$\DownloadParallel$ , $\Benign(\byzfrac<1)$, code for all machines}
\label{alg:ID:Benign:Parallel}
\end{algorithm}

\begin{theorem}
\label{thm: download-parallel}
In the $\Benign(\byzfrac<1)$ model, Algorithm
$\DownloadParallel$ w.h.p.
solves the $\IDp$ problem with $\Query = O(\frac{n\log^3n}{\goodfrac^3k})$, $\Time = O((\frac{n}{k}+k)\frac{\log^2 n}{\goodfrac^2})$ and $\Message= O(\frac{nk\log n}{\goodfrac}+\frac{n\log k\log^2 n}{\goodfrac^2})$.
\end{theorem}

\begin{proof}
By Thm. \ref{thm:IC:benign}, Procedure $\InputConvergecast$ takes $O(\frac{n\log^3 n}{\goodfrac^3 k})$ queries. In line \ref{line:ID:Benign:Parallel:Weak}, every machine invokes Procedure $\WeakResolve$ for $O(\frac{n}{k})$ times and spends $O(\frac{n\log n}{k\goodfrac})$ queries. Hence, $\Query=O(\frac{n\log^3 n}{\goodfrac^3 k})$. By Thm. \ref{thm:IC:benign}, Procedure $\InputConvergecast$ takes $O((\frac{n}{k}+k)\frac{\log^2 n}{\goodfrac^2})$ rounds. Since $\cF$ has  $\min(\ceil{\frac{n}{k}},k)$ roots, each root knows $O(\frac{n}{k}+k)$ bits. Also, a machine belongs to $O(\frac{\log^2 n}{\goodfrac^2})$ roots. Hence, by Thm. \ref{thm:weak:resolve}, executing line \ref{line:ID:Benign:Parallel:Weak} takes another $O((\frac{n}{k}+k)\frac{\log^2 n}{\goodfrac^2})$ rounds, so $\Time=O((\frac{n}{k}+k)\frac{\log^2 n}{\goodfrac^2})$. By Thm. \ref{thm:IC:benign}, Procedure $\InputConvergecast$ takes  $O(\frac{n\log k\log^2 n}{\goodfrac^2})$ messages. In line \ref{line:ID:Benign:Parallel:Weak}, each bit is sent by a root of size $O(\frac{\log n}{\goodfrac})$ to every machine, taking $O(\frac{nk\log n}{\goodfrac})$ messages. Hence, $\Message=O(\frac{nk\log n}{\goodfrac}+\frac{n\log k\log^2 n}{\goodfrac^2})$. 
\end{proof}

\subsubsection{Randomized algorithm in $\Benign(\byzfrac<1/2)$}
\label{sec:InDistr-beta =1/2}

When $\byzfrac < 1/2$, we can use public majorizing committees to obtain better solutions.

\begin{algorithm}
\begin{algorithmic}[1]
\State $\cC_1,\cC_2,\ldots,\cC_k\gets\ElectPublic(\frac{2\log n}{(1/2-\byzfrac)^2},k)$
\State Partition the indices in $\INDEX=[1,n]$ into k disjoint subsets $\INDEX_1,\INDEX_2,\ldots,\INDEX_k$ each of size almost $n/k$.
\For{$i=1,2,\ldots, k$ in parallel}
\label{line:ID:Benign:Major:cloud}
\For{every machine $\Machine\in\cC_i$}
\State $\res_M[j]\gets\CloudQuery(j)$ for every $j\in\INDEX_i$
\EndFor
\EndFor
\For{$i=1,2,\ldots, k$ in parallel}
\label{line:ID:Benign:Major:send}
\For{every machine $\Machine\in\cC_i$}
\State $\Machine$ sends $\langle\res_M[j],j\rangle$ for every $j\in\INDEX_i$ to every machine.
\EndFor
\EndFor
\For{every machine $M$ in parallel}
\For{$i=1,2,\ldots, n$ locally}
\State Let $\cC_j$ be the committee such that $i\in\INDEX_j$
\If{$M$ received more $\langle 1,i\rangle$ messages from $\cC_j$ than $\langle 0,i\rangle$ messages from $\cC_j$}
\State $\res_M[i]\gets1$
\Else
\State $\res_M[i]\gets0$
\EndIf
\EndFor
\EndFor
\end{algorithmic}
\caption{Algorithm $\MajorizingDownload$, $\Benign(\byzfrac<1/2)$, code for all machines}
\label{alg:ID:Benign:majorizing}
\end{algorithm}

\begin{theorem}
\label{thm: Disjunction majorizing Benign(1/2)}
In the $\Benign(\byzfrac<1/2)$ model, Algorithm
$\MajorizingDownload$
solves the $\IDp$ problem with $\Query=O\left(\frac{n \log^2 n}{k(1/2-\byzfrac)^4}  \right), \Time = O\left(\frac{n \log^2 n}{k(1/2-\byzfrac)^4}  \right)$, 
and $\Message = O\left(\frac{nk\log n}{(1/2-\byzfrac)^2}\right)$. 
\end{theorem}

\begin{proof}
Correctness follows since $\cC_1,\ldots, \cC_k$ are all public majorizing committees. 
Consider index $i$, and let $\cC_j$ be the (unique) committee such that $i\in\INDEX_j$. 
All honest machines in $\cC_j$ sent $\langle x_i,i\rangle$ to all machines and since $\cC_j$ has a majority of honest machines,
every machine $M$ will learn $x_i$ correctly.

A  machine belongs to $O(\frac{\log^2 n}{(1/2-\byzfrac)^4})$ committees, so the number of queries spent in the algorithm (in line \ref{line:ID:Benign:Major:cloud}) is 
$\Query=O\left(\frac{n \log^2 n}{k(1/2-\byzfrac)^4}  \right)$.
Machines communicattion (in line \ref{line:ID:Benign:Major:send}) takes 
$\Time = O\left(\frac{n \log^2 n}{k(1/2-\byzfrac)^4}  \right)$.
Each committee of size $O(\frac{\log n}{(1/2-\byzfrac)^2})$ sends $O(\frac{n}{k})$ bits to $k$ machines. Hence $\Message=O\left(\frac{nk\log n}{(1/2-\byzfrac)^2}\right)$. 
\end{proof}

\subsection{The $\XORp$ Problem in the Benign Model}

\subsubsection{Randomized parity algorithm in $\Benign(\byzfrac<1)$}
Before solving the $\XORp$ problem, we define a sub problem called the Weak Parity Resolution problem and provide Procedure $\WeakParityResolve$ to solve it. This is similar to the Weak Resolution problem and the corresponding Procedure $\WeakResolve$ that we described earlier. 

\textbf{Weak Parity Resolution Problem} In this problem, it is given that a weak committee $\cC$ knows bits $x_i$ present in the cloud for indices $i\in\Known\subseteq[1,n]$ . A machine $\Machine$ needs to learn the parity of these bits, i.e., $\oplus_{j\in\Known}x_j$, from $\cC$.

Clearly, Procedure $\WeakResolve$ can solve the Weak Parity Resolution Problem, but the round complexity will be $O(|\Known|)$. We describe Procedure $\WeakParityResolve$ that solves it in $O(1)$ rounds as follows. Let $\Known[i]$ denote the $i^{th}$ index in $\Known$. Let $XOR_\Known(l,r) = x_{\Known[l]}\oplus x_{\Known[l+1]}\ldots\oplus x_{\Known[r]}$ for $1\le l\le r\le |\Known|$. In the Weak Parity Resolution Problem, a machine $\Machine$ needs to learn the value of $XOR_\Known(1,|\Known|)$ from a weak committee $\cC$. If machines in $\cC$ send different values for $XOR_\Known(1,|\Known|)$, $\Machine$ concludes that some machines are Byzantine and should
be blacklisted. $\Machine$ binary searches for the largest $r<|\Known|$ such that machines in $\cC$ send the same $XOR_\Known(1,r)$ value but different $XOR_\Known(1,r+1)$ values. At this point, $\Machine$ just queries the cloud for $x_{r}$ and blacklists the Byzantine machines. Since $\cC$ is a weak committee, this happens at most $|\cC|-1$ times, after which all remaining machines in $\cC$ send the same $XOR_\Known(1,|\Known|)$ value.

\begin{algorithm}
\begin{algorithmic}[1]
\Statex \textbf{Input:} set of indices $\Known\subseteq[1,n]$, weak committee $\cC$ that knows bits $x_i$ for all indices $i\in\Known$ 
\Statex \textbf{Notation:} We use $\Known[i]$ to denote the $i^{th}$ element in $\Known$. $XOR(l,r)=x_{\Known[l]}\oplus x_{\Known[l+1]}\ldots\oplus x_{\Known[r]}$ for $1\le l\le r\le |\Known|$. $XOR_{\Machine'}(l,r)$ denotes the value of $XOR(l,r)$ sent by $\Machine'$ to $\Machine$
\Statex \textbf{Output:} $XOR(1,|\Known|)$
\State $\BYZ\gets \emptyset$
\State Ask every machine in $\cC\setminus\BYZ$ for $XOR(1,|\Known|)$
\While{machines in $\cC\setminus B$ sent different values for $XOR(1,|\Known|)$}
\label{line:WeakXORResolve:outerwhile}
\State $l\gets 1, r\gets |\Known|$\Comment{Binary search based blacklisting}
\While{$\l\neq r$}
\label{line:WeakXORResolve:innerwhile}
\State $m\gets\floor{\frac{l+r}{2}}$
\State Ask every machine in $\cC\setminus\BYZ$ for $XOR(l,m)$
\For{every machine $\Machine'\in\cC\setminus\BYZ$ locally}
\State Deduce that $XOR_{\Machine'}(m+1,r)=XOR_{\Machine'}(l,r)\oplus XOR_{\Machine'}(l,m)$
\EndFor
\If{Every machine in $\cC\setminus\BYZ$ sent the same value for $XOR(l,m)$}
~~$l\gets\ m+1$
\Else ~~$r\gets m$
\EndIf
\EndWhile
\State $\res[\Known[l]]\gets\CloudQuery(\Known[l])$\Comment{$l=r$}
\State Add machines $\Machine'$ for which
$XOR_{\Machine'}(l,l)\neq\res[\Known[l]]$ to $\BYZ$
\label{line:WeakXORResolve:query}
\EndWhile
\State \Return $XOR(1,|\Known|)$ sent by machines in $\cC\setminus B$.
\end{algorithmic}
\caption{Procedure $\WeakXORResolve(\cC,\Known)$, $\Benign(\byzfrac<1)$, code for machine $\Machine$}
\label{alg:Weak:XOR:Resolve}
\end{algorithm}

\begin{theorem}
\label{thm:parity:Benign}
In the $\Benign(\byzfrac<1)$ model, 
Procedure $\WeakXORResolve$ w.h.p.
solves the Weak Parity Resolution problem with $\Query = O(|\cC|)$, $\Time = O(|\cC|\log n)$ and $\Message = O(|\cC|^2\log n)$.
\end{theorem}

\begin{proof}
Clearly, an honest machine in $\cC$ is never blacklisted.
Correctness follows from the fact that the procedure returns only when all machines in $\cC\setminus B$ send the same value for $XOR(L,R)$. 

By line \ref{line:WeakXORResolve:query}, at least one machine is blacklisted after every query to the cloud, so $\Query=O(|\cC|)$. Hence, $\Query=O(|\cC|)$. Every execution of the while loop starting at line \ref{line:WeakXORResolve:innerwhile}
takes $O(\log n)$ rounds. Also, at least one machine is blacklisted every time the loop is executed. Hence, $\Time=O(|\cC|\log n)$, and similarly $\Message=O(|\cC|^2\log n)$.
\end{proof}

\begin{algorithm}
\begin{algorithmic}[1]
\State $\InputConvergecast()$. Let $\cF$ denote the convergecast forest.
\For{every root $v$ in $\cF$ in parallel}
\State Let $\INDEX_v$ denote the set of indices assigned to $v$ during Procedure $\InputConvergecast$
\State Every machine invokes Procedure $\WeakXORResolve(v,\INDEX_v)$
\label{line:XOR:Benign:Parallel:Weak}
\EndFor 
\State\Return
\end{algorithmic}
\caption{Algorithm $\ConvergeParity$, $\Benign(\byzfrac<1)$, code for all machines}
\label{alg:XOR:Benign:Parallel}
\end{algorithm}

\begin{theorem}\label{thm: benign parity converge}
In the $\Benign(\byzfrac<1)$ model, Algorithm 
$\ConvergeParity$ w.h.p.
solves the $\XORp$  problem with $\Query = O(\frac{n\log^3n}{\goodfrac^3k})$, 
$\Time = O((\frac{n}{k}+\frac{k^2}{n}+\frac{\log^2 n}{\goodfrac})\frac{\log^2 n}{\goodfrac^2})$ when $k^2>n$ and $\Time = O(\frac{\log ^4 n}{\goodfrac^3})$ when $k^2\le n$ and 
$\Message = O(\frac{n\log^3 n}{\goodfrac^2})$.
\end{theorem}
\begin{proof}
By Thm. \ref{thm:IC:benign}, Procedure $\InputConvergecast$ takes $\Time=O(\frac{n\log^3 n}{\goodfrac^3})$ rounds. Since there are $O(\frac{n}{k})$ roots in $\cF$, it takes $\Query=O(\frac{n\log n}{k\goodfrac})$ queries per machine over all invocations of Procedure $\WeakXORResolve$. Hence, $\Query=O(\frac{n\log^3n}{\goodfrac^3k})$. Procedure $\InputConvergecast$ takes $O((\frac{n}{k}+\frac{k^2}{n})\frac{\log^2 n}{\goodfrac^2})$ rounds when $k^2>n$ and $O(1)$ rounds when $k^2\le n$. By Thm. \ref{thm:parity:Benign}, line \ref{line:XOR:Benign:Parallel:Weak} can take $O(\frac{\log^4 n}{\goodfrac^3})$ rounds since a machine can belong to $O(\frac{\log^2 n}{\goodfrac^2})$ roots in $\cF$. Each of the $O(\frac{n}{k})$ roots of size $O(\frac{\log n}{\goodfrac})$ in $\cF$ send $O(\frac{\log^2 n}{\goodfrac})$ bits to all $k$ machines. Hence $\Message=O(\frac{n\log^3 n}{\goodfrac^2})$.
\end{proof}

\subsubsection{Randomized parity algorithm in $\Benign(\byzfrac<1/2)$}

In this section we derive a more efficient algorithm for $\XORp$ in the $\Benign(\byzfrac<1/2)$ model, capitalizing on the availability of small majorizing public committees.

\begin{algorithm}
\begin{algorithmic}[1]
\State $\cC_1,\cC_2,\ldots,\cC_k\gets\ElectPublic(\frac{2\log n}{(1/2-\byzfrac)^2},k)$
\State Partition the indices in $\INDEX=[1,n]$ into k disjoint subsets $\INDEX_1,\INDEX_2,\ldots,\INDEX_k$ each of size almost $n/k$.
\For{$i=1,2,\ldots, k$ in parallel}
\label{line:Parity:Benign:Major:cloud}
\For{every machine $\Machine\in\cC_i$}
\State $\res_M[j]\gets\CloudQuery(j)$ for every $j\in\INDEX_i$
\EndFor
\EndFor
\For{$i=1,2,\ldots, k$ in parallel}
\label{line:Parity:Benign:Major:send}
\For{every machine $\Machine\in\cC_i$}
\State $\semiparity_M\gets0$
\EndFor
\For{every index $j\in \INDEX_i$ locally}
\For{every machine $\Machine\in\cC_i$}
\State $\semiparity_M\gets \semiparity_M\oplus\res_M[j]$
\EndFor
\EndFor
\For{every machine $\Machine\in\cC_i$}
\State $\Machine$ sends $<\semiparity_M,i>$ to every machine.
\EndFor
\EndFor
\For{every machine $M$ in parallel}
\State $\parity_M\gets0$
\For{$i=1,2,\ldots, k$ locally}
\If{$M$ received more $\langle 1,i\rangle$ messages from machines in $\cC_i$ than $\langle 0,i\rangle$ messages from machines in $\cC_i$}
~~$\parity_M\gets1-\parity_M$
\EndIf
\EndFor
\State\Return $\parity_M$
\EndFor
\end{algorithmic}
\caption{Algorithm $\MajorizingParity$, $\Benign(\byzfrac<1/2)$, code for all machines}
\label{alg:Parity:Benign:majorizing}
\end{algorithm}

\begin{theorem}
In the $\Benign(\byzfrac<1/2)$ model, 
Algorithm $\MajorizingParity$ w.h.p.
solves the $\XORp$ problem with 
$\Query=O\left(\frac{n \log^2n}{k(1/2-\byzfrac)^4}  \right), \Time = O(\frac{\log^2 n}{(1/2-\byzfrac)^4})$, 
and $\Message = O\left(\frac{k^2\log n}{(1/2-\byzfrac)^2}\right)$. 
\end{theorem}

\begin{proof}
A machine belongs to $O(\log^2_\byzfrac n)$ committees, so $\Query=O((n \log^2_{\byzfrac} n)/k)$. A machine sends one bit to every machine for each committee it belongs to, so $\Time=O(\frac{\log^2 n}{(1/2-\byzfrac)^4})$. Each of the $k$ committees of size $O(\frac{\log n}{(1/2-\byzfrac)^2})$ sends one bit to all machines, so $\Message=O(\frac{k^2\log n}{(1/2-\byzfrac)^2})$
\end{proof}

\subsection{Lower Bounds for $\ORp$}


\begin{theorem}
\label{thm:LB-ORp-2}
In $\Benign(\byzfrac < 1)$ model, any randomized algorithm for $\ORp(\Density)$  that succeeds with constant probability has $\Query=\Omega(\frac{1}{\goodfrac k}\cdot\InverseDensity)$ in expectation. 
\end{theorem}

\begin{proof}  

Consider a randomized algorithm $\mathcal{R}$ that solves $\ORp(\Density)$ with probability at least 1/2 when $\Density$ is known to the algorithm. Recall that $\mathcal{R}$ must be able to distinguish between inputs with at least $\Density n$ 1's from the input with all 0's. Suppose (for the sake of contradiction)  that $\mathbb{E}[\Query]<\frac{1}{2\goodfrac k} \cdot  \InverseDensity$ for $\mathcal{R}$ on all inputs. Let $\TotQuery$ denote the total number of queries performed by honest machines executing $\mathcal{R}$. Clearly, $\mathbb{E}[\TotQuery] < (1/2) \InverseDensity$.

Consider the input $Z$ where all bits in $\INPUT$ are 0's. Let $p_i$ denote the probability that the $i^{th}$ bit $x_i$ is queried by some honest machine. Clearly $\sum_{i=1}^{n}p_i \le \mathbb{E}[\Query]<1/(2\ModifiedDensity)$. This implies that there is a set of indices $\INDEX$ of size $\Density n$ such that $\sum_{i\in\INDEX} p_i<\ModifiedDensity n/(2\ModifiedDensity n)<1/2$. 

Now consider the input $\bar{Z}$ such  that $x_i=1$ for $i\in\INDEX$ and 0 everywhere else.  The probability that $\mathcal{R}$ succeeds in distinguishing between $Z$ and $\bar{Z}$ is
$<1/2$, which is a contradiction. Therefore $\mathbb{E}[\TotQuery]\ge\frac{1}{2 \ModifiedDensity}=\Omega(\InverseDensity)$. Considering $\goodfrac k$ honest machines, we get $\mathbb{E}[\Query]=\Omega(\frac{1}{\goodfrac k}\cdot\frac{1}{\ModifiedDensity})$.
\end{proof}

\begin{corollary}
In the $\Benign(\byzfrac<1)$ model, any randomized algorithm for $\IDp$ that succeeds with constant probability has $\Query=\Omega\left(\frac{1}{\goodfrac k}\cdot\InverseDensity\right)$ in expectation. 
\end{corollary}

\begin{proof}
The claim holds since an algorithm for $\IDp$ implies an algorithm for $\ORp$ with the same complexity.
\end{proof}

\begin{corollary}
In the $\Benign(\byzfrac<1)$ model, any randomized algorithm for $\XORp$ that succeeds with constant probability  has 
$\Query=\Omega\left(\frac{1}{\goodfrac k}\cdot\InverseDensity\right)$ in expectation.
\end{corollary}
\begin{proof}
The proof follows by adapting the proof of Thm.~\ref{thm:LB-ORp-2} with $\Density = 1/n$.
\end{proof}

\clearpage

\centerline{\Large\bf Appendix: Some Deferred Proofs}


\inline Proof of Lemma \ref{lem:large-set-expander}:
\label{proof:large-set-expander}
\PROOFA
\hfill $\Box$

\inline Proof of Lemma \ref{lem:glse}:
\label{proof:glse}
\PROOFGLSE
\hfill $\Box$

\inline Proof of Lemma \ref{lem:coins}:
\label{proof:lem:coins}
\PROOFB
\hfill $\Box$

\inline Proof of Lemma \ref{lem:pub:comm:num}:
\ProofFourThree
\hfill $\Box$

\inline Proof of Lemma \ref{lem:pub:comm:maj}:
\ProofFourFour
\hfill $\Box$

\inline Proof of Lemma \ref{lem:pub:comm:weak}:
\ProofFourFive
\hfill $\Box$

\bigskip
\bibliographystyle{plainurl}
\bibliography{main}

\end{document}